\newif\ifarxiv
\arxivfalse 
\arxivtrue 
\ifarxiv 
\documentclass[letterpaper]{article}
\usepackage[margin=1in]{geometry}
\usepackage{natbib} 
\PassOptionsToPackage{hyphens}{url}\usepackage{hyperref}
\usepackage{graphicx}
\usepackage{amsthm}
\else 
\documentclass{article}

    \PassOptionsToPackage{numbers, compress}{natbib}

\usepackage{neurips_2026}
\fi 

\usepackage[utf8]{inputenc} 
\usepackage[T1]{fontenc}    
\usepackage{hyperref}       
\usepackage{url}            
\usepackage{booktabs}       
\usepackage{amsfonts}       
\usepackage{nicefrac}       
\usepackage{microtype}      
\usepackage{xcolor}         

\usepackage{amsthm}
\usepackage{amsmath}               
\usepackage{geometry}              
\usepackage{fancyhdr}              
\usepackage[shortlabels]{enumitem} 
\usepackage{enumitem}
\usepackage{mathrsfs}
\usepackage{mathtools}
\usepackage{tikz}
\usetikzlibrary{angles,quotes,positioning}
\tikzset{main node/.style={circle,fill=none,draw,minimum size=1cm,inner sep=0pt},}
\PassOptionsToPackage{hyphens}{url}\usepackage{hyperref}
\usepackage[autostyle=true,german=quotes]{csquotes}
\usepackage{cleveref}

\newcommand\numberthis{\addtocounter{equation}{1}\tag{\theequation}}

\usepackage{bm}
\usepackage{bbm}
 
\usepackage{graphicx}
\usepackage{color}

\usepackage{multirow} 

\newlist{todolist}{itemize}{2}
\setlist[todolist]{label=$\square$}
\usepackage{pifont}

\usepackage{caption}
\usepackage{subcaption}

\newtheorem{lemma}{Lemma}

\newtheorem{observationstar}{Observation}

\newtheorem{definition}{Definition}

\newcommand{\p}[1]{\left( #1 \right)}

\newcommand{\abs}[1]{\left \vert #1 \right \vert}

\newcommand{\cd}[0]{\cdot}
\newcommand{\nagents}[0]{\ensuremath{N}}
\newcommand{\ntasks}[0]{\ensuremath{T}}
\newcommand{\task}[0]{\ensuremath{t}}
\newcommand{\val}[0]{\ensuremath{v}}
\newcommand{\Val}[0]{\ensuremath{V}}

\newcommand{\modelset}[0]{\ensuremath{\mathcal{M}}}
\newcommand{\model}[0]{\ensuremath{m}}
\newcommand{\agg}[1]{\ensuremath{u(#1)}}
\newcommand{\ppick}[0]{\ensuremath{p}}

\newcommand{\valvec}[0]{\ensuremath{\vec{\val}}}

\newcommand{\abstracttext}[0]{
\begin{abstract}
Consider a marketplace of AI tools, each with slightly different strengths and weaknesses. By picking the right model for the task at hand, a user can do better than simply using the same model for everything. Routers operate under a similar principle, where sophisticated model selection can increase overall performance. However, aggregation is often noisy, reflecting imperfect user choices or routing decisions. This leads to two main questions: first, what does a \enquote{healthy marketplace} of models look like for maximizing consumer utility? Secondly, how can we incentivize producers to create such models? 
We show that winrate, a standard benchmark in LLM evaluation, can incentivize model creators to \emph{homogenize} for both types of model changes, reducing consumer welfare. We propose a new mechanism, weighted winrate, which rewards models for answers that are higher quality, and show that it provably improves incentives for producers to specialize and increases consumer welfare. We conclude by exploring the impact of our theoretical results in empirical benchmark datasets and discussing implications for benchmark design. 
\end{abstract}
}

\usepackage{framed, listings}
\usepackage{thmtools,thm-restate}

\usepackage{array}
\newcolumntype{C}[1]{>{\centering\let\newline\\\arraybackslash\hspace{0pt}}m{#1}}

\ifarxiv 
\else 
\title{Impacts of Aggregation on Model Diversity and Consumer Utility}

\author{%
  David S.~Hippocampus\thanks{Use footnote for providing further information
    about author (webpage, alternative address)---\emph{not} for acknowledging
    funding agencies.} \\
  Department of Computer Science\\
  Cranberry-Lemon University\\
  Pittsburgh, PA 15213 \\
  \texttt{hippo@cs.cranberry-lemon.edu} \\
}
\fi 

\begin{document}

\ifarxiv 
\title{Impacts of Aggregation on Model Diversity and Consumer Utility
}

\author{Kate Donahue\thanks{\url{kpd@illinois.edu}, CSAIL/LIDS,  Massachusetts Institute of Technology \& Siebel School of Computing and Data Science, University of Illinois, Urbana-Champaign} \and Manish Raghavan\thanks{\url{mragh@mit.edu}, MIT Sloan School of Management \& Department of Electrical Engineering and Computer Science.}}
\date{} 
\else 
\fi 

\maketitle

\abstracttext

\section{Introduction}

Modern AI models exhibit incredible  versatility. They are capable of performing a broad range of tasks---coding, translation, writing, and many more. In some settings, AI models exhibit varying abilities based on the task at hand: some write superior code, while others excel at legal reasoning. When this is the case, savvy consumers of models can benefit from \textit{aggregation} of models: strategically relying on the right model for the job, with the goal of outperforming any model in isolation. For example, LMArena ranks individual models as well as a router (called Max)
that aggregates individual responses across models; the router Max achieves the highest ELO score, outperforming individual
models~\citep{arena_blog_max_2026, lmarena}.
Anecdotally,
consumers are also beginning to act similarly, by strategically relying on
certain models for particular tasks and other models for different tasks. In
this sense, consumers in our setting derive utility from a \textit{set} of
models, not just a single one.

The potential benefits of aggregation are well-known. Learning paradigms like boosting~\citep{freund1997decision} rely on the aggregation of weak learners to build a strong overall system. Similarly, the literature on collective intelligence and team formation finds that diverse teams can succeed by leveraging the complementary strengths of individuals~\citep{hong2004groups}.
On the other hand, recent theory and evidence suggests that AI models are \textit{homogeneous}, failing to demonstrate the heterogeneity that might yield productive aggregation in practice~\citep{kleinberg2021algorithmic,bommasani2022picking,goel2025great,kim2025correlated}. Why is this the case, and what might lead to more diverse model ecosystems?

Our work investigates the \textit{incentives} created by the aggregation
  process itself. If model producers are rewarded based on how often their
  models are selected for use, what kinds of model ecosystems should we expect
  to see form? We focus on the incentive structures that govern the relationship
  between aggregation and ecosystem-level properties like homogeneity, asking
  whether alternative mechanisms could better align the objectives of producers
with the utility of users. Along the way, we find several counterintuitive
properties of aggregation that pose challenges for utility alignment.

\textbf{Contributions.}
In this paper, we investigate how the imperfect aggregation of agent
outputs---whether through automated routers or strategic human
selection---shapes both the realized utility of generative AI and the long-term
evolution of the agent marketplace. We focus on two central questions: what does a \enquote{healthy marketplace} of models that satisfies consumer demand look like? Second, how can we incentive self-interested producers to create such models? We begin by proposing a formal model of aggregation in
Section~\ref{sec:model}: the world is divided into discrete
\textit{tasks}, where agents have different performance on each
task.\footnote{We will sometimes use ``agent'' to refer to an AI model to avoid overloading the
term ``model.''} 
Consumers
noisily choose a model for each task according to a choice model---most of our
results focus on the Bradley-Terry choice model \citep{bradley1952rank, plackett1975analysis, luce1959individual}, though we consider
alternatives as well. In Section~\ref{sec:helpuser}, we focus on consumer welfare and demonstrate that the
addition of a new model or improvement of an existing model does not necessarily benefit consumers.
Intuitively, unless consumer choice is perfect, a worse model can ``distract'' a
user and lead to worse outputs. In Sections~\ref{sec:incententry} and~\ref{sec:incentretrain}, we study how how producer incentives affects consumer utility under aggregation. If producers are motivated by winrate (the probability their model is picked), we show that
competition leads to homogenization: however, consumer welfare would be
maximized by specialization, where different models invest in improvements on
different tasks. We propose and analyze a mechanism that better aligns producer
objectives with consumer welfare. Finally, we include experiments in
Section~\ref{sec:experiments} to explore how our results change incentives in real \enquote{markets} of models, before concluding.

\section{Model}\label{sec:model}
\ifarxiv
Here, we will briefly present our theoretical model for tasks, aggregation across tasks, objectives of consumers and model producers, as well as actions available to producers. 
\else 
\fi 

\noindent \textbf{Task structure:}  We assume that there are $\ntasks\geq 2$ different \enquote{tasks}
or \enquote{task types}, and a set of $\modelset$ different models. For each
task $\task \in [\ntasks]$, a given model $\model$ has value $\val_{\model,
\task}\geq 0$  or $\val_{\model, \task} = \emptyset$ (that is, fails to return a
response for this task). We will assume that such value is deterministic or,
alternatively, we will consider $v_{\model, t}$ to be the expected value when
using model $\model$ on task $t$. For a given set of models $\modelset$, we will use
$\valvec_{\task}$ to denote the list of values of responses on task $\task$. We will use $\oplus$ to denote adding another value to the list: that is $\valvec \oplus  \val$ is the vector $\valvec$ with the value $\val$ appended to it. We will often use $\nagents = \abs{\modelset}$, the number of models in a market. 
With this task structure, we can express the fact that
different models have different strengths and weaknesses. A model's
\emph{total value} is given by its sum of values over tasks, $\Val =
\sum_{\task \in [\ntasks]} \val_{\task}$, and its \emph{average value} is its total value, divided by the number of tasks: these quantities will turn out to be
crucial for later properties of consumer welfare and producer incentives. 

\noindent \textbf{Aggregation:}
The core phenomenon that we will study is \textit{aggregation}, when a user
derives utility across multiple available models. We will focus on \emph{choice-based aggregation}, where for each task a user has some probability $\ppick_{\model}(\valvec_{\task})\in [0, 1]$ of picking a particular model. We require that the total probability of models being picked sums up to 1: $\sum_{\model \in
[\nagents]} \ppick_{\model}(\valvec_{\task})) = 1$. A user's utility for a particular task given choice-based aggregation is given by the expected value of the model she picks: $\agg{\valvec_{\task}} = \sum_{\model \in [\nagents]}
\val_{\model, \task} \cd \ppick_{\model}(\valvec_{\task})$. We will define
$\ppick_{i}(\valvec) = 0$ if $\valvec_i = \emptyset$: any model that fails to
return an answer has probability 0 of being picked.  A few special cases of choice-based aggregation are: \emph{random} aggregation, where the choice model picks randomly among items, and \emph{optimal} aggregation, which always picks the item with highest value (randomizing in the case of ties). 

\ifarxiv 

\begin{definition}[Random]\label{def:random}
Random: pick randomly among all items. 
$$\ppick_{\model}(\valvec_\task) = 1/\abs{\valvec_{\task}} \quad  \forall \model \in [\nagents] \mid \valvec_{\model, \task} \ne \emptyset$$
\end{definition}

\begin{definition}[Optimal]\label{def:opt}
Optimal (perfect accuracy): always pick the item with highest value (randomizing between them in the case of ties). 
\begin{equation*}
\ppick_{\model}(\valvec_\task) = \begin{cases}
1/\abs{\{\max_{\model \in \modelset} \valvec_{\model, \task}\}} &  \val_{\model, \task} = \max_{\model \in \modelset} \val_{\model, \task}\\
0  &  \text{otherwise}
\end{cases}
\end{equation*}
\end{definition}
\else 
\fi 
While optimal and random aggregation are helpful limiting cases, neither is
a particularly useful model of aggregation in settings we are interested in
studying, such as human decision-making or routing between LLMs. Instead, we might expect decisions to be smooth and lie somewhere in between these two extremes: better than random, but worse than optimal, and most likely to \enquote{confuse} items of similar value: most likely to incorrectly pick a lower-value response when it is close in value to the best option, for example.  
More generally, we'd like a flexible
choice model that can express a range of routing \enquote{accuracy} levels in picking the best model for the task.
One choice model satisfying all of these properties is Bradley-Terry-Luce: 

\begin{definition}[Bradley-Terry-Luce ]\label{def:softmax}
Bradley-Terry-Luce (BTL) (alternatively known as Bradley-Terry or Plackett-Luce~\cite{bradley1952rank, plackett1975analysis, luce1959individual}): with a given temperature $\beta >0$, the probability of picking a item $\model$ from a list of values $\valvec_{\task}$ is given by:  
\ifarxiv 
\begin{equation*}
\ppick_{\model}(\valvec_\task) = \frac{\exp(\val_{\model, \task}/\beta)}{\sum_{\model'\in \modelset}\exp(\val_{\model', \task}/\beta)}
\end{equation*}
\else 
$\ppick_{\model}(\valvec_\task) = \frac{\exp(\val_{\model, \task}/\beta)}{\sum_{\model'\in \modelset}\exp(\val_{\model', \task}/\beta)}$. 
\fi 
\end{definition}
BTL aggregation approaches random selection in the limit of high temperature ($\beta \rightarrow \infty$) and approaches perfect selection in the limit of low temperature ($\beta \rightarrow 0$). The probability of picking an item is increases in value, and the probabilities are most close for items that are close in value. Because of these properties, Bradley-Terry is frequently used as a model of human behavior in market competition, e.g. \citep{jagadeesan2023improved, einav2025market, xu2025heterogeneous}. In Appendix \ref{app:errorpatternsrouter} we include experiments exploring its applicability to modeling errors in routing tasks. 

\noindent \textbf{Consumer welfare and producer incentives:}
Throughout, we will assume that consumer welfare is given by the sum of user utility over tasks, assuming aggregation by the BTL choice function: $\sum_{\task \in [\ntasks]} \sum_{\model \in [\nagents]}
\val_{\model, \task} \cd \ppick_{\model}(\valvec_{\task})$. In general, we will \emph{not} assume that producers are directly moving to optimize consumer welfare. Instead, we will assume in the status quo, they are motivated by \emph{winrate}, or the probability that their model is picked by the user: $\ppick_{\model}(\valvec_{\task})$\footnote{Across multiple tasks, we assume they aim to maximize the average winrate.}. One motivation for this is that winrate is used in leaderboard design, as in \cite{lmarena}, while other motivations come from the common assumption in modeling market competition that producers are aiming to maximize market share (e.g. \citep{jagadeesan2023improved, einav2025market, xu2025heterogeneous}). 

In Sections \ref{sec:incententry} and \ref{sec:incentretrain}, we will show that
winrate and consumer welfare are at odds with each other: producers
optimizing for winrate could lead to sub-optimal consumer welfare. To address
this, we propose a mechanism to align user and producer incentives: weighted winrate (Definition \ref{def:weightedwinrate}). Weighted winrate rewards producers not just for an output that is better than other outputs (reflected by winrate, $\ppick_{\model}(\valvec_{\task})$), but also for a higher value itself (reflected by value $\valvec_{\model, \task}$). 

\begin{definition}[Weighted winrate (ours)]\label{def:weightedwinrate}
\emph{Weighted winrate} is the probability that a model is picked \emph{multiplied by the value it gives}, or 
$\ppick_{\model}(\valvec_{\task}) \cd \valvec_{\model, \task}$
\end{definition}
We view this mechanism as very natural: it mimics, for example, the \emph{revenue} that a producer would get from selling a good at a given price, given some market share. One specific feature of modeling LLMs is that we can intervene through leaderboards to implement mechanisms, while such an intervention might be more challenging in other markets. Weighted winrate could be implemented by directly asking users to indicate not only their most preferred option from a set of responses, but also their value for that response. This could be done with coarsened features (e.g. score from 1 to 10, or a 1-5 star rating): we discuss complexities around eliciting labels in Section \ref{sec:conclude}.

\noindent \textbf{Producer actions (model creation and model replacement):} Throughout the paper, we will study two different types of strategic behaviors
that producers can take in pursuit of their objectives: \emph{Model creation} is when a producer creates a new model and adds it to the existing market, while \emph{model replacement} occurs when a producer replaces an existing model with a new one (e.g. with local post-training). Model creation mimics a setting where a producer creates an entirely new model, with complete flexibility over its strengths and weaknesses (limited by some cap in \emph{total value}). Model replacement mimics a setting where a model producer can make small (instantaneous) changes to an \emph{existing} model (e.g. post-training), that may slightly change its existing strengths and weaknesses. While model creation and replacement are connected, they are different actions with different implications for consumer welfare. For example, suppose we begin with a market made of two models, $\{A, B\}$. Model creation would result in the new market $\{A, B, C\}$, while model replacement would result in the new market $\{A, B'\}$, with different impacts on consumer welfare. 

\section{Related work}\label{sec:related}

\textbf{Aggregation of generative models:}
The closest area of research to ours is work that studies aggregation with generative AI tools. For example, \cite{ai2025beyond} designs optimal aggregation mechanisms given some fixed set of multiple LLMs: this could be viewed as the inverse of our setting, where we study a fixed aggregation mechanism and study how to incentivize the creation of models that satisfy that it. A related paper on incentives in benchmarks is \citet{ethayarajh2020utility}: their focus is on users picking models according to their rank on a benchmark, while we study consumers \emph{aggregating} across models on different tasks. 
Another paper that studies consumer aggregation is \cite{meenanivasini}, which differs from ours primarily in our focus on mechanisms and our smoother aggregation function (BTL).  
From an applied angle, \cite{lu2024does} studies \enquote{second opinions} in humans relying on AI tools. \ifarxiv This work relates to ours in that it studies the impact of multiple AI tools with imperfect human aggregation, but differs in that it focuses on empirical performance of humans using AI tools and focuses on using exactly two AI tools. \else \fi 
Other work studying the performance of agentic teams with specific focus on diversity includes \citep{kearns2026networked, bedi2025optimization, li2025rethinking, wang2024mixture, song2025human}.

A complementary line of work applies insights from voting theory to settings with multiple benchmarks, where different benchmarks may reflect varying performance across tasks, for example. Some work aims to design multi-task benchmarks that incorporate rankings of the same models across multiple distinct tasks: this work differs from ours in that the goal is typically to produce a single benchmark, while we study how producers compete over multiple tasks \citep{zhang2024inherent, gordienko2026beyond, colombo2022best, rofin2023vote}. Other work strategically selects subsets of representative benchmarks \citep{procacciametritocracy}. Finally, a related line of literature studies the performance of teams \citep{hong2004groups}, as well as specifically studying the utility of voting or other aggregation mechanisms with diverse teams (e.g. \citep{jiang2014diverse}). 

\noindent \textbf{Strategic competition:}
One related line of work studies model-providers competing over users, as in \citep{jagadeesan2023improved, einav2025market, xu2025heterogeneous, weimarket}. There are strong connections, such as the use of modeling user choice through BTL functions and characterizing producers' best response. However, one key difference of our paper from this work is our focus on the mechanism design aspect: our goal is not only to understand the behavior of producers, but to change their incentives with the goal of better aligning with consumer welfare. Other related works (e.g. \cite{haghtalab2020maximizing, hartline2023optimal, liu2022strategic}) have studied whether it is possible to improve consumer utility by designing mechanisms that agents strategically best-respond to. This line of work tends to differ from ours in its motivating application area (where ours focuses more specifically on generated AI. A few papers have specifically looked at mechanism design for leaderboards (e.g. \cite{chen2026leaderboard, hays2026strategic}): these tend to focus more on score inflation, as compared to our goal of managing homogenization. 

\ifarxiv 
More generally, there has been a rich literature studying agents competing over items, such as Colonel Blotto and Tullock contests. The Colonel Blotto game models a pair of Colonels competing over a series of \enquote{battlefields}, with some fixed total amount of value, and an agent \enquote{winning} a battlefield according to some cost function based on how strongly they out-compete the other Colonel \citep{borel1921theorie}. Tullock contests similarly study different agents who allocate different value in the hope of winning a probabilistic \emph{contest}, where the aggregation function is typically given by $\frac{\val_a}{\val_a + \val_b}$ \citep{Tullock2001}. The closest variant of Colonel Blotto is the \enquote{Lottery Blotto} game \cite{friedman1958game}, which studies Colonel Blotto where the outcome is determined by a Tullock contest function. Results here show that the unique equilibrium is to have players have allocation that is equal across tasks (if all tasks are equally valuable). While these literatures differ from our setting in multiple ways, one crucial one is the function governing \enquote{winrate.} Specifically, winrate is often either deterministic (by whichever agent has the larger value), or when probabilistic is often given by convex functions, such as $\frac{\val_{ai}}{\val_{ai} + \val_{bi}}$. Given a function like this, agents are generally incentivized to best-respond by evenly spreading value across tasks. However, in our setting we use BTL to determine winrate, which is \emph{not}
convex. This leads to much more complex patterns of behavior. We view this complexity as a benefit, given that it arises from satisfying natural properties like users becoming more easily \enquote{distracted} by items that are similar in value, discussed in Section \ref{sec:model}.
\else More generally, there has been a rich and long literature studying agents competing over items, such as the Colonel Blotto game, with a pair of Colonels competing over a series of \enquote{battlefields} \citep{borel1921theorie}, and Tullock contests where agents aim to allocate value in the hope of winning a probabilistic \emph{contest} \citep{Tullock2001}: see Appendix \ref{app:related} for a deeper discussion. A few key differences are the actions available to players differ substantially (\enquote{model creation,} \enquote{model replacement}) see Section \ref{sec:model}, which we view as being more natural for LLM creation) and our work's focus on consumer welfare. 
\fi 

\ifarxiv 
\noindent \textbf{Non-monotonicity in aggregation:} In Section \ref{sec:helpuser} we demonstrate a \enquote{non-monotonicity} result: increasing the value of a model could \emph{decrease} consumer welfare. Related non-monotonicities have been studied in other settings. 
For example, one close connection is in auctions, where results have shown that adding in more bidders could \emph{decrease} revenue (e.g. \citep{rastegari2011revenue}). While this phenomenon is very closely related, the types of aggregation mechanisms and objectives differ between our settings and auctions. Related non-monotonicities also occur in social choices, such as settings where it could be the case that \emph{more} votes for a candidate could cause her to go from winning a race to losing (e.g. \citep{lepelley2018monotonicity}). This could reflect violations of Pareto efficiency or the unanimity principle, which would declare that more support for an option should not \emph{reduce} its property of winning. One close connection is \cite{small1981applied}, which shows that consumer welfare is monotone with a different utility model where users derive utility according to the value plus the additive $\epsilon$  noise: however, in our setting, we argue that it makes more sense to use the noise terms $\epsilon_i$ to model decision errors, rather than true variation in quality that a user perfectly responds to. Our results differ from most voting settings in that we assume there is some true value over candidates (here, model responses) and we are studying the \enquote{expected value} of the election, while most voting mechanisms are candidate-neutral and focus on the probability of individual candidates being selected.

We note that our \enquote{monotonicity} condition is \emph{not} the same as \enquote{regularity} (e.g. \cite{cerreia2019deliberately, caliari2025luce}), a different property often studied in relation to choice models. Regularity is defined as saying that the probability of an item being picked only stays constant or decreases as more items are added (we refer to this property as \enquote{substitutability}, in Definition \ref{def:substitutability}). A choice function could exhibit monotonicity or regularity, or neither, or both: regularity is a property of the probability of an item being picked (given other items added to the set), and monotonicity is a property of the value that the choice function gives the user (as items in the set increase in value).
\else 
\fi 

\section{Aggregation and consumer welfare}\label{sec:helpuser}

First, in this section we begin by studying how consumer welfare given by aggregation over a set is influenced by model creation or model replacement. Our primary results for this section will illustrate that optimizing consumer welfare is surprisingly non-trivial, even leaving aside questions of producer incentives. All proofs are deferred to Appendix \ref{app:helpuser}. 

\noindent \textbf{Model creation:}
As a warm-up, note that model creation with perfect aggregation or random aggregation always has straightforward impacts on consumer welfare. With model creation, adding a new model always weakly improves consumer welfare, while with random aggregation, it helps only when the new model has higher average value.

\ifarxiv 
\begin{observationstar}
With \emph{perfect aggregation} model creation always weakly improves utility. \\
With \emph{random aggregation}, model creation strictly improves utility if and only if the new model has strictly higher total value than the existing average.
\end{observationstar}
\else 
\fi 

Theorem \ref{thrm:higheraccadd} shows that BTL aggregation adds complexity to this story. Specifically, given an existing set of models $\modelset$ and a new model $B$, there are scenarios where $B$ can be \emph{weaker} than every model in $\modelset$ and still the creation of $B$ can help consumer welfare, and also be scenarios where $B$ is \emph{better} than every model in $\modelset$ and the addition can \emph{hurt} consumer welfare.

\begin{restatable}{theorem}{higheraccadd}\label{thrm:higheraccadd}
Given BTL aggregation, it is possible to have a set of existing models $\modelset$ and a new model $B$ such that either 1)  $B$ has higher average value than every model in $ \modelset$, and yet adding $B$ to $\modelset$ decreases consumer welfare, or 2) $B$ has \emph{lower} average value than every model in $\modelset$, and yet adding $B$ to $\modelset$ \emph{increases} consumer welfare. Moreover, this effect is dependent on temperature ($\beta$). 
\end{restatable}
\ifarxiv 
The proof is by a series of examples of $B$ and $\modelset$ such that the properties hold: $B$ has higher (resp. lower) value than any model in $\modelset$ in isolation, but when added to $\modelset$, leads to worse (resp. better) value than the entire set $\modelset$ before aggregation. These examples rely on the degree of complementary strengths and weaknesses within $\modelset$, the degree of additional complementary strengths that $B$ brings to $\modelset$, and the strength of aggregation reflected in the temperature $\beta$.
\else 
\fi 
The implication of Theorem \ref{thrm:higheraccadd} is that when users aggregate across models, optimizing model creation for consumer welfare can be non-intuitive, an insight which will be useful in our later analysis of mechanisms to improve consumer welfare. 

\noindent \textbf{Model replacement:}
We next turn to analyzing model replacement: when one model is replaced with a different model with different values on each task. Again, perfect aggregation and random aggregation provide useful limiting cases for when model replacement improves consumer welfare: Note that with \emph{perfect aggregation}, model replacement improves consumer welfare if and only if it increases the maximum value on each task. With \emph{random aggregation}, model replacement improves consumer welfare if and only if it improves the average value of the model. One natural property we might want is the following: weak task-wise improvement in models would only lead to weak increases in utility. That is, replacing a model with one that dominates the previous model on every task should only \emph{improve} consumer welfare. Monotonicity (Definition \ref{def:monotone}) makes this property formal. 

\begin{definition}[Monotone]\label{def:monotone}
    A choice function is monotone if task-wise increase in values always leads to increased consumer welfare: for all $\valvec_\model, \valvec_\model'$ with $\valvec_{\model, \task}\leq \valvec_{\model, \task}'$ for all $\task \in [\ntasks]$, $\agg{\valvec} \leq \agg{\valvec'}$. 
\end{definition}

While monotonicity is a natural property to desire, unfortunately it is not guaranteed. Lemma \ref{lem:exnonmonotone} gives an example where monotonicity \emph{fails} with BTL aggregation. For intuition, the proof is given by example: consider a set with two values $\{x, 5\}$, for $x = 1, 4$. When the value of $x$ increases from 1 to 4, the value of the first item increased, but so does the probability that element $x$ would be picked. For BTL aggregation with noise parameter $\beta=1$, it turns out that the value of set $\{1, 5\}$ is \emph{higher} than the value of set $\{4, 5\}$\footnote{Related phenomena have been seen in auctions, voting, and social choice: \ifarxiv see related work in Section \ref{sec:related} \else see Appendix \ref{app:related} \fi for a deeper discussion. }. At a high level, the intuition behind non-monotonicity is that a mediocre answer can be more distracting than a clearly terrible answer: for this reason, model improvement can sometimes be harmful. 
\begin{restatable}{lemma}{exnonmonotone}\label{lem:exnonmonotone}
BTL aggregation is non-monotone: e.g. a model
could weakly \emph{increase} the value of its responses on every task, and consumer welfare might \emph{decrease}. 
\end{restatable}

\ifarxiv 
\begin{figure}
\centering 
\includegraphics[width=3in]{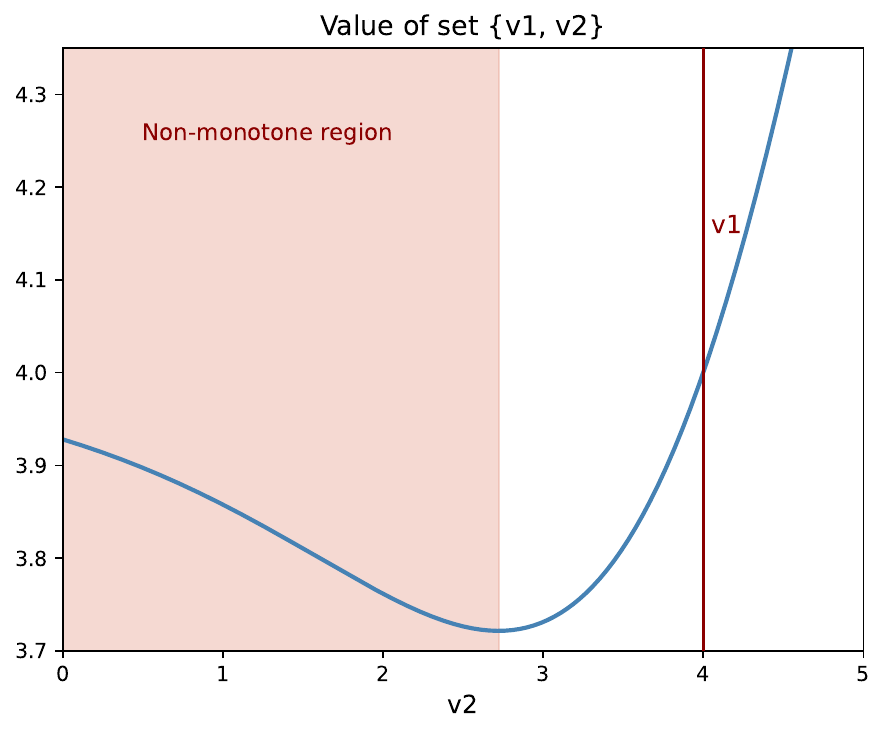}
\caption{Example demonstrating non-monotonicity (Definition \ref{def:monotone}) and non-monotonic region (Definition \ref{def:nonmonotoneregion}). Blue curve shows the value of a set containing two elements $\{v_1, v_2\}$, where aggregation is done with softmax with temperature $\beta$. $v_1$ has fixed value 4, while $v_2$'s value varies. Note that for parts of the figure, the value of $v_2$ \emph{increases}, but the value of the set \emph{decreases}.  }
\label{fig:softmax_example}
\end{figure}
\else 
\fi 

One natural question might be whether it is BTL specifically that leads to
non-monotonicity. Surprisingly, the answer turns out to be no---in fact, any
\emph{separable} choice function (Definition \ref{def:sep-def}) also cannot be
monotone, as shown in Theorem \ref{thrm:sepmonotone}. Intuitively, a function is separable if the probability of an item being picked is proportional to some function of its value. BTL is a separable choice function, but so are many more choice functions, so Theorem \ref{thrm:sepmonotone} implies non-monotonicity is a much wider problem. 

\begin{definition}[Separable]\label{def:sep-def}
A choice function $\ppick(\cdot)$ is separable if there is some function $f(\val) \geq 0$ such that $\ppick_i(\valvec)\propto f(\valvec_i)$: we also require $f(\cd)$ be differentiable, monotonic, and non-constant. 
\end{definition}

\begin{restatable}{theorem}{sepmonotone}\label{thrm:sepmonotone}
A separable choice function cannot be monotone.
\end{restatable}

However, we also show that non-monotonicity is not an inherent trait of choice functions: in the proof of Lemma \ref{lem:monotoneex} we give an example of a choice function $\ppick(\cd)$ constructed so that it satisfies monotonicity. At a high level, the aggregation function in Lemma \ref{lem:monotoneex} is based on the \emph{gap} between values between items, rather than the values themselves, which side-steps the impossibility condition in Theorem \ref{thrm:sepmonotone}. 
\begin{restatable}{lemma}{monotoneex}\label{lem:monotoneex}
There exists an aggregation function that is monotone in value. 
\end{restatable}

Finally, Lemma \ref{lem:sufficientmonotone} shows that BTL aggregation satisfies monotonicity under certain conditions. Specifically, we will say that a choice function is in the \enquote{monotone region} whenever, for a specific set of values $\valvec$, the derivative of consumer utility with respect to $\valvec_i$ is positive.
\ifarxiv 
\begin{definition}[Monotone region]\label{def:nonmonotoneregion}
A choice function $\ppick(\cdot)$ is in the 
\emph{monotone region} with respect to agent $i$, for a set of values $\valvec$ and parameters of choice function $\phi$, if the derivative of consumer utility with respect to $\valvec_i$ is non-negative, and is in the non-monotone region if the derivative is negative.  
\end{definition}
\else 
\fi 
\footnote{A function is monotone if for all values, it is always in a monotone region.} Lemma \ref{lem:sufficientmonotone} shows that the BTL function is in the monotone region under certain conditions on the values and level of noise ($\beta)$. In particular, it is in this region when $\beta =0$ (perfect aggregation), and when $\beta$ is sufficiently large (sufficiently noisy), where the \enquote{sufficiently large} is a function of the values $\valvec$. 
\begin{restatable}{lemma}{sufficientmonotone}\label{lem:sufficientmonotone}
BTL aggregation function is in the monotone region when 1) $\beta=0$, corresponding to perfect aggregation, or 2) when $\beta$ is sufficiently high: that is, for a given set of values $\valvec$ and a model $i$, BTL aggregation is guaranteed to be in the monotone region whenever $\beta > \max_{i,j} \valvec_j - \valvec_i$. 
\end{restatable}
\ifarxiv 
\noindent \textbf{Tension between monotonicity and model creation}
Finally, we want to highlight a tension between consumer welfare increase from model \emph{replacement} (monotonicity) and consumer welfare increase from model \emph{creation}. That is, for any choice functions satisfying natural properties (substitutability and anonymity), these two properties cannot be simultaneously satisfied. 

\begin{restatable}{theorem}{mutexclusive}\label{thrm:mutexclusive}
These two properties are mutually exclusive for choice-based, continuous aggregation functions satisfying substitutability and anonymity\footnote{Informally, a function satisfies anonymity if the probability of picking an item depends only on its value, and substitutability if the probability of picking an item only stays constant or increases when its value increases: see Appendix \ref{app:tensionmonotonecreate} for formal definitions. }: 1) the function is monotone, and 2) there are always weak benefits to adding in another agent. 
\end{restatable}
For intuition in the proof of Theorem \ref{thrm:mutexclusive}, we first show that if there are always weak benefits to adding in another model, there must be probability 0 of picking a model with low enough value (e.g., if it has lower value than any other model's value already in the set). However, as we increase the value of that model, the probability of picking that model must at some point go from 0 to a positive probability, while it is still the smallest value among all models: this decreases the value of the set, violating monotonicity. 

\else 
\fi 
Taken together, results in this section show that model improvement can fail to improve consumer welfare for a surprisingly broad set of scenarios. We will use insights from this section in designing incentives for model improvement in the following sections.

\section{Incentivizing model creation}\label{sec:incententry}

In Section \ref{sec:helpuser} we studied how model creation and model replacement affect consumer welfare. However, in current marketplaces, models are created by specific producers, who may be incentivized by objectives that are different from purely optimizing consumer welfare. Specifically, we will assume that producers are currently motivated by \emph{winrate}: the probability that their model is selected. In this section and in Section \ref{sec:incentretrain}, we will show how winrate and consumer welfare may be at odds with each other: producers optimizing for winrate could lead to sub-optimal consumer welfare. Finally, we will discuss our proposed mechanism to better align user and producer incentives (weighted winrate). 
In this section, we will explore how to incentivize \emph{model creation}: the addition of an entirely new model. All proofs are deferred to Appendix \ref{app:incententry}. We will assume the following structure: there are $\ntasks$ tasks and a set of existing models $\modelset$, and each task has value $\valvec_{\task}$. Then, a producer decides to make a new model of total value $\Val$: assume that they can create this model by allocating non-negative value arbitrarily across tasks, and choose to do so according to some objective: winrate, weighted winrate, or consumer welfare. Intuitively, this is connected to how producers might make strategic choices in how to shape their models based on incentives, such as investing in training data or capabilities tailored towards specific tasks. 
We will explore whether producers choose to specialize, placing
most value on a few tasks, or homogenize, spreading value across tasks. We will
show that winrate leads producers to homogenize, which leads to lower utility,
and our proposed mechanism (weighted winrate) leads producers to specialize,
leading to improved consumer welfare.

\noindent \textbf{Model producer actions optimizing different objectives: } We will begin by analyzing the optimal response to each objective. First, Lemma \ref{lem:optvalueb} shows that the optimal allocation for consumer welfare is to either maximally specialize, placing all value on a single task, or to abstain (if the total amount of value $\Val$ for the new model is very low). This result will be the benchmark by which we measure how well producers satisfy consumer welfare. 

\begin{restatable}{lemma}{optvalueb}\label{lem:optvalueb}
If a new model with total value $\Val$ aims to maximize consumer welfare, its best action is: if $\Val$ has higher value than the minimum existing values on tasks ($\Val \geq \min_{i \in [\ntasks]} \agg{\valvec_i}$), specialize by allocating all value on a single task $i$ which maximizes $\ppick_{\nagents+1}(\valvec_i \oplus \Val) \cd (\Val - \agg{\valvec_i})$. If $\Val$ is not sufficiently high, then its best action is to abstain (fail to create a model).   
\end{restatable}

Next, Theorem \ref{thrm:optwinb} studies best response for winrate, showing that that if the new model is sufficiently powerful (has sufficiently high $\Val$), producers are incentivized to \emph{homogenize}, equalizing their winrate across tasks. Intuitively, this occurs because producers optimizing for winrate face decreasing marginal benefit by increased value on a single task. Note the lower bound on $\Val$ in Theorem \ref{thrm:optwinb} depends on the strength of aggregation $\beta$: the producer's incentives are influenced by consumer noise\footnote{Theorem \ref{thrm:optwinb} gives us a lower bound on $\Val$, where any producer above this must be incentivized to homogenize. For a complementary result, Lemma \ref{lem:optwinbound} in Appendix \ref{app:incententry} gives an \emph{upper} bound on $\Val$, below which any producer must specialize. }.

\begin{restatable}{theorem}{optwinb}\label{thrm:optwinb}
Consider a set of tasks with existing values $\{\valvec_i\}$ across tasks. Then, an incoming model with total value $\Val$ optimizing for \emph{total winrate} must equalize winrate across tasks whenever:
$\Val > 2 \cd \ntasks \cd \beta \cd \max_{j \in [\ntasks]}\log\p{\sum_{i \in [\nagents]} \exp(\valvec_{ji}/\beta)}$
\end{restatable}

Finally, Lemma \ref{lem:optcombb} analyzes the incentives produced by our new mechanism, weighted winrate (Definition \ref{def:weightedwinrate}): weighted winrate incentivizes producers to maximally specialize by placing all value on a single task and returning a value of 0 elsewhere.  

\begin{restatable}{lemma}{optcombb}\label{lem:optcombb}
Suppose there is a set of existing models and tasks with values $\{\valvec_{\task}\}$. Then, an incoming model with total value $\Val$ optimizing for weighted winrate would best respond by putting all value on a single task (the task that maximizes $\ppick_{\nagents+1}(\valvec_i \oplus \Val)$), and value 0 elsewhere.
\end{restatable}
The results of Lemmas \ref{lem:optvalueb} and \ref{lem:optcombb} show that in order to maximize weighted winrate and consumer welfare, the optimal allocation has all value on a single task, and either value 0 or $\emptyset$ (abstain) on all other tasks. This outcome could reflect settings where a producer strategically chooses to not participate in a type of task---for example, Nano Banana specializes in image generation, while Gemini Deep Research specializes in academic research. However, it may not be realistic to have models with such performance, either because the benchmarks themselves may become saturated at a particular maximum value, or because scaling laws suggest that achieving increased performance becomes increasingly hard. 

However, the main story of our results generalizes to such settings. Here, we assume that every task has some maximum value $\val_i^*$ above which value cannot be achieved: this could reflect the maximum score on a benchmark, or the maximum utility that a user can obtain for a response. Then, our results for weighted winrate and consumer welfare generalize as follows: 
\begin{restatable}{lemma}{scalinglaw}\label{lem:scalinglaw}
If each task $i$ has a maximum value $\val_i^*$, the allocation that maximizes weighted winrate or consumer welfare is to greedily allocate value for each task according to the objectives in Lemmas \ref{lem:optvalueb} or \ref{lem:optcombb} respectively. 
\end{restatable}

\noindent \textbf{Impacts on consumer welfare: } Next, we will turn to analyzing the impact of weighted winrate on consumer welfare: does it lead to better outcomes than winrate? The weighted winrate allocation in Lemma \ref{lem:optcombb} differs from the social welfare allocation in Lemma \ref{lem:optvalueb} in two ways: first, the choice of which task to allocate value in may differ, secondly, agents optimizing for social welfare would abstain on some tasks (fail to return a value), while agents optimizing for weighted winrate would return an output of value 0. \ifarxiv \else \footnote{\fi Note that the weighted winrate mechanism leaves producers indifferent between returning an item of value 0 or abstaining, as both give value $\valvec_{i, \task} \cd \ppick_i (\valvec_{\task})=0$: we take the worst-case assumption that producers return a response of value 0. However, these actions are different from the consumer's perspective, because returning a response of value 0 steals probability mass from higher-value responses. \ifarxiv \else }\fi. Given these differences, it is clear that weighted winrate is weakly worse for consumer welfare than directly optimizing for welfare itself. But, as we show in Theorem \ref{thrm:weightedwinratebetter}, weighted winrate always at least as good as winrate: 
whenever winrate incentivizes producers to homogenize outputs by equalizing winrate (as in Theorem \ref{thrm:optwinb}), the incentives produced by weighted winrate would lead to strictly higher consumer welfare. 

\begin{restatable}{theorem}{weightedwinratebetter}\label{thrm:weightedwinratebetter}
Whenever agents would respond to winrate by equalizing winrate (conditions of Theorem \ref{thrm:optwinb}), weighted winrate would lead to strictly higher consumer welfare than winrate.
\end{restatable}

\section{Incentivizing model replacement}\label{sec:incentretrain}
In Section \ref{sec:incententry}, we studied incentives in model \emph{creation}, when producers are making an entirely new model. In this section, we will study incentives in model \emph{replacement}, when an existing model has instantaneous changes made to its strengths and weaknesses, mimicking small tweaks like post-training in models. Incremental updates to existing models could occur with post-training of a mostly-finalized model to fine-tune its performance on specific tasks, or minor tweaks to capabilities and tool use, potentially in response to competition from other models or performance on leaderboards. 
The high-level story will be that the incentives and outcomes are very similar to with model creation: winrate tends to incentivize homogenization, to the detriment of consumer welfare, while weighted winrate gives incentives that leads to improved consumer welfare.  All proofs are deferred to Appendix \ref{app:incentretrain}. We will assume the following structure: there are $\ntasks>1$ tasks, and exactly two producers $\nagents=2$. Each of the producers $A$ and $B$ has an existing allocation across tasks $\valvec_a, \valvec_b$. Each producer can make instantaneous changes to its relative allocation. In particular, they calculate the derivative of some objective  (e.g. winrate, weighted winrate, consumer welfare) and change their allocation in order to improve their utility according to this objective. 

\noindent \textbf{Model producer actions optimizing different objectives: } Again, we will begin by analyzing how model producers would respond to different objectives (social welfare, winrate, and weighted winrate). Specifically, Theorem \ref{thrm:winvaluepick} considers a pair of tasks, and shows that for winrate, producers are always incentivized to improve on the \emph{same} task, but for optimizing for consumer welfare, they would always prefer to optimize \emph{different} tasks. While Theorem  \ref{thrm:winvaluepick} shows how producers would make decisions between \emph{pairs} of tasks, this directly implies results about the total ordering over tasks: with winrate, producers would have identical orderings over tasks they choose to improve in, and for consumer welfare, producers would have exactly inverted orderings over tasks they choose to improve in. Note that this matches the pattern we saw in Section \ref{sec:incententry}: optimizing for winrate tends to incentivize homogenization, while optimizing for consumer welfare tends to lead to specialization.

\begin{restatable}{theorem}{winvaluepick}\label{thrm:winvaluepick}
For $\nagents=2$ agents optimizing for winrate, agents would always pick the \emph{same} task to increase value in, whereas agents optimizing for consumer welfare always pick \emph{different} tasks. 
\end{restatable}
Lemma \ref{lem:valueincent} shows that our proposed mechanism, weighted winrate, incentivizes behavior aligned with social welfare maximization whenever the models start out by being sufficiently specialized: that is, when their winrates begin by being sufficiently different. 

\begin{restatable}{lemma}{valueincent}\label{lem:valueincent}
For $\nagents=2$ agents responding to weighted winrate, agents are incentivized to improve in \emph{different tasks} whenever agents begin by being sufficiently specialized (e.g. $\abs{\ppick_a(\valvec_i) - \ppick_a(\valvec_j)}$ is sufficiently large). 
\end{restatable}
\noindent \textbf{Impacts on consumer welfare:}
Because weighted winrate incentives firms to specialize
Given that the incentives produced by weighted winrate align with socially optimal behavior, it is natural to think that weighted winrate should increase social welfare relative to winrate.
We study this question next.
Because agent development is on a trajectory of incremental improvement, not at equilibrium, we will consider \emph{instantaneous consumer welfare}, which is the change in consumer welfare given the instantaneous change in producers' values over tasks. In Section \ref{sec:incententry}, Theorem \ref{thrm:weightedwinratebetter} showed that weighted winrate improves consumer welfare in the model \emph{creation} setting. We might hope to similarly show that weighted winrate always improves consumer welfare in the model \emph{replacement} setting. However, the story turns out to be slightly more nuanced. First, Lemma \ref{lem:observationstar} shows that weighted winrate does \emph{not} always lead to an improvement over winrate.  

\begin{restatable}{lemma}{observationstar}\label{lem:observationstar}
Weighted winrate does \emph{not} always lead to improvement in instantaneous consumer welfare, as compared to the incentives produced by winrate. 
\end{restatable}

Specifically, this proof is given by an example (formally presented in Table \ref{tab:exinst} in Appendix \ref{app:incentretrain}.). However, this example has an unusual property: one model ($B$) is worse than the other model ($A$) on \emph{both} tasks, and moreover, the derivative of consumer welfare with respect to improvement (instantaneous consumer welfare) in model $B$ is \emph{negative} for both tasks: it is in the \emph{non-monotone} region as explored in Section \ref{sec:helpuser}). Theorem \ref{thrm:bothneg} shows that this is not an accident: every example where weighted winrate incentives producers to pick a worse pair of tasks than winrate, it \emph{must} be a setting where one producer is in the non-monotonic regime on \emph{both} tasks. Note that Theorem \ref{thrm:bothneg} is not an if-and-only-if condition: it is a conservative guarantee.  

\begin{restatable}{theorem}{bothneg}\label{thrm:bothneg}
Consider a model replacement setting with $N=2$ agents: the only scenario where weighted winrate causes \emph{worse} instantaneous consumer welfare than winrate is when one player is in the non-monotonic regime on \emph{both} picked tasks. 
\end{restatable}

We view this condition as being very useful for understanding how to guide producer incentives. Specifically, given that the only failure of weighted winrate occurs when one player is in the non-monotone regime, the best outcome for consumer welfare would be if that player could be incentivized to \emph{abstain} entirely: that is, to not improve their model at all. However, weighted winrate cannot incentivize producers to abstain, and we view studying how to incentivize model abstention as a fascinating and fruitful avenue for future work.   

\begin{figure}[htbp]
    \centering

    \begin{subfigure}[t]{0.45\textwidth}
        \centering
        \includegraphics[width=\linewidth]{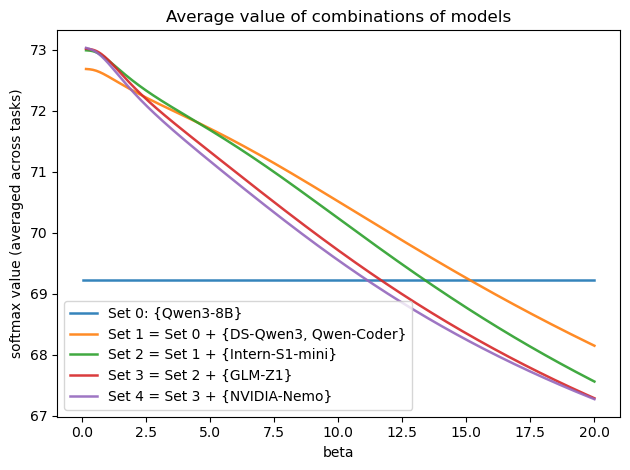}
        \caption{Value of combinations of models.}
        \label{fig:combrouter}
    \end{subfigure}
    \hfill
    \begin{subfigure}[t]{0.45\textwidth}
        \centering
        \includegraphics[width=\linewidth]{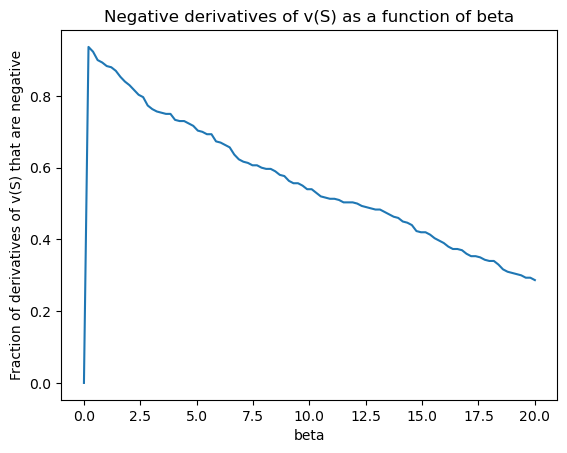}
        \caption{Prevalence of non-monotonicity. 
        }
        \label{fig:derivsrouter}
    \end{subfigure}

    \vspace{0.5em}

    \begin{subfigure}[t]{0.45\textwidth}
        \centering
        \includegraphics[width=\linewidth]{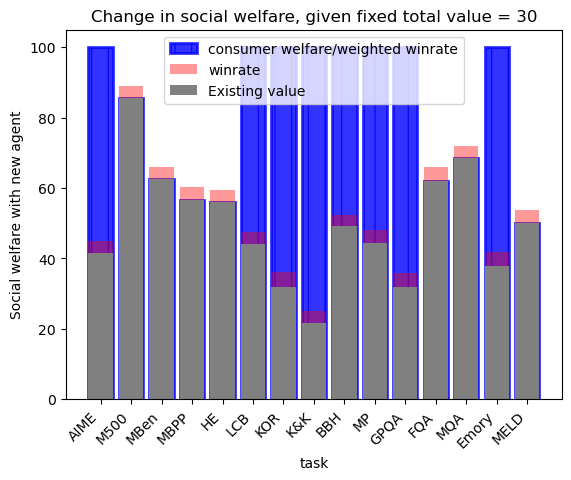}
        \caption{Change in consumer welfare given different incentives for model creation. }
        \label{fig:socialweflarerouter}
    \end{subfigure}
    \hfill
    \begin{subfigure}[t]{0.45\textwidth}
        \centering
        \includegraphics[width=\linewidth]{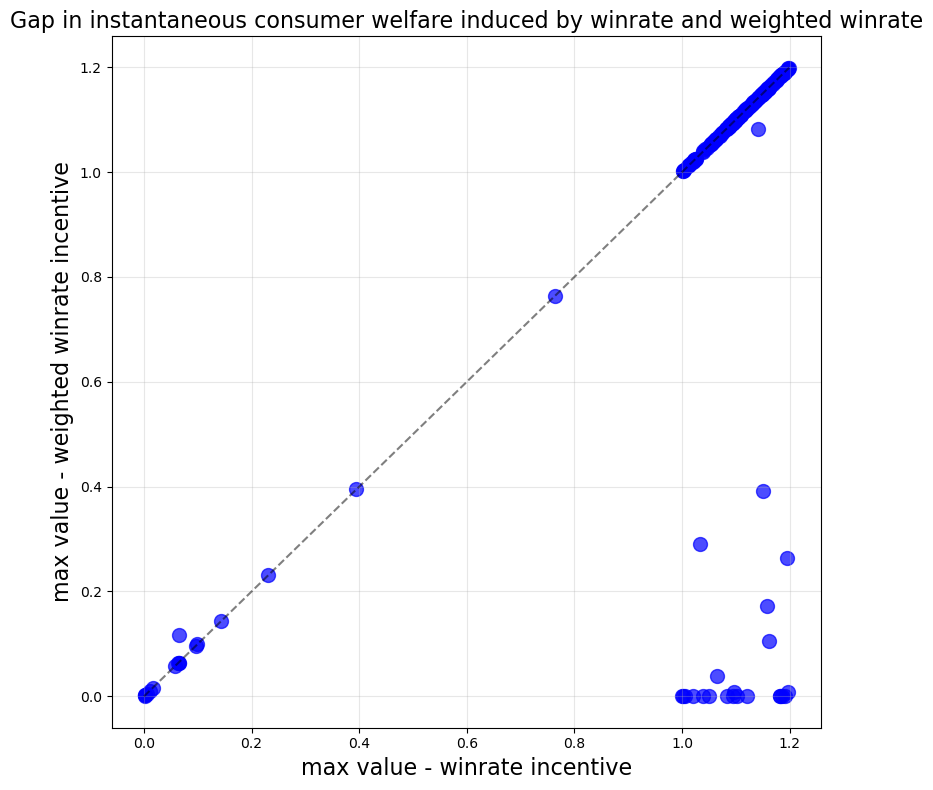}
        \caption{Gap in instantaneous consumer welfare given different incentives in model improvement.}
        \label{fig:instant-gap-router}
    \end{subfigure}

    \caption{Experiments with LLMRouterBench data \cite{li2026llmrouterbench}}
    \label{fig:routerbench_allfigs}
    \vspace{0cm}
\end{figure}

\section{Experiments}\label{sec:experiments} 
Finally, we will explore what our theoretical results imply about real markets of models.
We will use the benchmark LLMRouterbench \citep{li2026llmrouterbench} because it includes real-valued estimates of model performance on discrete tasks with substantial heterogeneity in performance\ifarxiv \else \footnote{Code and data to reproduce them will be provided in a final version.} \fi.  In Appendix \ref{app:experiments} we include experiments on two other datasets (MT-Bench-101 \citep{bai2024mt} and WildBench \citep{lin2024wildbench}).  
\ifarxiv Code and data to reproduce all figures is available at \url{https://github.com/kpdonahue/aggregation}. 
\else \fi

\noindent \textbf{Maximizing consumer welfare (extending Section \ref{sec:helpuser})}: Theorem \ref{thrm:higheraccadd} showed that there are scenarios where adding \enquote{weaker} models in isolation could be beneficial, depending on the noise in BTL aggregation ($\beta$). Figure \ref{fig:combrouter} shows the optimal subset of models as a function of $\beta$: there are settings where a team strictly outperforms the single strongest model alone, and the size of the optimal team decreases as noise increases. Next, we explore the non-monotonicity we studied in Section \ref{sec:helpuser}, where we showed that there can be cases where increasing the value of a model could \emph{decrease} consumer welfare. 
However, our theoretical results left open the question of how \enquote{common} this phenomenon was in real datasets. In Figure \ref{fig:derivsrouter}, we look at every model/task combination and see whether the model is within the non-monotone region: whether improvement by this model would \emph{decrease} consumer welfare. Interestingly, a relatively large proportion of models are in the non-monotone region, though this proportion is decreasing in $\beta$ (as we would expect from Lemma \ref{lem:sufficientmonotone}). 

\noindent \textbf{Incentivizing model creation (extending Section \ref{sec:incententry})} Next, we turn to incentivizing producers to create models that improve consumer welfare. Note that in LLMRouterBench \citep{li2026llmrouterbench} the value of each individual task is bounded by value 100, and thus we apply adapted versions of our results (Lemma \ref{lem:scalinglaw}). We consider an existing market comprised of the 5 weakest models and an incoming model producer with average of 53 over all tasks: we chose these parameters to ensure new model is sufficiently powerful, as in Theorem \ref{thrm:optwinb}. Then, we consider the actions that each objective (winrate, weighted winrate, consumer welfare) would incentivize and their impact on consumer welfare. For this setting, the winrate objective would incentivize the producer to equalize winrate across tasks, while weighted winrate and consumer welfare would both incentivize the producer to maximally specialize on 8 tasks (either abstaining or giving value 0 to the others). Figure \ref{fig:socialweflarerouter} visually shows the change in consumer welfare: the prior value of the existing marketplace is in {\color{gray}gray}, change in value induced by winrate is in {\color{red}red}, and change in value induced by consumer welfare or weighted winrate incentives are in {\color{blue}blue}, with overlapping regions in {\color{purple}purple}. The high-level takeaway from this figure is that weighted winrate incentives the creation of models that lead to much larger benefit to consumers: the blue region is much larger than red.

\noindent \textbf{Incentivizing model replacement (extending Section \ref{sec:incentretrain})}
In Section \ref{sec:incentretrain} we studied model replacement, where the goal was to incentivize model creators to pick tasks to improve in that increase consumer welfare. 
Theorem \ref{thrm:bothneg} showed that weighted winrate improves incentives for consumer welfare, so long as it is \emph{not} the case that one of the two producers is in the non-monotone regime on both tasks. To explore the incentives for model replacement with this dataset, we take every pair of models and calculate tasks they would be incentivized to improve in, given different objectives (winrate, weighted winrate, consumer welfare). In Figure \ref{fig:instant-gap-router}, we display how far from optimal (social welfare) each of weighted winrate and winrate incentives are. The $x$-axis gives the gap between winrate and optimal, while the $y$ axis gives the gap between weighted winrate and optimal: points are along the diagonal if they give equal incentives, and below (resp. above) the diagonal if weighted winrate is better (resp. worse). Note that weighted winrate gives better incentives \emph{almost everywhere}: almost every point is along or below the diagonal. This suggests that empirical performance of weighted winrate is stronger than our conservative theoretical guarantees. 
\section{Future directions}\label{sec:conclude}
\ifarxiv 
\paragraph{Our contributions:} In this work, we have studied a setting where
users derive utility from \emph{aggregating} over multiple models, capturing the
idea that consumers may noisily select a model based on their expectations of
performance on a given task.
We showed that under aggregation, the relationship between agent quality and
consumer welfare is surprisingly non-trivial: for a broad range of choice functions, there are cases where
increasing the quality of any individual agent
could \emph{decrease} consumer welfare. Despite this complex interaction, we
studied how aggregation can lead to a wedge between producer incentives and
consumer welfare: producers may focus too much on competing for attention in the
same tasks, whereas consumers benefit more from specialization. We proposed a
novel mechanism, weighted winrate, and showed that it better aligns producer
objectives with consumer welfare.

\noindent \textbf{Future directions:} There are many natural directions for extensions
in this space. To start, one remaining gap in our solution was incentivizing
producers to \emph{abstain} (fail to give a response). With weighted
winrate, producers are indifferent between returning a response of value 0
and failing to respond at all, but when users are noisy,
returning a response with low or 0 value can \emph{hurt} their utility.
Mechanisms that nudge producers to abstain could help, but these
interventions must be done delicately: if multiple agents abstain, it
could lead to correlated failures and even worse utility for users. 

Another direction could handle complexities of modeling both user and producer
behavior in this case. For example, it may be the case that certain tasks are
inherently harder than others, or improvement is non-linear in effort (i.e.,
more complex scaling laws). It may also be the case that value is heterogeneous
or stochastic even within tasks, which could make implementation of weighted
winrate more challenging. Additional questions around implementation center on
acquiring accurate measurements of values: winrate was selected as a metric
because users are often inaccurate at providing real-valued measures of
outcomes. If user-provided measures of value are sufficiently expensive or
noisy, it may be that other mechanisms could improve upon weighted winrate.
Finally, our results generally incentivize producers to specialize. This
improves consumer welfare when they can consider a large set of candidate agents,
but in the case that users are bandwidth-limited, such specialization may no
longer be desirable. 
\else 
There are many natural directions for extensions of this work. To start, one direction is incentivizing
producers to \emph{abstain} (fail to give a response). With weighted
winrate, producers are indifferent between returning a response of value 0
and failing to respond at all, but when users are noisy,
returning a response with low or 0 value can \emph{hurt} their utility.
Mechanisms that nudge producers to abstain could help, but these
interventions must be done delicately: if multiple agents abstain, it
could lead to correlated failures and even worse utility for users. Another direction could handle complexities of modeling both user and producer
behavior in this case. For example, it may be the case that certain tasks are
inherently harder than others, or improvement is non-linear in effort (i.e.,
more complex scaling laws). 
Additional questions around implementation center on
acquiring accurate measurements of values: winrate was selected as a metric
because users are often inaccurate at providing real-valued measures of
outcomes.
If user-provided measures of value are sufficiently expensive or
noisy, it may be that other mechanisms could improve upon weighted winrate. 
Finally, our results generally incentivize producers to specialize. This
improves consumer welfare when they can consider a large set of candidate agents,
but in the case that users are bandwidth-limited (such as by prices of models), such specialization may no
longer be desirable. 
\fi

\ifarxiv
\subsubsection*{Acknowledgments}
This work was supported in part by a METEOR Postdoctoral Fellowship from MIT. We are extremely grateful to Nivasini Ananthakrishnan, Meena Jagadeesan, Brendan Lucier, Ariel Procaccia, Han Shao, Kiran Tomlinson, Serena Wang, Ming Yin, members of Bailey Flanigan \& Manish Raghavan's lab meeting, and members of Sanmi Koyejo's AI Measurement Science reading group for invaluable discussions, and to Rachel Li especially for assistance with experiments in Appendix \ref{app:errorpatternsrouter}. 

\else 

\fi 

\bibliographystyle{ACM-Reference-Format}
\bibliography{sample-bibliography}

@String{Computer = "{IEEE} Computer" }

@String{Springer = "Springer-Verlag" }

@article{cerreia2019deliberately,
  title={Deliberately stochastic},
  author={Cerreia-Vioglio, Simone and Dillenberger, David and Ortoleva, Pietro and Riella, Gil},
  journal={American Economic Review},
  volume={109},
  number={7},
  pages={2425--2445},
  year={2019},
  publisher={American Economic Association 2014 Broadway, Suite 305, Nashville, TN 37203}
}

@inproceedings{zhang2025beyond,
  title={Beyond gpt-5: Making llms cheaper and better via performance-efficiency optimized routing},
  author={Zhang, Yiqun and Li, Hao and Chen, Jianhao and Zhang, Hangfan and Ye, Peng and Bai, Lei and Hu, Shuyue},
  booktitle={Proceedings of the 2025 7th International Conference on Distributed Artificial Intelligence},
  pages={122--129},
  year={2025}
}

@article{caliari2025luce,
  title={The Luce Model, Regularity, and Choice Overload},
  author={Caliari, Daniele and Petri, Henrik},
  journal={arXiv preprint arXiv:2502.21063},
  year={2025}
}

@misc{lin2024wildbench,
    title={WildBench: Benchmarking LLMs with Challenging Tasks from Real Users in the Wild},
    author={Bill Yuchen Lin and Yuntian Deng and Khyathi Chandu and Faeze Brahman and Abhilasha Ravichander and Valentina Pyatkin and Nouha Dziri and Ronan Le Bras and Yejin Choi},
    year={2024},
    eprint={2406.04770},
    archivePrefix={arXiv},
    primaryClass={cs.CL},
    url={https://arxiv.org/abs/2406.04770}
}

@article{li2026llmrouterbench,
  title={LLMRouterBench: A Massive Benchmark and Unified Framework for LLM Routing},
  author={Li, Hao and Zhang, Yiqun and Guo, Zhaoyan and Wang, Chenxu and Tang, Shengji and Zhang, Qiaosheng and Chen, Yang and Qi, Biqing and Ye, Peng and Bai, Lei and others},
  journal={arXiv preprint arXiv:2601.07206},
  year={2026}
}

@article{hong2004groups,
  title={Groups of diverse problem solvers can outperform groups of high-ability problem solvers},
  author={Hong, Lu and Page, Scott E},
  journal={Proceedings of the National Academy of Sciences},
  volume={101},
  number={46},
  pages={16385--16389},
  year={2004},
  publisher={National Academy of Sciences}
}

@article{plackett1975analysis,
  title={The analysis of permutations},
  author={Plackett, Robin L},
  journal={Journal of the Royal Statistical Society Series C: Applied Statistics},
  volume={24},
  number={2},
  pages={193--202},
  year={1975},
  publisher={Oxford University Press}
}

@article{bradley1952rank,
  title={Rank analysis of incomplete block designs: I. the method of paired comparisons},
  author={Bradley, Ralph Allan and Terry, Milton E},
  journal={Biometrika},
  volume={39},
  number={3/4},
  pages={324--345},
  year={1952},
  publisher={JSTOR}
}

@Inbook{Tullock2001,
author="Tullock, Gordon",
editor="Lockard, Alan A.
and Tullock, Gordon",
title="Efficient Rent Seeking",
bookTitle="Efficient Rent-Seeking: Chronicle of an Intellectual Quagmire",
year="2001",
publisher="Springer US",
address="Boston, MA",
pages="3--16",
isbn="978-1-4757-5055-3",
doi="10.1007/978-1-4757-5055-3_2",
url="https://doi.org/10.1007/978-1-4757-5055-3_2"
}

@article{ramachandran2017searching,
  title={Searching for activation functions},
  author={Ramachandran, Prajit and Zoph, Barret and Le, Quoc V},
  journal={arXiv preprint arXiv:1710.05941},
  year={2017}
}

@inproceedings{lmarena,
  author       = {Wei{-}Lin Chiang and
                  Lianmin Zheng and
                  Ying Sheng and
                  Anastasios Nikolas Angelopoulos and
                  Tianle Li and
                  Dacheng Li and
                  Banghua Zhu and
                  Hao Zhang and
                  Michael I. Jordan and
                  Joseph E. Gonzalez and
                  Ion Stoica},
  title        = {Chatbot Arena: An Open Platform for Evaluating LLMs by Human Preference},
  booktitle    = {Forty-first International Conference on Machine Learning, {ICML} 2024,
                  Vienna, Austria, July 21-27, 2024},
  year         = {2024},
}

@article{lepelley2018monotonicity,
  title={Monotonicity paradoxes in three-candidate elections using scoring elimination rules},
  author={Lepelley, Dominique and Moyouwou, Issofa and Smaoui, Hatem},
  journal={Social Choice and Welfare},
  volume={50},
  number={1},
  pages={1--33},
  year={2018},
  publisher={Springer}
}

@article{borel1921theorie,
  title={La th{\'e}orie du jeu et les {\'e}quations int{\'e}gralesa noyau sym{\'e}trique},
  author={Borel, Emile},
  journal={Comptes rendus de l’Acad{\'e}mie des Sciences},
  volume={173},
  number={1304-1308},
  pages={58},
  year={1921}
}

@book{luce1959individual,
  title={Individual choice behavior},
  author={Luce, R Duncan},
  volume={4},
  year={1959},
  publisher={Wiley New York}
}

@article{haghtalab2020maximizing,
  title={Maximizing welfare with incentive-aware evaluation mechanisms},
  author={Haghtalab, Nika and Immorlica, Nicole and Lucier, Brendan and Wang, Jack Z},
  journal={arXiv preprint arXiv:2011.01956},
  year={2020}
}

@inproceedings{liu2022strategic,
  title={Strategic ranking},
  author={Liu, Lydia T and Garg, Nikhil and Borgs, Christian},
  booktitle={International Conference on Artificial Intelligence and Statistics},
  pages={2489--2518},
  year={2022},
  organization={PMLR}
}

@article{hays2026strategic,
  title={Strategic Candidacy in Generative AI Arenas},
  author={Hays, Chris and Li, Rachel and Flanigan, Bailey and Raghavan, Manish},
  journal={arXiv preprint arXiv:2603.26891},
  year={2026}
}

@article{chen2026leaderboard,
  title={Leaderboard Incentives: Model Rankings under Strategic Post-Training},
  author={Chen, Yatong and Zhang, Guanhua and Hardt, Moritz},
  journal={arXiv preprint arXiv:2603.08371},
  year={2026}
}

@inproceedings{hartline2023optimal,
  title={Optimal scoring rules for multi-dimensional effort},
  author={Hartline, Jason D and Shan, Liren and Li, Yingkai and Wu, Yifan},
  booktitle={The Thirty Sixth Annual Conference on Learning Theory},
  pages={2624--2650},
  year={2023},
  organization={PMLR}
}

@inproceedings{bai2024mt,
  title={MT-Bench-101: A Fine-Grained Benchmark for Evaluating Large Language Models in Multi-Turn Dialogues},
  author={Bai, Ge and Liu, Jie and Bu, Xingyuan and He, Yancheng and Liu, Jiaheng and Zhou, Zhanhui and Lin, Zhuoran and Su, Wenbo and Ge, Tiezheng and Zheng, Bo and others},
  booktitle={Proceedings of the 62nd Annual Meeting of the Association for Computational Linguistics (Volume 1: Long Papers)},
  pages={7421--7454},
  year={2024}
}

@article{small1981applied,
  title={Applied welfare economics with discrete choice models},
  author={Small, Kenneth A and Rosen, Harvey S},
  journal={Econometrica: Journal of the Econometric Society},
  pages={105--130},
  year={1981},
  publisher={JSTOR}
}

@article{rastegari2011revenue,
  title={Revenue monotonicity in deterministic, dominant-strategy combinatorial auctions},
  author={Rastegari, Baharak and Condon, Anne and Leyton-Brown, Kevin},
  journal={Artificial Intelligence},
  volume={175},
  number={2},
  pages={441--456},
  year={2011},
  publisher={Elsevier}
}

@article{xu2025heterogeneous,
  title={Heterogeneous Data Game: Characterizing the Model Competition Across Multiple Data Sources},
  author={Xu, Renzhe and Wang, Kang and Li, Bo},
  journal={arXiv preprint arXiv:2505.07688},
  year={2025}
}

@article{einav2025market,
  title={A market for accuracy: Classification under competition},
  author={Einav, Ohad and Rosenfeld, Nir},
  journal={arXiv preprint arXiv:2502.18052},
  year={2025}
}

@article{jagadeesan2023improved,
  title={Improved bayes risk can yield reduced social welfare under competition},
  author={Jagadeesan, Meena and Jordan, Michael and Steinhardt, Jacob and Haghtalab, Nika},
  journal={Advances in Neural Information Processing Systems},
  volume={36},
  pages={66940--66952},
  year={2023}
}

@article{bedi2025optimization,
  title={The Optimization Paradox in Clinical AI Multi-Agent Systems},
  author={Bedi, Suhana and Mlauzi, Iddah and Shin, Daniel and Koyejo, Sanmi and Shah, Nigam H},
  journal={arXiv preprint arXiv:2506.06574},
  year={2025}
}

@article{jiang2014diverse,
  title={Diverse randomized agents vote to win},
  author={Jiang, Albert X and Marcolino, Leandro S and Procaccia, Ariel D and Sandholm, Tuomas and Shah, Nisarg and Tambe, Milind},
  journal={Advances in Neural Information Processing Systems},
  volume={27},
  year={2014}
}

@article{ai2025beyond,
  title={Beyond majority voting: LLM aggregation by leveraging higher-order information},
  author={Ai, Rui and Pan, Yuqi and Simchi-Levi, David and Tambe, Milind and Xu, Haifeng},
  journal={arXiv preprint arXiv:2510.01499},
  year={2025}
}

@article{lu2024does,
  title={Does more advice help? the effects of second opinions in AI-assisted decision making},
  author={Lu, Zhuoran and Wang, Dakuo and Yin, Ming},
  journal={Proceedings of the ACM on Human-Computer Interaction},
  volume={8},
  number={CSCW1},
  pages={1--31},
  year={2024},
  publisher={ACM New York, NY, USA}
}

@inproceedings{kearns2026networked,
  title={Networked Information Aggregation via Machine Learning},
  author={Kearns, Michael and Roth, Aaron and Ryu, Emily},
  booktitle={Proceedings of the 2026 Annual ACM-SIAM Symposium on Discrete Algorithms (SODA)},
  pages={4799--4845},
  year={2026},
  organization={SIAM}
}

@inproceedings{song2025human,
  title={Human-AI Collaboration with Misaligned Preferences},
  author={Song, Jiaxin and Shahkar, Parnian and Donahue, Kate and Chaudhury, Bhaskar Ray},
  booktitle={Proceedings of the 5th ACM Conference on Equity and Access in Algorithms, Mechanisms, and Optimization},
  pages={294--294},
  year={2025}
}

@misc{meenanivasini,
      title={Power and Limitations of Aggregation in Compound AI Systems}, 
      author={Nivasini Ananthakrishnan and Meena Jagadeesan},
      year={2026},
      eprint={2602.21556},
      archivePrefix={arXiv},
      primaryClass={cs.AI},
      url={https://arxiv.org/abs/2602.21556}, 
}

@article{friedman1958game,
  title={Game-theory models in the allocation of advertising expenditures},
  author={Friedman, Lawrence},
  journal={Operations research},
  volume={6},
  number={5},
  pages={699--709},
  year={1958},
  publisher={INFORMS}
}

@inproceedings{ethayarajh2020utility,
  title={Utility is in the eye of the user: A critique of NLP leaderboards},
  author={Ethayarajh, Kawin and Jurafsky, Dan},
  booktitle={Proceedings of the 2020 Conference on Empirical Methods in Natural Language Processing (EMNLP)},
  pages={4846--4853},
  year={2020}
}

@article{wang2024mixture,
  title={Mixture-of-agents enhances large language model capabilities},
  author={Wang, Junlin and Wang, Jue and Athiwaratkun, Ben and Zhang, Ce and Zou, James},
  journal={arXiv preprint arXiv:2406.04692},
  year={2024}
}

@article{li2025rethinking,
  title={Rethinking mixture-of-agents: Is mixing different large language models beneficial?},
  author={Li, Wenzhe and Lin, Yong and Xia, Mengzhou and Jin, Chi},
  journal={arXiv preprint arXiv:2502.00674},
  year={2025}
}

@inproceedings{rofin2023vote,
  title={Vote’n’rank: Revision of benchmarking with social choice theory},
  author={Rofin, Mark and Mikhailov, Vladislav and Florinsky, Mikhail and Kravchenko, Andrey and Shavrina, Tatiana and Tutubalina, Elena and Karabekyan, Daniel and Artemova, Ekaterina},
  booktitle={Proceedings of the 17th Conference of the European Chapter of the Association for Computational Linguistics},
  pages={670--686},
  year={2023}
}

@article{colombo2022best,
  title={What are the best systems? new perspectives on nlp benchmarking},
  author={Colombo, Pierre and Noiry, Nathan and Irurozki, Ekhine and Cl{\'e}men{\c{c}}on, St{\'e}phan},
  journal={Advances in neural information processing systems},
  volume={35},
  pages={26915--26932},
  year={2022}
}

@inproceedings{procacciametritocracy,
  title={Metritocracy: Representative Metrics for Lite Benchmarks},
  author={Procaccia, Ariel D and Schiffer, Benjamin and Wang, Serena Lutong and Zhang, Shirley},
  booktitle={The Thirty-ninth Annual Conference on Neural Information Processing Systems}
}

@article{gordienko2026beyond,
  title={Beyond Arrow: From Impossibility to Possibilities in Multi-Criteria Benchmarking},
  author={Gordienko, Polina and Jansen, Christoph and Rodemann, Julian and Schollmeyer, Georg},
  journal={arXiv preprint arXiv:2602.07593},
  year={2026}
}

@inproceedings{weimarket,
  title={Market Games for Generative Models: Equilibria, Welfare, and Strategic Entry},
  author={Wei, Xiukun and Shi, Min and Zhang, Xueru},
  booktitle={The Fourteenth International Conference on Learning Representations}
}

@inproceedings{zhang2024inherent,
  title={Inherent trade-offs between diversity and stability in multi-task benchmarks},
  author={Zhang, Guanhua and Hardt, Moritz},
  booktitle={Proceedings of the 41st International Conference on Machine Learning},
  pages={58984--59002},
  year={2024}
}

@misc{arena_blog_max_2026,
  author = {{Arena Team}},
  title = {Introducing {M}ax},
  howpublished = {\url{https://arena.ai/blog/introducing-max/}},
  year = {2026},
  month = {February},
  note = {Accessed: 2026-02-09}
}

@article{freund1997decision,
  title={A decision-theoretic generalization of on-line learning and an application to boosting},
  author={Freund, Yoav and Schapire, Robert E},
  journal={Journal of computer and system sciences},
  volume={55},
  number={1},
  pages={119--139},
  year={1997},
  publisher={Elsevier}
}

@article{kleinberg2021algorithmic,
  title={Algorithmic monoculture and social welfare},
  author={Kleinberg, Jon and Raghavan, Manish},
  journal={Proceedings of the National Academy of Sciences},
  volume={118},
  number={22},
  pages={e2018340118},
  year={2021},
  publisher={National Academy of Sciences}
}

@article{bommasani2022picking,
  title={Picking on the same person: Does algorithmic monoculture lead to outcome homogenization?},
  author={Bommasani, Rishi and Creel, Kathleen A and Kumar, Ananya and Jurafsky, Dan and Liang, Percy S},
  journal={Advances in Neural Information Processing Systems},
  volume={35},
  pages={3663--3678},
  year={2022}
}

@article{kim2025correlated,
  title={Correlated Errors in Large Language Models},
  author={Kim, Elliot and Garg, Avi and Peng, Kenny and Garg, Nikhil},
  journal={arXiv preprint arXiv:2506.07962},
  year={2025}
}

@article{goel2025great,
  title={Great models think alike and this undermines ai oversight},
  author={Goel, Shashwat and Struber, Joschka and Auzina, Ilze Amanda and Chandra, Karuna K and Kumaraguru, Ponnurangam and Kiela, Douwe and Prabhu, Ameya and Bethge, Matthias and Geiping, Jonas},
  journal={arXiv preprint arXiv:2502.04313},
  year={2025}
}

\clearpage 
\appendix

\section{Additional results on model}
\ifarxiv 
\else 
\subsection{Additional related work}\label{app:related}
\paragraph{Non-monotonicity in aggregation:} In Section \ref{sec:helpuser} we demonstrate a \enquote{non-monotonicity} result: increasing the value of a model could \emph{decrease} consumer welfare. Related non-monotonicities have been studied in other settings. 
For example, one close connection is in auctions, where results have shown that adding in more bidders could \emph{decrease} revenue (e.g. \citep{rastegari2011revenue}). While this phenomenon is very closely related, the types of aggregation mechanisms and objectives differ between our settings and auctions. Related non-monotonicities also occur in social choices, such as settings where it could be the case that \emph{more} votes for a candidate could cause her to go from winning a race to losing (e.g. \citep{lepelley2018monotonicity}). This could reflect violations of Pareto efficiency or the unanimity principle, which would declare that more support for an option should not \emph{reduce} its property of winning. One close connection is \cite{small1981applied}, which shows that consumer welfare is monotone with a different utility model where users derive utility according to the value plus the additive $\epsilon$  noise: however, in our setting, we argue that it makes more sense to use the noise terms $\epsilon_i$ to model decision errors, rather than true variation in quality that a user perfectly responds to. Our results differ from most voting settings in that we assume there is some true value over candidates (here, model responses) and we are studying the \enquote{expected value} of the election, while most voting mechanisms are candidate-neutral and focus on the probability of individual candidates being selected.

We note that our \enquote{monotonicity} condition is \emph{not} the same as \enquote{regularity} (e.g. \cite{cerreia2019deliberately, caliari2025luce}), a different property often studied in relation to choice models. Regularity is defined as saying that the probability of an item being picked only stays constant or decreases as more items are added (we refer to this property as \enquote{substitutability}, in Definition \ref{def:substitutability}). A choice function could exhibit monotonicity or regularity, or neither, or both: regularity is a property of the probability of an item being picked (given other items added to the set), and monotonicity is a property of the value that the choice function gives the user (as items in the set increase in value). 

\textbf{Connection to Colonel Blotto, Tullock contests:}
Note that the \enquote{model-creation} setting in Section \ref{sec:incententry} has close ties to a rich vein of literature for how agents compete,
such as Colonel Blotto and Tullock contests. The Colonel Blotto game models a pair of Colonels competing over a series of \enquote{battlefields}, with some fixed total amount of value, and an agent \enquote{winning} a battlefield according to some cost function based on how strongly they out-compete the other Colonel \citep{borel1921theorie}. Tullock contests similarly study different agents who allocate different value in the hope of winning a probabilistic \emph{contest}, where the aggregation function is typically given by $\frac{\val_a}{\val_a + \val_b}$ \citep{Tullock2001}. The closest variant of Colonel Blotto is the \enquote{Lottery Blotto} game \cite{friedman1958game}, which studies Colonel Blotto where the outcome is determined by a Tullock contest function. Results here show that the unique equilibrium is to have players have allocation that is equal across tasks (if all tasks are equally valuable). While these literatures differ from our setting in multiple ways, one crucial one is the function governing \enquote{winrate.} Specifically, winrate is often either deterministic (by whichever agent has the larger value), or when probabilistic is often given by convex functions, such as $\frac{\val_{ai}}{\val_{ai} + \val_{bi}}$. Given a function like this, agents are generally incentivized to best-respond by evenly spreading value across tasks. 
However, in our setting we use BTL to determine winrate, which is \emph{not}
convex. This leads to much more complex patterns of behavior. We view this complexity as a benefit, given that it arises from satisfying natural properties like users becoming more easily \enquote{distracted} by items that are similar in value, discussed in Section \ref{sec:model}.

\fi 
\begin{figure}
\centering 
    \includegraphics[width=5in]{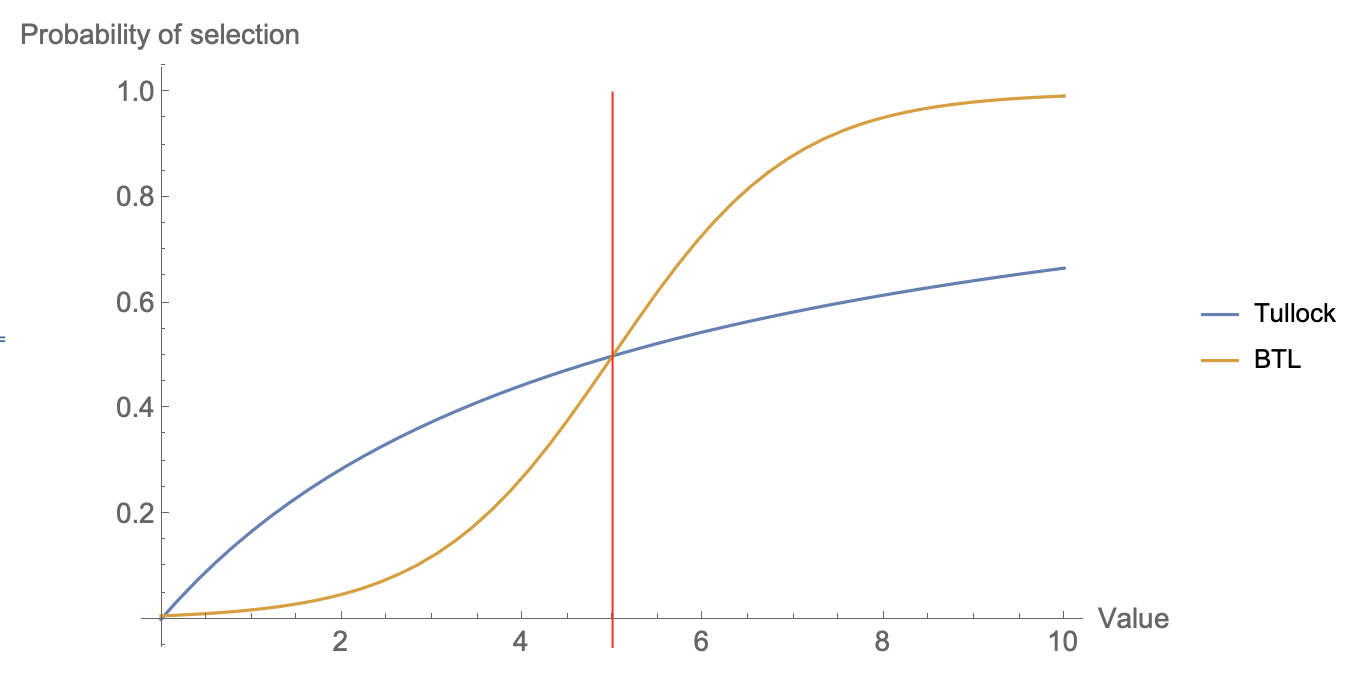}
    \caption{Given a setting with two players with values $v_i, v_j$, the probability of selection of $v_i$, given $v_j=5$ (red vertical line) for Tullock contest function $\frac{v_i}{v_i + v_j}$ and BTL $\frac{\exp(v_i/\beta)}{\exp(v_i/\beta) + \exp(v_j/\beta)}$ with $\beta=1$. Note that the probability for BTL changes most quickly around $v_i \approx v_j$. }
    \label{fig:tullockexample}
\end{figure}

\subsection{Error patterns in routing}\label{app:errorpatternsrouter}
In this section we include further explorations of empirical patterns in \enquote{errors} in LLM routing. As we noted in Section \ref{sec:related}, LLM routing is an active area of research, and many contain substantially different approaches. One key difference is that the source of randomness in LLM routing is different. When humans make choices, the randomness in their decision often comes from inherent uncertainty or variability in human behavior, or \enquote{irrelevant} features. However, many LLM routing tools are deterministic at the instance-level decision of where to route any particular query, given a fixed algorithm, but have randomness involved in the training process of the algorithm. In this subsection we will explore how such randomness relates to the model in Section \ref{sec:model} and our paper's general implications. 

Specifically, we recreate the results in \cite{zhang2025beyond}, which is a clustering-based routing method. In this approach, the router is trained on multiple benchmark datasets. A randomized 70\%/30\% split of data is used for training and test performance. In the training process, the data is grouped into 64 clusters. In the inference stage, for each new query the routing algorithm calculates which of the 64 clusters the query is closest to in embedding space, and then deterministically routes the query to whichever model performs best on that cluster. 

In Table \ref{tab:routercostperf} we include the empirical performance of each model on each benchmark (as calculated in our 70\% split), and in Table \ref{tab:routerprob} we return the proportion of test queries (in the 30\% split) from that dataset that were routed to each model. Recall that because items are rated based on clusters that imperfectly overlap with benchmarks, the routing tool will \emph{not} always route the query from a particular benchmark. For example, the routing tool maps queries from the AIME dataset to \verb+deepseek-r1-0528+ always, even though its training performance on this dataset is only tied for second best (\verb+a22b-thinking-2507+ has slightly higher performance, 0.93 as compared to 0.9). However, there is a general connection between strength on a benchmark and frequency with which that model is used on that benchmark. 

Figure \ref{fig:routerbench-ex-gapvsperf} presents the connections between these. In general, the routing tool is more likely to make \enquote{mistakes} in routing between two models that have very similar performance. When a model has substantially lower performance than the best model on a task, its probability of being routed to drops substantially. Note that this pattern of errors does not exactly match those given by BTL - which is at least partially due to fact that this experiment only captures some source of randomness in routing and not others (such as the randomness in the test/train split). However, from this experiment, routers do seem to display many of the key features we referenced in Section \ref{sec:model}: they are better than random but worse than perfect, and tend to more frequently make errors between options that are closer in value. Given this, we suspect that the high-level insights from our paper could also apply to routing tools. 

\begin{figure}
\centering 
    \includegraphics[width=5in]{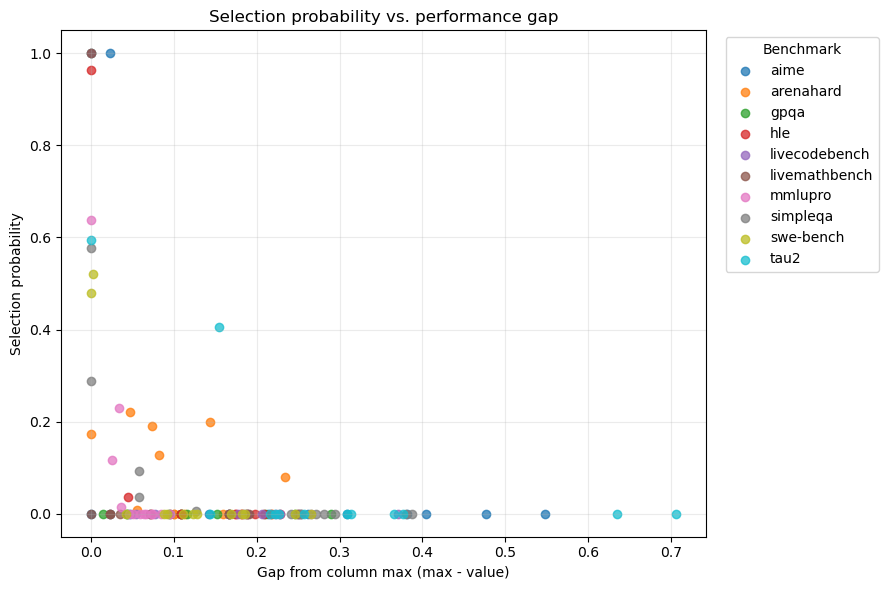}
    \caption{For each benchmark and each model, the $x$-axis gives the gap between that model's performance on that benchmark (e.g. a gap of 0 means that model is the best), and the $y$ axis measures how frequently that model is selected by the router for that benchmark. }
    \label{fig:routerbench-ex-gapvsperf}
\end{figure}

\begin{table}
\begin{tabular}{C{2cm}C{1.cm} C{1.cm} C{1.cm} C{1.cm} C{1.cm} C{1.cm} C{1.cm} C{1.cm} C{1.cm} C{1.cm}}
\toprule
model & aime & arenahard & gpqa & hle & livecodebench & livemathbench & mmlupro & simpleqa & swe-bench & tau2 \\
\midrule
claude-sonnet-4 & 0.38 & 0.53 & 0.66 & 0.05 & 0.61 & 0.63 & 0.83 & 0.16 & 0.34 & 0.48 \\
deepseek-r1-0528 & 0.9 & 0.62 & 0.8 & 0.16 & 0.82 & 0.8 & 0.84 & 0.26 & 0.3 & 0.39 \\
deepseek-v3-0324 & 0.45 & 0.61 & 0.56 & 0.04 & 0.69 & 0.63 & 0.78 & 0.28 & 0.25 & 0.07 \\
deepseek-v3.1-terminus & 0.55 & 0.67 & 0.75 & 0.09 & 0.68 & 0.76 & 0.84 & 0.24 & 0.25 & 0.49 \\
gemini-2.5-flash & 0.67 & 0.55 & 0.7 & 0.08 & 0.62 & 0.77 & 0.81 & 0.3 & 0.17 & 0.34 \\
gemini-2.5-pro & 0.88 & 0.77 & 0.85 & 0.22 & 0.8 & 0.55 & 0.85 & 0.54 & 0.34 & 0.45 \\
glm-4.6 & 0.79 & 0.65 & 0.73 & 0.16 & 0.65 & 0.63 & 0.8 & 0.26 & 0.21 & 0.56 \\
gpt-5 & 0.9 & 0.69 & 0.83 & 0.26 & 0.87 & 0.8 & 0.87 & 0.48 & 0.16 & 0.71 \\
gpt-5-chat & 0.71 & 0.68 & 0.74 & 0.06 & 0.64 & 0.61 & 0.81 & 0.41 & 0.1 & 0 \\
intern-s1 & 0.52 & 0.67 & 0.64 & 0.08 & 0.5 & 0.73 & 0.82 & 0.15 & 0.08 & 0.33 \\
kimi-k2-0905 & 0.62 & 0.7 & 0.74 & 0.07 & 0.7 & 0.7 & 0.8 & 0.28 & 0.23 & 0.48 \\
qwen3-235b-a22b-2507 & 0.76 & 0.71 & 0.6 & 0.09 & 0.67 & 0.77 & 0.82 & 0.54 & 0.16 & 0.4 \\
qwen3-235b-a22b-thinking-2507 & 0.93 & 0.72 & 0.8 & 0.08 & 0.8 & 0.55 & 0.79 & 0.48 & 0.22 & 0.55 \\
\bottomrule
\end{tabular}
\caption{Performance of models on different benchmarks (reproduction of \cite{zhang2025beyond}, 70\% split of data for training). }
\label{tab:routercostperf}
\end{table}

\begin{table}
\begin{tabular}{C{2cm}C{1.cm} C{1.cm} C{1.cm} C{1.cm} C{1.cm} C{1.cm} C{1.cm} C{1.cm} C{1.cm} C{1.cm}}
\toprule
model & aime & arenahard & gpqa & hle & livecodebench & livemathbench & mmlupro & simpleqa & swe-bench & tau2 \\
\midrule
claude-sonnet-4 & 0 & 0.08 & 0 & 0 & 0 & 0 & 0 & 0 & 0.48 & 0 \\
deepseek-r1-0528 & 1 & 0.2 & 0 & 0 & 0 & 1 & 0.23 & 0 & 0 & 0 \\
deepseek-v3-0324 & 0 & 0 & 0 & 0 & 0 & 0 & 0 & 0 & 0 & 0 \\
deepseek-v3.1-terminus & 0 & 0 & 0 & 0 & 0 & 0 & 0.01 & 0 & 0 & 0 \\
gemini-2.5-flash & 0 & 0 & 0 & 0 & 0 & 0 & 0 & 0 & 0 & 0 \\
gemini-2.5-pro & 0 & 0.17 & 1 & 0.04 & 0 & 0 & 0.12 & 0.58 & 0.52 & 0 \\
glm-4.6 & 0 & 0 & 0 & 0 & 0 & 0 & 0 & 0 & 0 & 0 \\
gpt-5 & 0 & 0.19 & 0 & 0.96 & 1 & 0 & 0.64 & 0.09 & 0 & 0.6 \\
gpt-5-chat & 0 & 0.13 & 0 & 0 & 0 & 0 & 0 & 0.01 & 0 & 0 \\
intern-s1 & 0 & 0 & 0 & 0 & 0 & 0 & 0 & 0 & 0 & 0 \\
kimi-k2-0905 & 0 & 0 & 0 & 0 & 0 & 0 & 0 & 0 & 0 & 0 \\
qwen3-235b-a22b-2507 & 0 & 0.01 & 0 & 0 & 0 & 0 & 0 & 0.29 & 0 & 0 \\
qwen3-235b-a22b-thinking-2507 & 0 & 0.22 & 0 & 0 & 0 & 0 & 0 & 0.04 & 0 & 0.4 \\
\bottomrule
\end{tabular}
\caption{Empirical routing proportion (30\% test split), reproduction of results in \cite{zhang2025beyond}.}
\label{tab:routerprob}
\end{table}

\subsection{Experiments with benchmark datasets  (additional plots from Section \ref{sec:experiments})}\label{app:experiments}
Here, we include additional figures for Section \ref{sec:experiments}. 
\subsubsection{LLMRouterBench}
Figure \ref{fig:llmrouterbench} \cite{li2026llmrouterbench} displays the performance of different models on benchmarks, and is the information used to create Figure \ref{fig:routerbench_allfigs}. 
Table \ref{tab:llmrouterbench-allocation} displays the allocation of value across tasks in the experiment in Figure \ref{fig:socialweflarerouter}: firms would respond to objectives (winrate, weighted winrate, consumer welfare) in different ways, which result in the consumer welfare outcomes displayed in Figure \ref{fig:socialweflarerouter}. 

\begin{figure}
    \centering
    \includegraphics[width=0.8\linewidth]{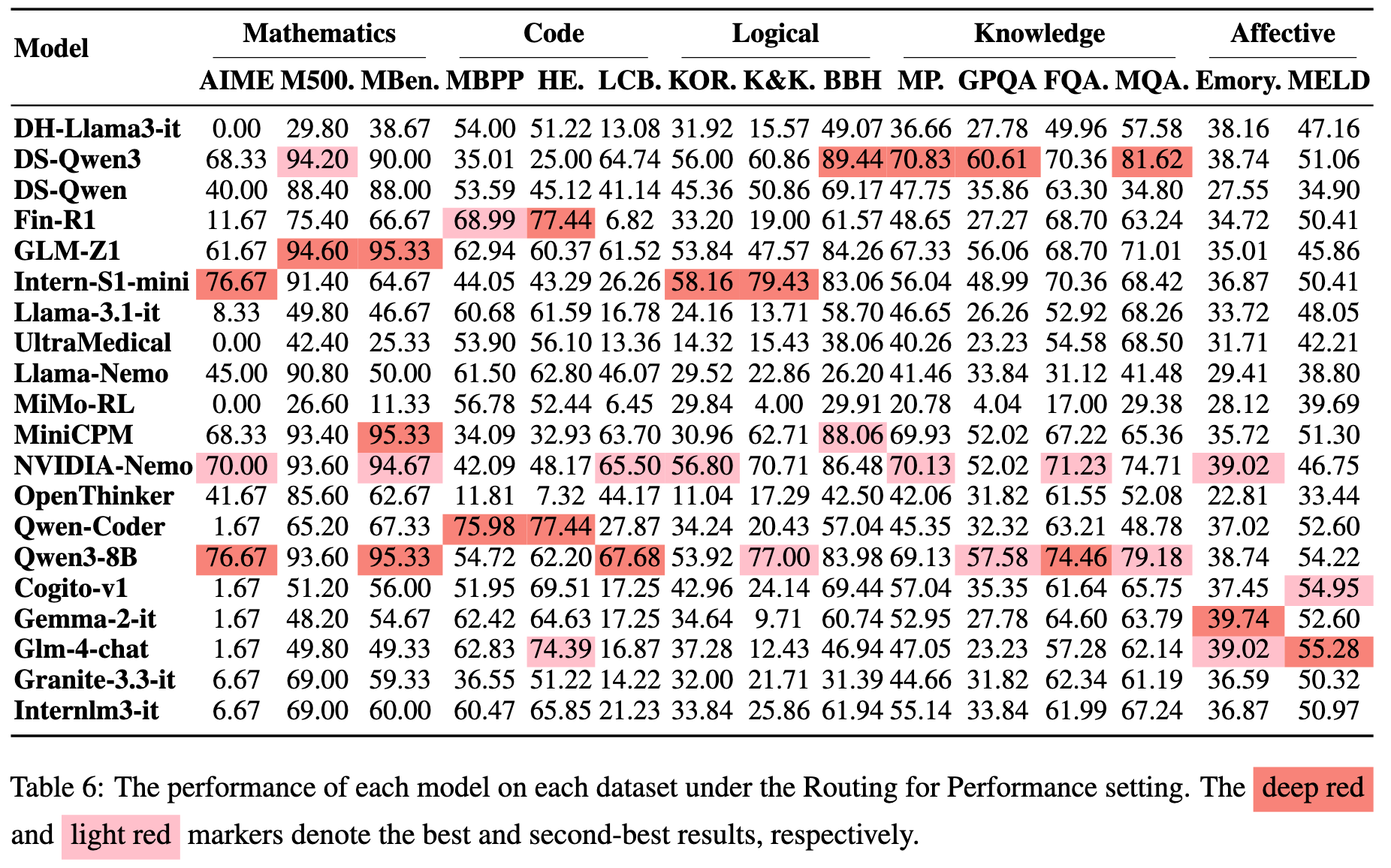}
    \caption{Performance of different models on LLMRouterbench \cite{li2026llmrouterbench}. }
    \label{fig:llmrouterbench}
\end{figure}

\begin{table}[]
\centering 
\begin{tabular}{|c|c|c|c|}
\hline
               & \textbf{Winrate} & \textbf{\begin{tabular}[c]{@{}c@{}}Weighted \\ winrate\end{tabular}} & \textbf{\begin{tabular}[c]{@{}c@{}}Consumer \\ welfare\end{tabular}} \\ \hline
\textbf{AIME}  & 45               & 100                                                                  & 100                                                                  \\ \hline
\textbf{M500}  & 89               & 0                                                                    & $\emptyset$                                                          \\ \hline
\textbf{MBen}  & 66               & 0                                                                    & $\emptyset$                                                          \\ \hline
\textbf{MBPP}  & 60               & 0                                                                    & $\emptyset$                                                          \\ \hline
\textbf{HE}    & 60               & 0                                                                    & $\emptyset$                                                          \\ \hline
\textbf{LCB}   & 48               & 100                                                                  & 100                                                                  \\ \hline
\textbf{KOR}   & 36               & 100                                                                  & 100                                                                  \\ \hline
\textbf{K\&K}  & 25               & 100                                                                  & 100                                                                  \\ \hline
\textbf{BBH}   & 53               & 100                                                                  & 100                                                                  \\ \hline
\textbf{MP}    & 48               & 100                                                                  & 100                                                                  \\ \hline
\textbf{GPQA}  & 36               & 100                                                                  & 100                                                                  \\ \hline
\textbf{FQA}   & 66               & 0                                                                    & $\emptyset$                                                          \\ \hline
\textbf{MQA}   & 72               & 0                                                                    & $\emptyset$                                                          \\ \hline
\textbf{Emory} & 42               & 100                                                                  & 100                                                                  \\ \hline
\textbf{MELD}  & 54               & 0                                                                    & $\emptyset$                                                          \\ \hline
\end{tabular}
\caption{Given a total allocation of value 800 (average 53 across tasks), the allocation that maximizes a) winrate, b) weighted winrate, or c) consumer welfare, with $\beta=1$. Winrate maintains a constant 0.97 winrate (except for in tasks where that isn't possible due to a max value cap of 100), and weighted winrate and consumer welfare are almost identical, except that consumer welfare abstains where weighted winrate returns an item of value 0.}
\label{tab:llmrouterbench-allocation}
\end{table}

\subsubsection{MT-Bench-101}

Figure \ref{fig:mtbench-data} displays MT-Bench-101 \citep{bai2024mt} performance across tasks. Table \ref{tab:max_alloc} gives the allocations that lead to welfare changes in Figure \ref{fig:socialweflaremt}. 

\begin{figure}
    \centering
    \includegraphics[width=0.8\linewidth]{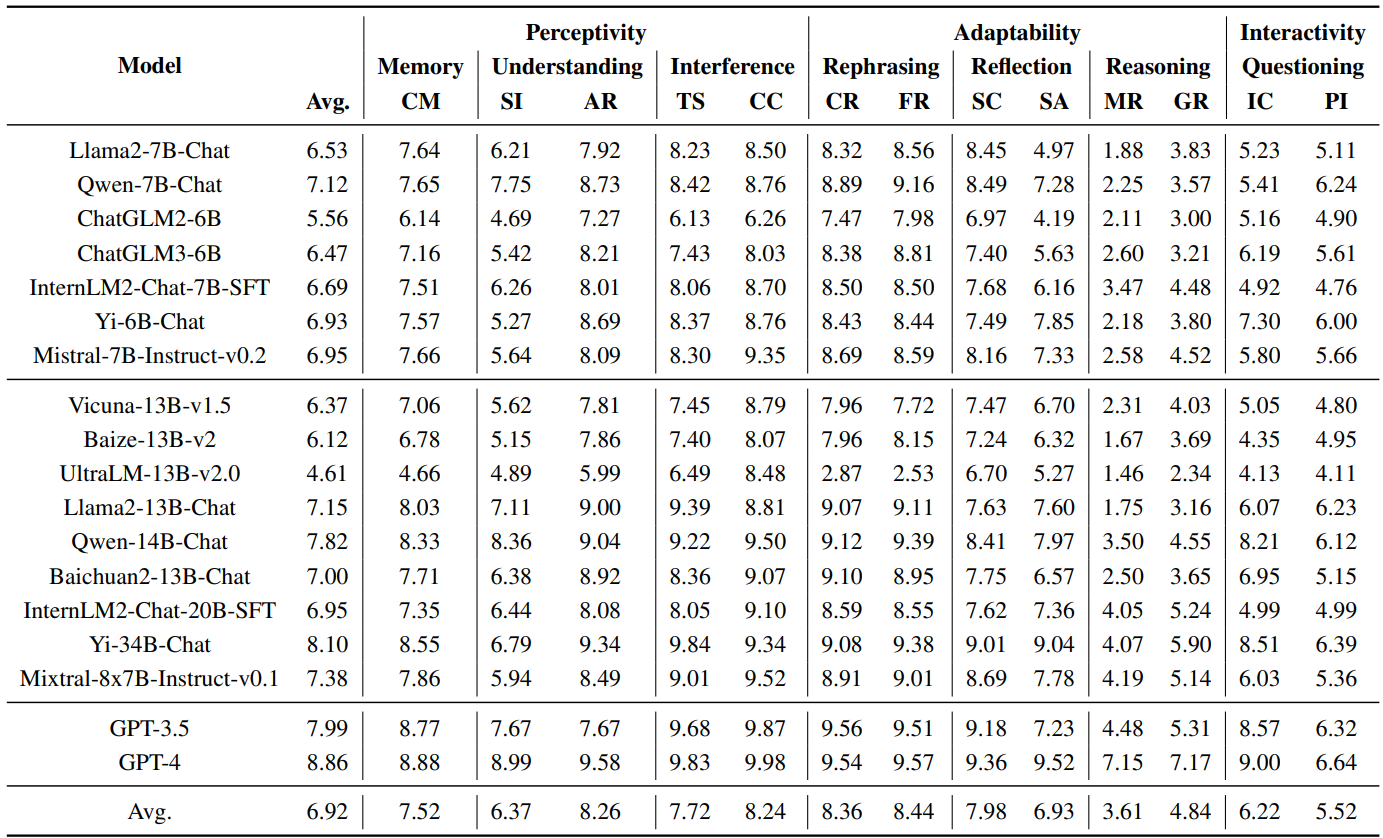}
    \caption{Performance of different models on MT-Bench-101 \citep{bai2024mt}}
    \label{fig:mtbench-data}
\end{figure}

\begin{figure}[htbp]
    \centering

    \begin{subfigure}[t]{0.49\textwidth}
        \centering
        \includegraphics[width=\linewidth]{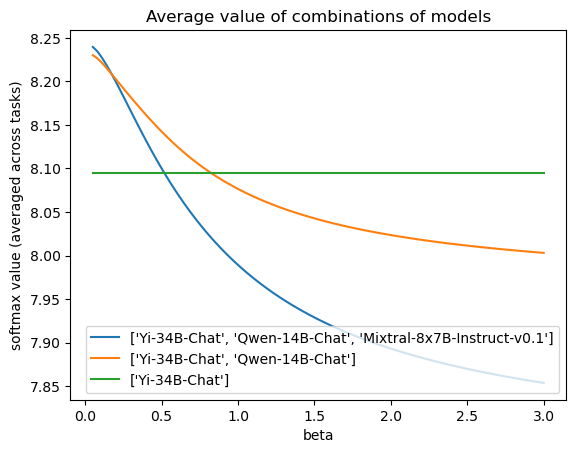}
        \caption{Value of combinations of models.}
        \label{fig:fig:combmt}
    \end{subfigure}
    \hfill
    \begin{subfigure}[t]{0.49\textwidth}
        \centering
        \includegraphics[width=\linewidth]{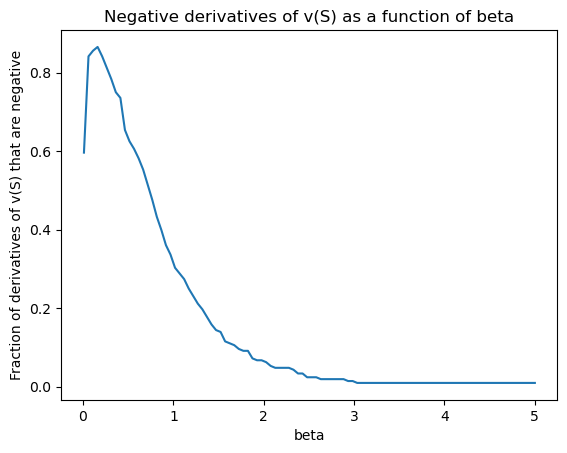}
        \caption{Fraction of derivatives of consumer welfare that are negative.}
        \label{fig:derivsmt}
    \end{subfigure}

    \vspace{0.5em}

    \begin{subfigure}[t]{0.49\textwidth}
        \centering
        \includegraphics[width=\linewidth]{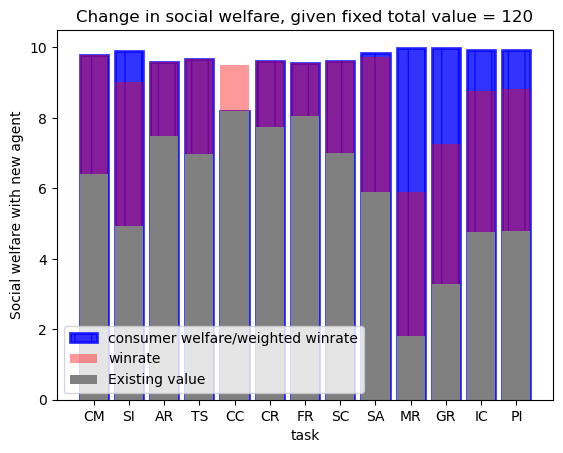}
        \caption{Change in consumer welfare given different incentives for model creation. }
        \label{fig:socialweflaremt}
    \end{subfigure}
    \hfill
    \begin{subfigure}[t]{0.49\textwidth}
        \centering
        \includegraphics[width=\linewidth]{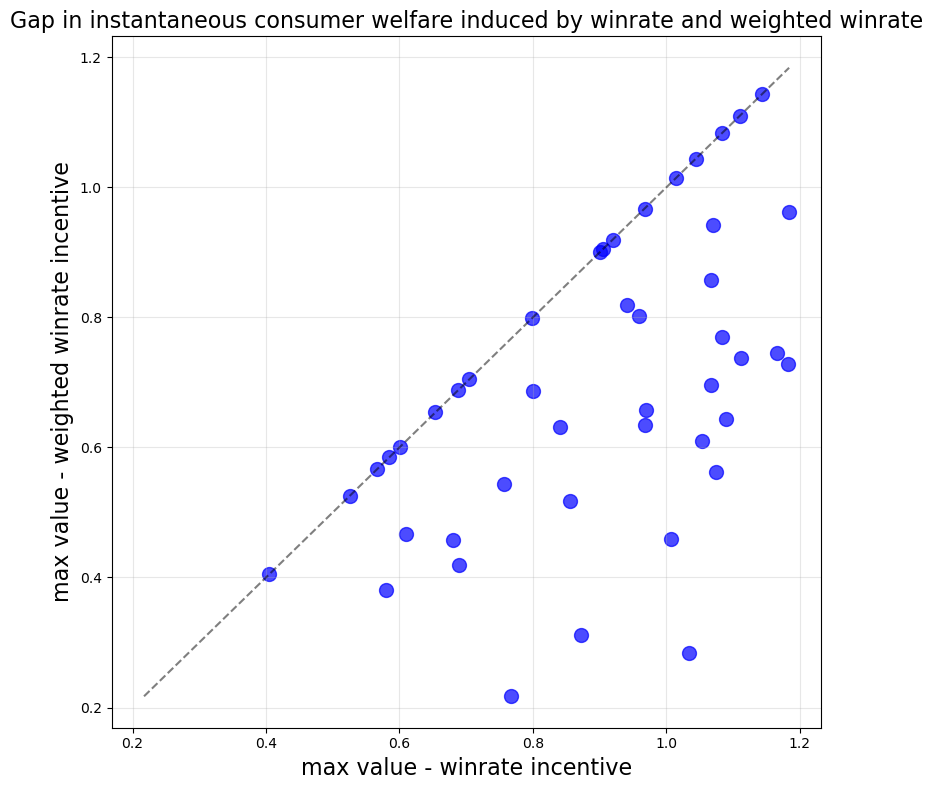}
        \caption{Gap in instantaneous consumer welfare given different incentives in model improvement.}
        \label{fig:instant-gap-mt}
    \end{subfigure}

    \caption{Experiments with MTBench-101 data \cite{bai2024mt}}
    \label{fig:mtbench_allfigs}
\end{figure}

\begin{table}[]
\begin{tabular}{|c|c|c|c|c|c|c|c|c|c|c|c|c|c|}
\hline
                          & \textbf{CM} & \textbf{SI} & \textbf{AR} & \textbf{TS} & \textbf{CC} & \textbf{CR} & \textbf{FR} & \textbf{SC} & \textbf{SA} & \textbf{MR} & \textbf{GR} & \textbf{IC} & \textbf{PI} \\ \hline
\textbf{Winrate}          & 10          & 9.2         & 10          & 10          & 10          & 10          & 10          & 10          & 9.88        & 6.06        & 7.43        & 8.93        & 9           \\ \hline
\textbf{Weighted winrate} & 10          & 10          & 10          & 10          & 0           & 10          & 10          & 10          & 10          & 10          & 10          & 10          & 10          \\ \hline
\textbf{Consumer welfare}     & 10          & 10          & 10          & 10          & $\emptyset$ & 10          & 10          & 10          & 10          & 10          & 10          & 10          & 10          \\ \hline
\end{tabular}
\caption{Given a total allocation of value 120 (average 9.23 across tasks), the allocation that maximizes a) winrate, b) weighted winrate, or c) consumer welfare, with $\beta=1$. Winrate maintains a constant 0.96 winrate (except for in tasks where that isn't possible due to a max value cap of 10), and weighted winrate and consumer welfare are almost identical, except that consumer welfare abstains for the smallest task, while weighted winrate does not. For space constraints, the task names are abbreviated. }
\label{tab:max_alloc}
\end{table}


%

\subsubsection{Wildbench}

Table \ref{tab:wildbenchdata} displays performance of the top 10 models on Wildbench task-specific benchmarks \citep{lin2024wildbench}. Figure \ref{fig:wildbench_allfigs} displays all figures related to experiments on Wildbench data, and  Table \ref{tab:wildbench-allocation} gives the allocations that lead to welfare changes in Figure \ref{fig:socialweflarewild}. 

\begin{table}
\centering 
\begin{tabular}{C{3.5cm}C{2.cm} C{2.cm} C{2.cm} C{2.cm} C{2.cm} }
\toprule
 & Creative Tasks & Coding \& Debugging & Planning \& Reasoning & Information/Advice seeking & Math \& Data Analysis \\
\midrule
Athene-70B & 7.9 & 7.76 & 7.96 & 7.94 & 7.58 \\
gpt-4o-2024-05-13 & 7.78 & 7.91 & 7.89 & 7.73 & 7.59 \\
gpt-4o-mini-2024-07-18 & 7.87 & 7.58 & 7.69 & 7.61 & 7.27 \\
Mistral-Large-2 & 7.75 & 7.25 & 7.59 & 7.6 & 7.13 \\
yi-large-preview & 7.63 & 7.29 & 7.53 & 7.64 & 7.06 \\
gpt-4-turbo-2024-04-09 & 7.73 & 7.37 & 7.48 & 7.58 & 6.96 \\
claude-3-5-sonnet-20240620 & 7.43 & 7.52 & 7.43 & 7.42 & 6.88 \\
deepseek-v2-chat-0628 & 7.51 & 7.36 & 7.35 & 7.14 & 7.01 \\
gemma-2-9b-it-DPO & 7.77 & 6.92 & 7.41 & 7.69 & 6.58 \\
gemma-2-9b-it-SimPO & 7.66 & 6.95 & 7.43 & 7.51 & 6.72 \\
\bottomrule
\end{tabular}
\caption{Performance of models on different tasks in Wildbench benchmark \cite{lin2024wildbench}, retrieved from \url{https://github.com/allenai/WildBench/blob/main/leaderboard/data_dir/score.json}. For concisesness, only top 10 models out of 61 are presented here, but all were used in analysis for Figure \ref{fig:wildbench_allfigs} unless otherwise noted. }
\label{tab:wildbenchdata}
\end{table}

\begin{figure}[htbp]
    \centering

    \begin{subfigure}[t]{0.49\textwidth}
        \centering
        \includegraphics[width=\linewidth]{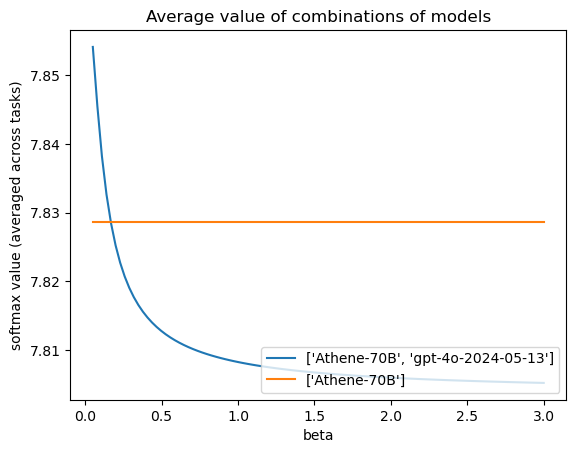}
        \caption{Value of combinations of models.}
        \label{fig:fig:combwild}
    \end{subfigure}
    \hfill
    \begin{subfigure}[t]{0.49\textwidth}
        \centering
        \includegraphics[width=\linewidth]{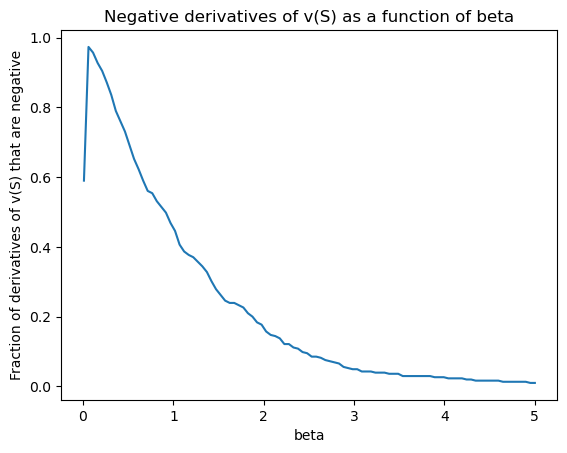}
        \caption{Fraction of derivatives of consumer welfare that are negative.}
        \label{fig:derivswild}
    \end{subfigure}

    \vspace{0.5em}

    \begin{subfigure}[t]{0.49\textwidth}
    \vspace{0pt}
        \centering
        \includegraphics[width=\linewidth]{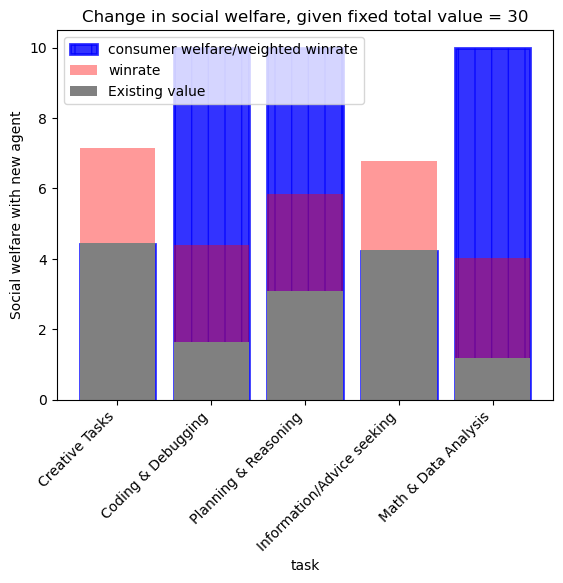}
        \caption{Change in consumer welfare given different incentives for model creation. }
        \label{fig:socialweflarewild}
    \end{subfigure}
    \hfill
    \begin{subfigure}[t]{0.49\textwidth}
    \vspace{0pt}
        \centering
        \includegraphics[width=\linewidth]{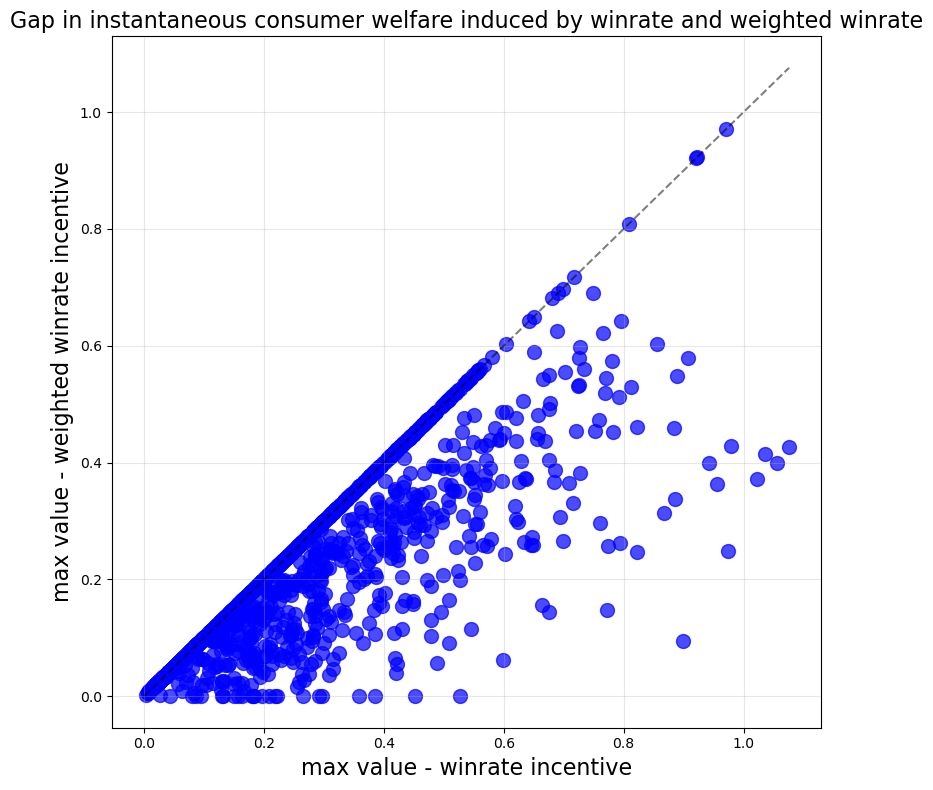}
        \caption{Gap in instantaneous consumer welfare given different incentives in model improvement.}
        \label{fig:instant-gap-wild}
    \end{subfigure}

    \caption{Experiments with Wildbench data \cite{lin2024wildbench}}
    \label{fig:wildbench_allfigs}
\end{figure}

\begin{table}
\begin{tabular}{C{3.5cm}|C{2.cm} |C{2.cm}| C{2.cm} |C{2.5cm}| C{2.cm} }
\hline 
 & Creative Tasks & Coding \& Debugging & Planning \& Reasoning & Information/Advice seeking & Math \& Data Analysis \\
 \hline  
\textbf{Winrate}          & 7.46        & 4.69 & 6.14 & 7.07        & 4.35 \\ \hline
\textbf{Weighted winrate} & 0           & 10   & 10   & 0           & 10   \\ \hline
\textbf{Consumer welfare} & $\emptyset$ & 10   & 10   & $\emptyset$ & 10   \\ \hline
\end{tabular}
\caption{Allocation of values for Figure \ref{fig:socialweflarewild}. }
\label{tab:wildbench-allocation}
\end{table}

\section{Proofs for Section \ref{sec:helpuser}}\label{app:helpuser}

\begin{figure}
\centering 
    \includegraphics[width=3in]{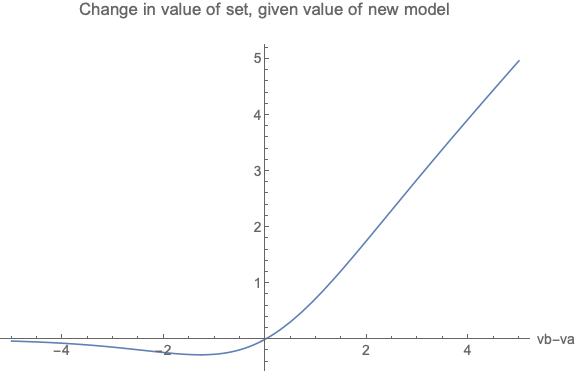}
    \caption{Plot of the function in Lemma \ref{lem:pairwisebenefitexact} for a single task, with $\beta=1$: equivalent to 
    $\frac{\exp(\Delta/\beta)}{\exp(\Delta/\beta) + 1}  \cd \Delta$, for $\Delta = \val_{b} -\val_{a}$. Note that in the example in Lemma \ref{lem:exnonmonotone}, as $\Delta = \val_b - \val_a$ goes from $-4$ to $-1$, the curve becomes \emph{more} negative, meaning that it reduces consumer welfare. 
    }
    \label{fig:softmaxplot}
\end{figure}

\begin{restatable}{lemma}{pairwisebenefitexact}\label{lem:pairwisebenefitexact}
Adding a model $B$ to a model $A$ increases consumer welfare with BTL aggregation if and only if 
 \begin{align}
     0 \leq \sum_{\task \in [\ntasks]} \frac{\exp((\val_{\task, b} -\val_{\task, a})/\beta)}{\exp((\val_{\task, b} -\val_{\task, a})/\beta) + 1}  \cd (\val_{\task, b} -\val_{\task, a})
     \label{eq:deltaone}
 \end{align}
\end{restatable}
\begin{proof}
Fixing a single task $\task$, consumer welfare with models $\{A,B\}$ is given by: 
\begin{align*}
\val_{\task, a} \cd \frac{\exp(\val_{\task, a}/\beta)}{\exp(\val_{\task, a}/\beta) + \exp(\val_{\task, b}/\beta)} + \val_{\task, b} \cd \frac{\exp(\val_{\task, b}/\beta)}{\exp(\val_{\task, a}/\beta) + \exp(\val_{\task, b}/\beta)}\\
= \val_{\task, a} \cd \p{1-\frac{\exp(\val_{\task, b}/\beta)}{\exp(\val_{\task, a}/\beta) + \exp(\val_{\task, b}/\beta)}} + \val_{\task, b} \cd \frac{\exp(\val_{\task, b}/\beta)}{\exp(\val_{\task, a}/\beta) + \exp(\val_{\task, b}/\beta)}\\
= (\val_{\task, b}-\val_{\task, a}) \cd \frac{\exp(\val_{\task, b}/\beta)}{\exp(\val_{\task, a}/\beta) + \exp(\val_{\task, b}/\beta)} + \val_{\task, a}
\end{align*}
Note that the value of agent $A$ alone is exactly $\val_{\task, a}$. Thus, $\{A, B\}$ together has greater value exactly when 
$$ \sum_{\task \in [\ntasks]} (\val_{\task, b}-\val_{\task, a}) \cd \frac{\exp(\val_{\task, b}/\beta)}{\exp(\val_{\task, a}/\beta) + \exp(\val_{\task, b}/\beta)} \geq 0$$
\end{proof}
Note that this function is commonly known as the swish function when used for non-linearity in modern neural networks \citep{ramachandran2017searching} in Lemma \ref{lem:pairwisebenefitexact} satisfies a few properties for consumer welfare, proven in Theorem \ref{thrm:higheraccadd}.

\higheraccadd*
\begin{proof}
1. First, we will show an example where $B$ has \emph{higher} value than every model in $\modelset$, and yet adding $B$ to $\modelset$ \emph{reduces} consumer welfare. 
This proof is by example: consider a setting with exactly two types of tasks $\ntasks = 2$. Suppose that initially there are two different models: $A$ and $B$, with values across both tasks of $[8, 5]$, and $[5, 8]$, and the incoming model $C$ has values $[6, 8.1]$. Suppose BTL aggregation is done with accuracy value $\beta=2$. Then, the consumer welfare of different sets are:
\begin{itemize}
    \item Each model in isolation: $\{13, 13, 14.1\}$.
    \item $\{A, B\}$: 14.91. 
    \item The new set $A, B, C = 14.85$
\end{itemize}
Note that $\agg{A, B, C} < \agg{A, B}$, even though $C$ has higher value in isolation than either $A$ or $B$. 

2. Second, we will show an example where $B$ has \emph{lower} value than every model in $\modelset$, and yet adding $B$ to $\modelset$ \emph{increases} consumer welfare. 

Consider a setting with exactly two tasks $\ntasks=2$, and $\modelset$ is made up of a single model A, which has values $[6, 8]$, while model B has values $[7, 5]$. Note that while model B has lower total value, its relative strengths are anti-correlated with model A. If choice-based aggregation is \enquote{strong} enough, adding model $B$ could help. Specifically: 
\begin{itemize}
    \item BTL aggregation with noise level $\beta=1$ leads to higher value: $\agg{\{\model_a, \model_b\}} = 14.6 > \agg{\model_a} = 14$. 
    \item BTL aggregation with noise level $\beta=5$ leads to \emph{lower} value: $\agg{\{\model_a, \model_b\}} = 13.5 < \agg{\model_a} = 14$.
\end{itemize}
\end{proof}

\exnonmonotone*
\begin{proof}
Proof by example: assume there is a single task $\ntasks =1$ and two agents, $A$ and $B$, with aggregation by BTL with $\beta=1$. Suppose that $A$ has value 5. If $B$ has value 1, value of $\agg{[1, 5]} = 4.92$, and $B$ is picked with probability $\ppick_b([1, 5])=0.017$. If $B$ has value 4, value of $\agg{[1, 5]} = 4.73$, and $B$ is picked with probability $\ppick_b([1, 5])=0.27$.
\end{proof}

\ifarxiv 
\else 
\begin{definition}[Monotone region]\label{def:nonmonotoneregion}
A choice function $\ppick(\cdot)$ is in the 
\emph{monotone region} with respect to agent $i$, for a set of values $\valvec$ and parameters of choice function $\phi$, if the derivative of consumer utility with respect to $\valvec_i$ is non-negative, and is in the non-monotone region if the derivative is negative.  
\end{definition}
\fi

\sepmonotone*
\begin{proof}
For a separable function to be monotone, we require that for all $\valvec$: 
$$0 \leq \frac{\partial}{\partial \valvec_i} \agg{\valvec} = \frac{\partial}{\partial \valvec_i} \sum_{k \in [\nagents]} \valvec_k \cd \frac{f(\valvec_k)}{\sum_{j \in [\nagents]} f(\valvec_j)} = \frac{\partial}{\partial \valvec_i}   \frac{\sum_{k \in [\nagents]} \valvec_k \cd f(\valvec_k)}{\sum_{j \in [\nagents]} f(\valvec_j)}
$$
In this proof, we will show that for any separable function $f(\cd)$, we can construct a $\valvec$ such that it is not monotone. We will begin by assuming that for at least one element $k$, $\valvec_k >0$, which means that $\sum_{k \in [\nagents]} f(\valvec_k) >0$. 
For the fraction above, the derivative of the numerator is given by 
$\valvec_i \cd \frac{\partial}{\valvec_i} f(\valvec_i) + f(\valvec_i)$ and the derivative of the denominator is given by $\frac{\partial}{\valvec_i} f(\valvec_i)$. 
Thus, the sign of the derivative of the overall fraction is given by: 
\begin{align*}
\p{\valvec_i \cd \frac{\partial}{\partial \valvec_i} f(\valvec_i) + f(\valvec_i)} \cd  \p{\sum_{j \in [\nagents]} f(\valvec_j)} -  \p{\sum_{j \in [\nagents]} \valvec_j \cd  f(\valvec_j)} \cd \frac{\partial}{\partial \valvec_i} f(\valvec_i)  & \geq 0 \\ 
    \p{\valvec_i \cd \frac{\partial}{\partial \valvec_i} f(\valvec_i) + f(\valvec_i)} \cd  E(\valvec) -  \p{\sum_{j \in [\nagents]} \valvec_j \cd  f(\valvec_j)} \cd \frac{\partial}{\partial \valvec_i} f(\valvec_i)
    &\ge 0 \tag{Defining $E(\valvec) = \sum_{j \in [\nagents]} f(\valvec_j)$}\\
    E(\valvec) \cd f(\valvec_i) + \frac{\partial}{\partial \valvec_i} f(\valvec_i) \cd \p{\valvec_i \cd E(\valvec)  - \sum_{j \in [\nagents]} \valvec_j \cd  f(\valvec_j)} & \ge 0\\
    E(\valvec) \cd f(\valvec_i) + \frac{\partial}{\partial \valvec_i} f(\valvec_i) \cd \p{\valvec_i \cd E(\valvec)  - E(\valvec) \cd \agg{\valvec}} & \ge 0 \tag{Using $\agg{\valvec} = \sum_{j \in [\nagents]} \frac{\valvec_j \cd f(\valvec_j)}{E(\valvec)}$}\\
    E(\valvec) \cd f(\valvec_i) + \frac{\partial}{\partial \valvec_i} f(\valvec_i) \cd E(\valvec) \cd \p{\valvec_i  - \agg{\valvec}} & \ge 0\\
    f(\valvec_i) + \frac{\partial}{\partial \valvec_i} f(\valvec_i) \cd \p{\valvec_i  - \agg{\valvec}} & \ge 0 \tag{Assuming $E(\valvec) > 0$}\\
    f(\valvec_i) + \frac{\partial}{\partial \valvec_i} f(\valvec_i) \cd \valvec_i  & \ge  \frac{\partial}{\partial \valvec_i} f(\valvec_i) \cd \agg{\valvec} 
    \numberthis \label{eq:sep-cond}
\end{align*}
We're assuming that there as at least one $\valvec_i$ with $f(\val_i) > 0$. We will now construct a vector $\valvec$ such that  Equation \ref{eq:sep-cond} cannot always hold (and thus, $f(\cd)$ cannot always be monotone. Specifically, consider a setting with $\nagents=2$. Then, 
$$\agg{\valvec} = \frac{f(\valvec_1) \cd \valvec_1 + f(\valvec_2) \cd \valvec_2}{f(\valvec_1) + f(\valvec_2)}$$
Assume that $\valvec_2 > \valvec_1 > 0$, and $\frac{\partial}{\partial \valvec_1}f(\valvec_1) >0$. Then, Equation \ref{eq:sep-cond} written with respect to $\valvec_1$ becomes: 
\begin{align*}
    f(\valvec_1) + \frac{\partial}{\partial \valvec_1} f(\valvec_1) \cd \valvec_1  & \ge  \frac{\partial}{\partial \valvec_1} f(\valvec_1) \cd \agg{\valvec} \\
    \frac{f(\valvec_1) + \frac{\partial}{\partial \valvec_1} f(\valvec_1) \cd \valvec_1  }{\frac{\partial}{\partial \valvec_1} f(\valvec_1)}& \ge  \agg{\valvec}  \tag{$\frac{\partial}{\partial \valvec_1} f(\valvec_1)>0$}\\
     \frac{f(\valvec_1) + \frac{\partial}{\partial \valvec_1} f(\valvec_1) \cd \valvec_1  }{\frac{\partial}{\partial \valvec_1} f(\valvec_1)}  & \geq \frac{f(\valvec_1) \cd \valvec_1 + f(\valvec_2) \cd \valvec_2}{f(\valvec_1) + f(\valvec_2)}\numberthis \label{eq:sepcondtwo}
\end{align*}
Note that the lefthand side of Equation \ref{eq:sepcondtwo} depends solely on $\valvec_1$, while the righthand side additionally depends on $\valvec_2$. Because $f(\valvec_2)$ is weakly increasing in $\valvec_2$, as $\valvec_2 \rightarrow \infty$, the righthand side of Equation \ref{eq:sepcondtwo} also approaches $\infty$. However, because the lefthand side is a constant in $\valvec_2$, for sufficiently large $\valvec_2$ Equation \ref{eq:sepcondtwo} must not hold, and thus the function given by $f(\cd)$ cannot always be monotonic.  
\end{proof}

\begin{restatable}{definition}{monotonefunction}[Monotone example]\label{ex:monotonefunction}
Consider an aggregation function given by 
$\ppick_i(\valvec)  = \frac{1}{\binom{\abs{\valvec}}{2}} \cd w_i(\valvec)$ with:
$$w_i(\valvec) = \sum_{j=1, i\ne j}^{{\abs{\valvec}}} w_{ij}(\valvec) = \sum_{j=1, i\ne j}^{\abs{\valvec}} \frac{\max(\valvec_i-\valvec_j, 0) + c}{\abs{\valvec_i-\valvec_j} + 2c}$$  
for any $c>0$. 
\end{restatable}

\monotoneex*
\begin{proof}
We will use the function in Example \ref{ex:monotonefunction}, restated here for convenience, and show that it is always monotone: 
\monotonefunction* 
As we must know from Theorem \ref{thrm:sepmonotone}, the function in Example \ref{ex:monotonefunction} is \emph{not} separable: it relies on the gap in value between items. 

First, we will prove that this function is a valid choice-based aggregation method: specifically that it always sums up to a probability of 1. First, note that $w_{ij} + w_{ji} = 1$ for all $i\ne j$. Summing over all $\binom{\abs{\valvec}}{2}$ pairs gives $\sum_{i \in [\abs{\valvec}]} \sum_{j \in [\abs{\valvec}], j \ne i} w_{i,j}(\valvec) = \binom{\abs{\valvec}}{2}$, so $\sum_{i \in [\abs{\valvec}]} w_i(\valvec) = 1$, as desired. Because $w_{i,j} \geq 0$, we know that $w_i(\valvec)\geq 0$ and thus each probability is non-negative, as required.

Next, we will show that this function is monotone in value, or 
\begin{align*}
    0 \leq & \frac{\partial}{\partial \valvec_i} \agg{\valvec} \quad \forall i \in [\abs{\valvec}]\\
     = & \frac{\partial}{\partial \valvec_i} \sum_{i \in [\abs{\valvec}]} \ppick_i(\valvec) \cd \valvec_i\\
     = & \ppick_i(\valvec) + \valvec_i \cd \frac{\partial}{\partial v_i} \ppick_i(\valvec) + \sum_{j \ne i, j \in [\abs{\valvec}]} \frac{\partial}{\partial \valvec_i} \valvec_j \cd \ppick_j(\valvec) \\
     \propto & \sum_{k \in [\abs{\valvec}], k \ne i} w_{ik}(\valvec) + \valvec_i \cd  \frac{\partial}{\partial v_i} \sum_{k \in [\abs{\valvec}], k \ne i} w_{ik}(\valvec) + \sum_{j \ne i, j \in [\abs{\valvec}]} \valvec_j  \sum_{k \in [\abs{\valvec}], k\ne i, j} \frac{\partial}{\partial \valvec_i}  \cd w_{jk}(\valvec) \tag{Definition \ref{ex:monotonefunction}}\\
     = & \sum_{k \in [\abs{\valvec}], k \ne i} w_{ik}(\valvec) + \valvec_i \cd  \frac{\partial}{\partial v_i} \sum_{k \in [\abs{\valvec}], k \ne i} w_{ik}(\valvec) + \sum_{j \ne i, j \in [\abs{\valvec}]} \valvec_j \frac{\partial}{\partial \valvec_i}  \cd w_{ji}(\valvec) \tag{$\frac{\partial}{\partial \valvec_i} w_{jk} = 0$ if $j\ne i$ and $k \ne i$}\\
     = & \sum_{j \in [\abs{\valvec}], j \ne i} w_{ij}(\valvec) + \valvec_i \cd  \frac{\partial}{\partial v_i} \sum_{j \in [\abs{\valvec}], j \ne i} w_{ij}(\valvec) - \sum_{j \ne i, j \in [\abs{\valvec}]} \valvec_j \frac{\partial}{\partial \valvec_i}  \cd w_{ij}(\valvec) \tag{$w_{ji} = 1-w_{ij}$}\\
     = & \sum_{j \in [\abs{\valvec}], j \ne i} \p{w_{ij}(\valvec) +  (\valvec_i - \valvec_j) \cd \frac{\partial}{\partial \valvec_i}  \cd w_{ij}(\valvec)}
\end{align*}
Next, we will show that every term within the sum above is positive: that is, that $w_{ij}(\valvec) +  (\valvec_i - \valvec_j) \cd \frac{\partial}{\partial \valvec_i}  \cd w_{ij}(\valvec)\geq 0$ always. 

Note that by the definition in Definition \ref{ex:monotonefunction}, we know that: 
$$\begin{cases}
    \frac{\partial}{\partial \valvec_i}w_{ij} = \frac{\partial}{\partial \valvec_i} \frac{\valvec_i - \valvec_j + c}{\valvec_i - \valvec_j + 2 c}= \frac{c}{(\abs(v_i-v_j) + 2c)^2} \quad \valvec_i \geq \valvec_j\\
    \frac{\partial}{\partial \valvec_i}w_{ij} =\frac{\partial}{\partial \valvec_i} \frac{c}{\valvec_j-\valvec_i + 2 c}= \frac{c}{(\abs(\valvec_i-\valvec_j) + 2c)^2} \quad \valvec_i < \valvec_j\\
\end{cases}$$
Thus, the desired condition holds for every pair with $\valvec_i > \valvec_j$ if: 
\begin{align*}
     0 < & w_{ij}(\valvec) +  (\valvec_i - \valvec_j) \cd \frac{\partial}{\partial \valvec_i}  \cd w_{ij}(\valvec)\\
     = &  \frac{\valvec_i - \valvec_j + c}{\valvec_i - \valvec_j + 2c}+ (\valvec_i - \valvec_j) \cd \frac{c}{(\valvec_i-\valvec_j + 2c)^2}
\end{align*}
which is automatically satisfied because both terms are positive. The desired condition holds for every $j$ pair with $\valvec_i < \valvec_j$ if: 
\begin{align*}
     0 < & w_{ij}(\valvec) +  (\valvec_i - \valvec_j) \cd \frac{\partial}{\partial \valvec_i}  \cd w_{ij}(\valvec)\\
    0 < &  \frac{c}{\valvec_j - \valvec_i + 2 c}+ (\valvec_i - \valvec_j) \cd \frac{c}{(\valvec_j-\valvec_i + 2c)^2}\\
     (\valvec_j - \valvec_i) \cd \frac{c}{(\valvec_j-\valvec_i + 2c)^2} <  &  \frac{c}{\valvec_j - \valvec_i + 2 c}\\ 
     (\valvec_j - \valvec_i) \cd \frac{1}{\valvec_j-\valvec_i + 2c} <  &  1 \tag{$\valvec_j - \valvec_i + c>0$}\\ 
     \valvec_j - \valvec_i < & \valvec_j-\valvec_i + 2c
\end{align*}
which is satisfied because $c>0$, by assumption. 
\end{proof}

\begin{restatable}{lemma}{septthreea}\label{lem:septthreea}
Consider a vector of values $\valvec$ and a new vector $\valvec'$ where with index $j$ changed and all else constant. Then, $\agg{\valvec'} \geq \agg{\valvec}$ if and only if the following condition holds:
$$\sum_{i \in \abs{\valvec}, i \ne j} (\ppick_i(\valvec) - \ppick_i(\valvec')) \cd (\valvec_j'-\valvec_i) >  (\valvec_j-\valvec_j')\cd \ppick_j(\valvec)$$
\end{restatable}
\begin{proof}
This can be seen simply by rewriting the value function:
$$\agg{\valvec}= \valvec_j \cd \ppick_j(\valvec) + \sum_{i \in \valvec, i\ne j} \ppick_i(\valvec) \cd \valvec_i$$
and 
$$\agg{\valvec'} = \valvec_j'\cd \ppick_j(\valvec') + \sum_{i \in \valvec, i \ne j} \ppick_i(\valvec') \cd \valvec_i$$
Note that we can write: 
\begin{align*}
    \ppick_j(\valvec') = &   \ppick_j(\valvec) + \ppick_j(\valvec') - \ppick_j(\valvec) \\
    = &  \ppick_j(\valvec) + 1-\sum_{i \ne j} \ppick_i(\valvec')  -1 + \sum_{i\ne j} \ppick_i(\valvec) \\
    = &  \ppick_j(\valvec) + \sum_{i \in \valvec, i \ne j} \ppick_i(\valvec) - \ppick_i(\valvec')
\end{align*}
The condition where the value of $\valvec'$ is larger is given by:  
\begin{align*}
\agg{\valvec'} \geq &  \agg{\valvec}\\
\valvec_j' \cd \ppick_j(\valvec') + \sum_{i \in \valvec, i \ne j} \ppick_i(\valvec') \cd \valvec_i \geq & \valvec_j\cd \ppick_j(\valvec) + \sum_{i \in \valvec, i \ne j} \ppick_i(\valvec) \cd \valvec_i\\
\valvec_j' \cd \p{\ppick_j(\valvec) + \sum_{i \in \valvec, i \ne j} \ppick_i(\valvec) - \ppick_i(\valvec')} + \sum_{i \in \valvec, i \ne j} \ppick_i(\valvec') \cd \valvec_i \geq &  \valvec_j\cd \ppick_j(\valvec) + \sum_{i \in \valvec, i \ne j} \ppick_i(\valvec) \cd \valvec_i\\
\valvec_j' \cd \p{\sum_{i \in \valvec, i \ne j} \ppick_i(\valvec) - \ppick_i(\valvec')} + \sum_{i \in \valvec, i \ne j} \p{\ppick_i(\valvec') -  \ppick_i(\valvec)} \cd \valvec_i \geq &  \valvec_j\cd \ppick_j(\valvec)  - \valvec_j' \cd \ppick_j(\valvec)\\
\sum_{i \in \valvec, i \ne j} (\ppick_i(\valvec) - \ppick_i(\valvec')) \cd (\valvec_j' - \valvec_i) > & (\valvec_j - \valvec_j') \cd \ppick_j(\valvec)
\end{align*}
giving the desired condition. 
\end{proof}

\begin{restatable}{lemma}{monotoneppos}\label{lem:monotoneppos}
Consider any continuous increasing function $f(\valvec)$. Then, if $f$ is monotone it must have $\ppick_j(\valvec) > 0$ always. 
\end{restatable}
\begin{proof}
This proof will work by considering a hypothetical vector of values $\valvec$ and showing that monotonicity is impossible to guarantee if we allow $\ppick_j(\valvec) =0$. Specifically, we will construct a separate vector $\valvec'$ that is identical to $\valvec$ except in index $j$. We will rely on the result from Lemma \ref{lem:septthreea}, which says that such a vector has $\agg{\valvec'}\geq \agg{\valvec}$ if and only if: 
$$\sum_{i \in \abs{\valvec}, i \ne j} (\ppick_i(\valvec) - \ppick_i(\valvec')) \cd (\valvec_j'-\valvec_i) >  (\valvec_j-\valvec_j')\cd \ppick_j(\valvec)$$

Consider a hypothetical $\valvec$ with a minimal element $\valvec_1$ and $\ppick_1(\valvec) = 0$. Assume that $1$ is the only element with $\ppick_1(\valvec)=0$. In particular, consider the second smallest element $2$, and assume that $\ppick_2(\valvec) >0$.

Consider the hypothetical vector $\valvec'$ with $\valvec_1' = \valvec_2' = \valvec_2$. Because we assumed anonymity, we require $\ppick_1(\valvec') = \ppick_2(\valvec')$. Note by substitutability, we must have that increasing $\valvec_1$ can only keep constant or decrease the probability of other elements $\valvec_i$ being picked. This means that we must have $\ppick_1(\valvec') = \ppick_2(\valvec')>0$: if this is not the case, then by conservation of total probability, we must have $\ppick_j(\valvec') > \ppick_j(\valvec)$, which violates substitutability. 

Given, $\ppick_1(\valvec') = \ppick_2(\valvec')>0$, by continuity, there must exist a point $\valvec_1 < \valvec_1'' < \valvec_2$ such that $\ppick_1(\valvec'') > 0$. Note that this means that $\ppick_1(\valvec'') > \ppick_1(\valvec) = 0$. We will again show that monotonicity is violated here. 

We note that: 
\begin{enumerate}
    \item By assumption, $\valvec_j  < \valvec_j'$ and so $\valvec_j - \valvec_j'< 0$. 
    \item By assumption, for all other $i \ne j$, we have $\valvec_j'< \valvec_i$, and so $\valvec_j' - \valvec_i < 0$. 
    \item By substitutability (Def. \ref{def:substitutability}), $\ppick_i(\valvec) - \ppick_i(\valvec') \geq 0$ for all $i \ne j$.  
\end{enumerate}
Considering the inequality from Lemma \ref{lem:septthreea}, we know that $\agg{\valvec'} \geq \agg{\valvec}$ exactly whenever: 
\begin{equation}\label{eq:necessarymonotone}
\sum_{i \in \valvec, i \ne j} (\ppick_i(\valvec) - \ppick_i(\valvec')) \cd (\val_j'-\val_i) >  (\val_j-\val_j')\cd \ppick_j(\valvec)
\end{equation}
For the lefthand side, property 2 tells us that $\val_j' - \val_i$ is negative always and property 3 tells us that $\ppick_i(\valvec) - \ppick_i(\valvec') \geq 0$. Note that the LHS can only be 0 if $\ppick_i(\valvec) = \ppick_i(\valvec')$ $\forall i \ne j$: however, this is impossible by the requirement that $\ppick_j(\valvec') > \ppick_j(\valvec)$ (given by the continuity of the selection function). Thus, we know that the lefthand side of Equation \ref{eq:necessarymonotone} is always strictly negative. Thus, in order for the function to be monotone, we must have that the RHS of Equation \ref{eq:necessarymonotone} is always strictly negative. By Property 1, we know that $\val_j - \val_j'<0$, and so we must have $\ppick_j(\valvec) > 0$ strictly. 
\end{proof}

\subsection{Tension between monotonicity and model creation}\label{app:tensionmonotonecreate}

\begin{definition}[Substitutability]\label{def:substitutability}
A selection function $\ppick(\cdot)$ satisfies \emph{substitutability} if increasing the value of an item $\valvec_{\model, \task}$ either holds constant or decreases the probability of every other item being picked. That is, for $\valvec_{\task}, \valvec_{\task'}$ identical except with $\valvec_{\model, \task} < \valvec_{\model, \task}'$, we must have $\ppick_{\model}(\valvec_{\task}) \leq \ppick_{\model}(\valvec_{\task}') $. 
\end{definition}

\begin{definition}[Anonymity]\label{def:anon}
A selection function $\ppick(\cdot)$ satisfies \emph{anonymity} if the probability of picking an item depends only on its value.
\end{definition}
\ifarxiv 
\mutexclusive*
\begin{proof}
This comes almost directly from Lemma \ref{lem:monotoneppos}. First, suppose that the function is always monotone in value. Then, by Lemma \ref{lem:monotoneppos}, we know that there must be a positive probability of picking any item within $\valvec$. Then, if we create a $\valvec'$ with a new element such that the vector becomes $\valvec\oplus \val$, where it is smaller than the current minimal element $\val < \min_{i \in \abs{\valvec}} \valvec_i$, then we must have $\ppick_{\nagents+1}(\valvec')>0$. Because total probability is conserved, we must see that the probability of picking every other item $i \ne j$ weakly decreases. Because every other item in $\valvec'$ has strictly higher value than $\val$, this means that we must have $\agg{\valvec'} < \agg{\valvec}$, and so there are strict harms from adding in another agent. \\
Next, we assume that there are always weak benefits to adding in another agent: that is, $\agg{\valvec \oplus\val} > \agg{\valvec}$ always. In this case, we must have that for an element $\val < \min_{i \in \abs{\valvec}} \valvec_i$, if $\agg{\valvec \oplus\val} > \agg{\valvec}$, then we must have $\ppick_{\nagents+1}(\valvec) = 0$. By the conditions of Lemma \ref{lem:monotoneppos}, this must mean that the associated selection function cannot be monotone. 
\end{proof}
\else 

Here, we want to highlight a tension between consumer welfare increase from model \emph{replacement} (monotonicity) and consumer welfare increase from model \emph{creation}. That is, for any choice functions satisfying natural properties (substitutability and anonymity, Definitions \ref{def:substitutability},\ref{def:anon}), these two properties cannot be simultaneously satisfied. 

\begin{definition}[Substitutability]\label{def:substitutability}
A selection function $\ppick(\cdot)$ satisfies \emph{substitutability} if increasing the value of an item $\valvec_{\model, \task}$ either holds constant or decreases the probability of every other item being picked. That is, for $\valvec_{\task}, \valvec_{\task'}$ identical except with $\valvec_{\model, \task} < \valvec_{\model, \task}'$, we must have $\ppick_{\model}(\valvec_{\task}) \leq \ppick_{\model}(\valvec_{\task}') $. 
\end{definition}

\begin{definition}[Anonymity]\label{def:anon}
A selection function $\ppick(\cdot)$ satisfies \emph{anonymity} if the probability of picking an item depends only on its value.
\end{definition}

\begin{restatable}{theorem}{mutexclusive}\label{thrm:mutexclusive}
These two properties are mutually exclusive for choice-based, continuous aggregation functions satisfying substitutability and anonymity\footnote{Informally, a function satisfies anonymity if the probability of picking an item depends only on its value, and substitutability if the probability of picking an item only stays constant or increases when its value increases.}: 1) the function is monotone, and 2) there are always weak benefits to adding in another agent. 
\end{restatable}
\begin{proof}
This comes almost directly from Lemma \ref{lem:monotoneppos}. First, suppose that the function is always monotone in value. Then, by Lemma \ref{lem:monotoneppos}, we know that there must be a positive probability of picking any item within $\valvec$. Then, if we create a $\valvec'$ with a new element such that the vector becomes $\valvec\oplus \val$, where it is smaller than the current minimal element $\val < \min_{i \in \abs{\valvec}} \valvec_i$, then we must have $\ppick_{\nagents+1}(\valvec')>0$. Because total probability is conserved, we must see that the probability of picking every other item $i \ne j$ weakly decreases. Because every other item in $\valvec'$ has strictly higher value than $\val$, this means that we must have $\agg{\valvec'} < \agg{\valvec}$, and so there are strict harms from adding in another agent. \\
Next, we assume that there are always weak benefits to adding in another agent: that is, $\agg{\valvec \oplus\val} > \agg{\valvec}$ always. In this case, we must have that for an element $\val < \min_{i \in \abs{\valvec}} \valvec_i$, if $\agg{\valvec \oplus\val} > \agg{\valvec}$, then we must have $\ppick_{\nagents+1}(\valvec) = 0$. By the conditions of Lemma \ref{lem:monotoneppos}, this must mean that the associated selection function cannot be monotone. 
\end{proof}

\fi 
\sufficientmonotone*
\begin{proof}
First, we note that perfect aggregation ($\beta=0$) is always monotone because increasing elements of a list $\valvec$ only increases the maximum element in that list. Similarly, random aggregation ($\beta=\infty$) is always monotone because increasing elements of a list $\valvec$ only increases the average of that list. 

Next, we turn to analyzing the conditions for monotonicity with $\beta \in (0, \infty)$. A function is monotone whenever the below quantity is positive for all $i \in \abs{\valvec}$: 
\begin{align*}
\frac{\partial}{\partial \valvec_i} \agg{\valvec} =  & \sum_{j \in \abs{\valvec}}\frac{\partial}{\partial \valvec_i} \valvec_j \cd \ppick_j(\valvec)\\
= & \ppick_i(\valvec) + \valvec_i \cd \frac{\partial}{\partial \valvec_i} \ppick_i(\valvec) + \sum_{j \ne i} \valvec_j \cd \frac{\partial}{\partial \valvec_i} \ppick_j(\valvec)\\
= &  \ppick_i(\valvec) + \valvec_i \cd \ppick_i(\valvec) \cd (1-\ppick_i(\valvec))/\beta - \sum_{j\ne i} \valvec_j \cd \ppick_i(\valvec) \cd \ppick_j(\valvec)/\beta \tag{Properties of BTL}\\
= & \ppick_i(\valvec)\cd\p{1 + \sum_{j\ne i}(\valvec_i-\valvec_j) \ppick_j(\valvec)/\beta}
\end{align*}
The above quantity is positive so long as
\begin{align*}
1 + \sum_{j\ne i}(v_i-v_j) p_j(\valvec)/\beta& >0\\
\beta + \sum_{j\ne i}(v_i-v_j) p_j(\valvec) & >0\\
\beta &> \sum_{j\ne i}(v_j-v_i) \cd p_j(\valvec)
\numberthis \label{eq:moncondition}
\end{align*}
If $\valvec_i \geq \valvec_j$ for all $i \ne j$ (that is, if element $i$ is the largest) then Equation \ref{eq:moncondition} is automatically satisfied for player $i$ because the righthand side is negative. Next, we will consider cases where there are some agents such that $\valvec_i < \valvec_j$. Note that as $\beta$ increases, $\ppick_j(\valvec)$ may increase or decrease, but is upper bounded by $1$, and we have $\sum_{j\ne i} \ppick_j(\valvec) \leq 1$. 
Thus, we have: 
$$\sum_{j\ne i}(v_j-v_i) \cd p_j(\valvec) \leq \max_{j \ne i} \valvec_j - \valvec_i$$
For any agent $i$, for any $\beta> \beta_i^* = \max_{j \ne i} \valvec_j - \valvec_i$, Equation \ref{eq:moncondition} is guaranteed to be positive, and the condition is monotone. 
\end{proof}

\begin{figure}
\centering 
    \includegraphics[width=3.5in]{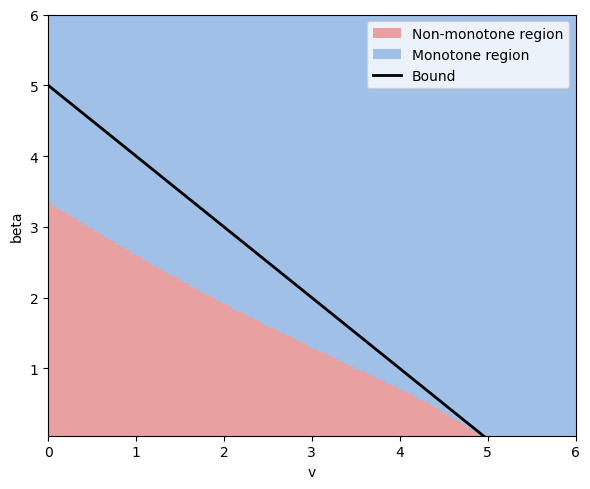}
    \caption{Figure demonstrating the bound in Lemma \ref{lem:sufficientmonotone}, where there is a set of values $[1, 3, 5, \val]$, with $\val$ varying along the $x$ axis and $\beta$ varying along the $y$ axis. {\color{blue}Blue} regions denote settings where BTL aggregation is in the monotone region with respect to $\val$, and {\color{red}red} regions denote settings where it is not. The black line gives the bound in Lemma \ref{lem:sufficientmonotone}: above this line BTL aggregation is guaranteed to be in the monotone region with respect to $\val$. }
    \label{fig:monotoneregion}
\end{figure}

\section{Proofs for Section \ref{sec:incententry}}\label{app:incententry}

\begin{lemma}\label{lem:addnew}
The change in consumer welfare when a new agent with value $\val'$ joins a task with existing list of values $\valvec$ is given by: 
$$\agg{\valvec \oplus \val'} - \agg{\valvec} =   \ppick_{\nagents+1}(\valvec \oplus \val') \cd (\valvec' - \agg{\valvec})$$
where $\ppick_{\nagents+1}(\cd)$ gives the probability that the new item is picked. 
\end{lemma}
\begin{proof}
Suppose that we have $\nagents = \abs{\valvec}-1$ agents initially. Then, we can calculate the change in value as:
\begin{align}
    \agg{\valvec \oplus \val'} - \agg{\valvec} = &\sum_{i \in [\nagents]} \valvec_i \cd \ppick_i(\valvec \oplus \val') + \val' \cd \ppick_{\nagents}(\valvec \oplus \val') - \sum_{i \in [\nagents]} \valvec_i \cd \ppick_i(\valvec)  \nonumber\\ 
    = & \val' \cd \ppick_{\nagents}(\valvec \oplus \val') -\p{\sum_{i \in [\nagents]} \valvec_i \cd \p{\ppick_i(\valvec)- \ppick_i(\valvec \oplus \val')} }\nonumber\\
    = &  \val' \cd \frac{\exp(\val'/\beta)}{E + \exp(\val'/\beta)} - \sum_{i \in [\nagents]}\valvec_i \cd \exp(\valvec_i/\beta)\p{\frac{1}{E}-\frac{1}{E + \exp(\val'/\beta)}}\nonumber\\
    = &  \val' \cd \frac{\exp(\val'/\beta)}{E + \exp(\val'/\beta)} - \sum_{i \in [\nagents]}\valvec_i \cd \exp(\valvec_i/\beta)\frac{\exp(\val'/\beta)}{E \cd (E + \exp(\val'/\beta)})\nonumber\\
    = & \frac{\exp(\val'/\beta)}{E + \exp(\val'/\beta)}\p{ \val'- \sum_{i \in [\nagents]}\valvec_i \cd \frac{\exp(\valvec_i/\beta)}{E}}\nonumber\\
    = & \ppick_{\nagents}(\valvec \oplus \val') \cd (\valvec' - \agg{\valvec})\label{eq:addinghelps}
\end{align}
as desired. 
\end{proof}

\begin{restatable}{lemma}{reallocatemoretwovalue}\label{lem:reallocatemoretwovalue}
Consider a pair of tasks where the existing allocation is $\valvec_1, \valvec_2$. A new agent has total value $\Val$ to allocate across both tasks, and is optimizing for consumer welfare. Then, the optimal arrangement is always to put all value on a single task and abstain on the other.
\end{restatable}
\begin{proof}
Recall from Lemma \ref{lem:addnew} that when new value $\val'$ is added to a task with existing values $\valvec$, the change in consumer welfare is given by: 
$$\agg{\valvec \oplus \val'} - \agg{\valvec} =   \ppick_{\nagents}(\valvec \oplus \val') \cd (\valvec' - \agg{\valvec})$$
where $\ppick_{\nagents}(\cd)$ gives the probability that the new item is picked.
Thus, given an allocation $\val, \Val-\val$ across both tasks, the \emph{change} in consumer welfare across both tasks is given by: 
\begin{align}
    \ppick_{\nagents+1}(\valvec_1 \oplus \val) \cd (\val - \agg{\valvec_1}) + \ppick_{\nagents+1}(\valvec_2 \oplus \Val- \val) \cd (\Val-\val - \agg{\valvec_2}) \label{eq:valueobjs}
\end{align}
Our goal is to show that the objective in Equation \ref{eq:valueobjs} is maximized at one of the endpoints. Note that this objective is \emph{not} convex, and in particular can have multiple local maxima and minima. We will prove that Equation \ref{eq:valueobjs} is maximized at its endpoints by considering different values of $\val \in (0, \Val)$ and showing that the value obtained there must be lower. 

First, consider the setting where we have $\val$ such that $\val < \agg{\val_1}$ or $\Val-\val <\agg{\val_2}$: that is, at least one of the allocation is \emph{decreasing} the average value on the task. First, note that if the allocation is decreasing on both tasks, then the total change is negative: 
\begin{align*}
    \ppick_{\nagents+1}(\valvec_1 \oplus \val) \cd (\val - \agg{\valvec_1}) + \ppick_{\nagents+1}(\valvec_2 \oplus \Val- \val) \cd (\Val-\val - \agg{\valvec_2}) \leq 0\\
     \leq & \max_{i \in [1, 2]} \ppick_{\nagents+1}(\valvec_i \oplus \Val) \cd (\Val - \agg{\valvec_i})
\end{align*}
And so the condition is satisfied. Next, suppose that the allocation leads to negative change on exactly one task: WLOG, assume it occurs on task 1. Then, we know that: 
\begin{align*}
    \ppick_{\nagents+1}(\valvec_1 \oplus \val) \cd (\val - \agg{\valvec_1}) + \ppick_{\nagents+1}(\valvec_2 \oplus \Val- \val) \cd (\Val-\val - \agg{\valvec_2}) < & \ppick_{\nagents+1}(\valvec_1 \oplus \val) \cd 0 \\
    & + \ppick_{\nagents+1}(\valvec_2 \oplus \Val- \val) \cd (\Val-\val - \agg{\valvec_2})\\
    = & \ppick_{\nagents+1}(\valvec_2 \oplus \Val- \val) \cd (\Val-\val - \agg{\valvec_2})\\
     <  & \ppick_{\nagents+1}(\valvec_2 \oplus \Val) \cd (\Val - \agg{\valvec_2})\\
     \leq & \max_{i \in [1, 2]} \ppick_{\nagents+1}(\valvec_i \oplus \Val) \cd (\Val - \agg{\valvec_i})
\end{align*}
again, as desired. 
Finally, we consider the case where  $\val \geq  \agg{\val_1}$ and $\Val-\val \geq \agg{\val_2}$: the allocation leads to positive benefits on both tasks. WLOG, assume that the task has a higher chance of being picked on task 1 than 2: that is:
\begin{equation}
\ppick_{\nagents+1}(\valvec_1 \oplus \val) \geq  \ppick_{\nagents+1}(\valvec_2 \oplus \Val-\val) \label{eq:larger}
\end{equation}
Then, we know that:  

\begin{align*}
    \ppick_{\nagents+1}(\valvec_1 \oplus \val) \cd (\val - \agg{\valvec_1}) + \ppick_{\nagents+1}(\valvec_2 \oplus \Val- \val) \cd (\Val-\val - \agg{\valvec_2}) <& \ppick_{\nagents+1}(\valvec_1 \oplus \val) \cd (\val - \agg{\valvec_1}) \\
     & + \ppick_{\nagents+1}(\valvec_1 \oplus \Val- \val) \cd (\Val-\val) \\
     &+ \ppick_{\nagents+1}(\valvec_2 \oplus \Val- \val) \cd (- \agg{\valvec_2}) \tag{Eq. \ref{eq:larger}}\\
    =& \ppick_{\nagents+1}(\valvec_1 \oplus \val) \cd (\Val - \agg{\valvec_1})\\
     & - \ppick_{\nagents+1}(\valvec_2 \oplus \Val- \val) \cd \agg{\valvec_2}\\
     <  & \ppick_{\nagents+1}(\valvec_1 \oplus \Val) \cd (\Val - \agg{\valvec_1})\\
     \leq & \max_{i \in [1, 2]} \ppick_{\nagents+1}(\valvec_i \oplus \Val) \cd (\Val - \agg{\valvec_i})
\end{align*}
again, as desired. 

\end{proof}

\optvalueb*
\begin{proof}
First, Lemma \ref{lem:addnew} tells us that when a new value $\val'$ is added to a task with existing list of values $\valvec$, the change in consumer welfare is given by: 
$$\agg{\valvec \oplus \val'} - \agg{\valvec} =   \ppick_{\nagents}(\valvec \oplus \val') \cd (\valvec' - \agg{\valvec})$$
where $\ppick_{\nagents}(\cd)$ gives the probability that the new item is picked.

This immediately handles the setting where $\Val < \agg{\valvec}$: if total value is less than current value, any possible allocation would harm consumer welfare, and the best action is to entirely abstain from entering the market. 

Next, Lemma \ref{lem:reallocatemoretwovalue} tells us that the optimal allocation for maximizing value must place all value on a single task. 

Finally, we can directly see which allocation maximizes consumer welfare: it is whichever task $i$ maximizes: 
\begin{align*}
    \ppick_{\nagents+1}(\valvec_i \oplus \Val) \cd (\Val - \agg{\valvec_i})
\end{align*}
\end{proof}

\begin{restatable}{lemma}{reallocatemoreb}\label{lem:reallocatemoreb}
Consider a pair of tasks where the existing allocation is $\valvec_1, \valvec_2$. A new agent has total value $\Val$ to allocate across both tasks, and is optimizing for total winrate. Then, the optimal arrangement for total winrate with BTL by the new agent is given by:
\begin{itemize}
    \item If $\Val< \beta \cd \log( E_1 \cd E_2)$, where $E_i =\sum_{j \in \abs{\valvec}} \exp(\valvec_{j, i}/\beta))$, have the allocation maximally anti-correlated, with all value on the task with smaller $E_i$ value: $\val_1 = \Val, \val_2 = 0$ (or vice-versa). 
    \item If $\Val\geq  \beta \cd \log( E_1 \cd E_2)$, where $E_i =\sum_{j \in \abs{\valvec}} \exp(\valvec_{j, i}/\beta))$, the optimal arrangement is to have the allocation such that winrate is equal across both tasks.
\end{itemize}
\end{restatable}
\begin{proof}
Denote the new agent's allocation on tasks by $\val_1, \Val-\val_1$. Denote the new agent by index $c$ and denote $E_i = \sum_{j \in \abs{\valvec_i}} \exp(\val_j/\beta)$. Then,  new agent $c$ winning on a given task with its allocated value $\val_1$ is given by the BTL function:
$$p_c(\valvec_i \oplus \val_i) = \frac{\exp(\val_i/\beta)}{E_i + \exp(\val_i/\beta)}$$
The total winrate for agent $c$ across the two tasks is:
$$W(v_1) = p_c(\valvec_1 \oplus \val_1) + p_c(\valvec_2 \oplus \Val- \val_1) = \frac{\exp(\val_1/\beta)}{E_1 + \exp(\val_1/\beta)} + \frac{\exp((\Val - \val_1)/\beta)}{E_2 + \exp((\Val - \val_1)/\beta)}$$
Next, we will analyze this function and show when it is maximized at the endpoints, rather than towards the center. First, note that the derivative of $W$ with respect to $\val_1$ is:
$$\frac{\partial W}{\partial \val_1} =  p_c'(\valvec_1 \oplus \val_1) - \ppick_c'(\valvec_2 \oplus \Val- \val_1) = (p_c(\valvec_1 \oplus \val_1) \cd (1-p_c(\valvec_1 \oplus \val_1)) -\ppick_c(\valvec_2 \oplus \Val- \val_1) \cd (1-\ppick_c(\valvec_2 \oplus \Val- \val_1)))\cd 1/\beta $$
And thus $\frac{\partial}{\partial \val_1} W(\val_1)$ is 0 either a) when $p_c(\val_1) = p_c(\Val-\val_1)$  or b) when $p_c(v_1) = 1-p_c(V-v_1)$. We will first show that $b)$ is only possible if $\Val =\log(E_1 \cd E_2)$. This occurs when: 
\begin{align*}
    p_c(\valvec_1 \oplus \val_1) =& 1-\ppick_c(\valvec_2 \oplus \Val- \val_1)\\
    \frac{\exp(\val_1/\beta)}{E_1 + \exp(\val_1/\beta)} = &  \frac{E_2}{E_2 + \exp((\Val-\val_1)/\beta)}\\
    E_2 \cd \exp(\val_1/\beta) + \exp(\Val/\beta) = & E_1 \cd E_2 + E_2 \cd \exp(\val_1/\beta) \\ 
    \exp(\Val/\beta) = & E_1 \cd E_2 \\ 
    \Val/\beta =& \log(E_1 \cd E_2) 
\end{align*}
Note that this condition holds irrespective of $\val_1$: that is, if we have $\Val =  \beta \cd \log(E_1 \cd E_2) $, then we must have $\frac{\partial}{\partial \val_1} W = 0$ always, and so any allocation will give equal total winrate across both tasks to player $c$. 

We next turn to the setting where $\Val \ne \beta \cd \log(E_1 \cd E_2)$: thus, this means that the only point of zero derivative $\frac{\partial}{\partial \val_1} W$ occurs exactly where the probabilities of winning are identical: $p_c(\val_1) = p_c(\Val-\val_1)$. We will show whether this is a maximum or minimum by examining the second derivative of the function $W$ at the this point. First, we note that by the properties of the BTL function: 
$$\frac{\partial W}{\partial \val_1} = (p_c(\valvec_1 \oplus \val_1) \cd (1-p_c(\valvec_1 \oplus \val_1)) -\ppick_c(\valvec_2 \oplus \Val- \val_1) \cd (1-\ppick_c(\valvec_2 \oplus \Val- \val_1)))\cd 1/\beta $$ 
and thus the second derivative is given by: 
\begin{align*}
    \frac{\partial^2}{\partial \val_1^2} W  = & \ppick_c(\valvec_1 \oplus \val_1) \cd (1-\ppick_c(\valvec_1 \oplus \val_1)/\beta^2 -2 \cd \ppick_c(\valvec_1 \oplus \val_1) \cd (1-\ppick_c(\valvec_1 \oplus \val_1))/\beta^2 \\
    & - (-\ppick_c(\valvec_2 \oplus \Val - \val_1) \cd (1-\ppick_c(\valvec_2 \oplus \Val - \val_1)) +2\ppick_c(\valvec_2 \oplus \Val - \val_1)^2 \cd (1-\ppick_c(\valvec_2 \oplus \Val - \val_1)))/\beta^2\\
     = & 1/\beta^2 \cd \p{\ppick_c(\valvec_1\oplus \val_1) \cd (1-\ppick_c(\valvec_1 \oplus \val_1)) \cd (1-2 \cd \ppick_c(\valvec_1 \oplus \val_1))} \\
     & + 1/\beta^2 \cd \p{\ppick_c(\valvec_2 \oplus \Val - \val_1) \cd (1-\ppick_c(\valvec_2 \oplus \Val-\val_1)) \cd (1-2 \cd \ppick_c(\valvec_2 \oplus \Val - \val_1))}
\end{align*}
Next, we can evaluate this second derivative at the point where $\ppick_c(\valvec_1 \oplus \val_1) = \ppick_c(\valvec_2 \oplus \Val - \val_1) = \ppick^*$, and thus we can simplify the above term to: 
$$2 \cd\ppick^*\cd (1-\ppick^*) \cd (1-2 \cd \ppick^*) \cd 1/\beta^2$$
Note that $\ppick^*\cd (1-\ppick^*) >0$ always, and so the sign of the second derivative depends entirely on whether the term $1-2 \cd \ppick^*$ is positive or negative. This occurs exactly when equalizing the winrates allows the new agent to have winrate $>0.5$ on both tasks. When does this occur? The value giving equal winrate is: 

\begin{align*}
    p_c(\valvec_1 \oplus \val_1) =& \ppick_c(\valvec_2 \oplus \Val- \val_1)\\
    \frac{\exp(\val_1/\beta)}{E_1 + \exp(\val_1/\beta)} = &  \frac{\exp((\Val-\val_1)/\beta)}{E_2 + \exp((\Val-\val_1)/\beta)}\\
E_2 \cd \exp(\val_1/\beta) + \exp(\Val/\beta) = & E_1 \cd \exp((\Val-\val_1)/\beta) + \exp(\Val/\beta) \\
E_2 \cd \exp(\val_1/\beta) = & E_1\cd \exp((\Val-\val_1)/\beta)\\
\exp(2\val_1/\beta) = & \frac{E_1}{E_2} \cd \exp(\Val/\beta)\\
\val_1 = &  0.5 \cd \beta  \cd \log\p{\frac{E_1}{E_2}} + 0.5 \cd \Val
\end{align*}
This gives winrate $>0.5$ whenever: 
\begin{align*}
    \exp(\val_1/\beta) > & E_1\\
    \sqrt{\frac{E_1}{E_2}} \cd \exp(\Val/(2\beta)) > &  E_1\\
    \frac{E_1}{E_2} \cd \exp(\Val/\beta) > &  E_1^2\\
    \exp(\Val/\beta) > &  E_1 \cd E_2\\
    \Val > & \beta \cd \log(E_1 \cd E_2)
\end{align*}
Thus, if $\Val$ is below this bound, then the midpoint has a positive second derivative, and thus is a minimum: in this case, the function $W$ is maximized at the endpoints, with $\val_1 = \Val, \val_2 = 0$. Note that the allocation will always have nonzero allocation on the task with smaller $E_i$, because if $E_1 < E_2$, then: 
\begin{align*}
\ppick_{\nagents+1}(\valvec_1 \oplus \Val) + \ppick_{\nagents+1}(\valvec_2 \oplus 0)  > & \ppick_{\nagents+1}(\valvec_1 \oplus 0) + \ppick_{\nagents+1}(\valvec_2 \oplus \Val)\\ 
\frac{\exp(\Val/\beta)}{\exp(\Val/\beta)+E_1} + \frac{1}{1+E_2}  > & \frac{\exp(\Val/\beta)}{\exp(\Val/\beta)+E_2} + \frac{1}{1+E_1} \\
E_2>& E_1
\end{align*}
On the other hand, if  $\Val > \beta \cd \log(E_1 \cd E_2)$, then the internal zero has a \emph{negative} second derivative, and thus is a \emph{maximum}: in this case, the function $W$ is maximized at this internal point. 
\end{proof}
\optwinb*
\begin{proof}
This proof comes almost directly from the proof of Lemma \ref{lem:reallocatemoreb}. Note that Lemma \ref{lem:reallocatemoreb} tells us that if, for any \emph{pair} of tasks $1, 2$, if the total value the new agent $c$ has between those two tasks is at least 
$$\Val\geq \beta \cd \p{\log(\sum_{i \in \abs{\valvec_1}} \exp(\valvec_{i1}/\beta)) + \log(\sum_{i \in \abs{\valvec_2}} \exp(\valvec_{i2}/\beta))} = \beta \cd \p{\log(E_1) + \log(E_2)}$$
where we have used $E_i =\sum_{j \in \abs{\valvec}} \exp(\valvec_{j,i}/\beta))$ for conciseness. If this is satisfied, then the optimal allocation \emph{across those two tasks} is for agent $C$ to equalize winrate (with winrate $>0.5$) across both of them. \\

Next, we will show that if agent $C$'s total value across all tasks is at least $\Val \geq 2 \cd \max_{i \in [\ntasks]}\beta \cd \ntasks \log(E_i)$, any allocation besides one with equal winrate across tasks must have \emph{lower} total winrate.

First, we note that if for all tasks $i$, we have $\val_i > \beta \cd \log(E_i) = \beta \cd \log(\sum_{j \in [\nagents]} \exp(\valvec_{j,i}/\beta))$, then the precondition of Lemma \ref{lem:reallocatemoreb} is satisfied, and the optimal allocation is always to equalize winrate. 

Next, we will consider the case where $\val_k < \beta \cd \log(E_k) = \beta \cd \log(\sum_{i \in [\nagents]} \exp(\valvec_{k,i}/\beta)$ for some task $k$, and show by the pigeonhole principle that there must exist some other task $j$ such that $\val_k + \val_j > \beta \cd \p{\log(E_k)+ \log(E_j)}$, and again by Lemma \ref{lem:reallocatemoreb} it would be optimal to equalize winrate across tasks. Because we know that the total value across all tasks is $\Val \geq 2 \cd \max_{i \in [\ntasks]}\beta \cd \ntasks \log(E_i)$, if on a particular task $k$ we have $\val_k < \beta \cd \log(E_k)$, we must have that the \emph{average} value across all other tasks is lower bounded by: 

\begin{align*}
    \frac{\Val - \val_k}{\ntasks-1}  \geq & \frac{2 \cd \beta \cd \ntasks \cd \max_{i \in [\ntasks]} \log(E_i) - \beta \cd \log(E_k)}{\ntasks-1}\\
      > &  \frac{2 \cd \beta \cd (\ntasks-1) \cd \max_{i \in  [\ntasks]} \log(E_i)}{\ntasks-1}\\
     = &2 \cd  \beta \cd \max_{i [\ntasks]} \log(E_i)\\
\end{align*}
By this reasoning, because the \emph{average} value across all other tasks is at least $ 2 \cd  \beta \cd \max_{i [\ntasks]} \log(E_i)$, there must exist at least one task $j$ with value

$$\val_j \geq 2 \cd  \beta \cd \max_{i [\ntasks]} \log(E_i) $$
Which implies that: 
$$\val_j + \val_k \geq 2 \cd  \beta \cd \max_{i [\ntasks]} \log(E_i) \geq \beta \cd \p{\log(E_j) +\log(E_k) } $$
If this task exists, then by the preconditions of Lemma \ref{lem:reallocatemoreb}, total winrate could be increased by equalizing winrate across both task $j$ and $k$. 
\end{proof}

\optcombb*
\begin{proof}
In order to prove this, we will again consider any arrangement $\{\valvec_{\task}\}$ with positive value across multiple tasks and show that the alternative arrangement $\{\valvec_{\task}'\}$ with value on only a single task strictly improves the $\sum_{\task \in [\ntasks]}\val_{c\task} \cd \ppick_c(\valvec_{\task})$ objective. Specifically, we will assume that the allocation has positive value only on the task with the smallest current total logexp sum: 
$$\text{argmin}_{\task \in [\ntasks]} \sum_{i \in [\nagents]} \exp(\valvec_{i, \task}/\beta) = \text{argmin}_{\task \in [\ntasks]} E_\task$$
WLOG, call this task 1. The condition we wish to show is, for new agent $c$: 
\begin{align*}
    \sum_{\task \in [\ntasks]} \valvec_{c, \task}' \cd \ppick_c(\valvec_{\task}') > &  \sum_{\task \in [\ntasks]} \valvec_{c, \task} \cd \ppick_c(\valvec_{\task})\\
    \Val_c \cd \ppick_c(\valvec_{1}') + 0 \cd \sum_{\task \ne 1} 0 \cd \ppick_c(\valvec_{\task}') > &  \sum_{\task \in [\ntasks]} \valvec_{c, \task} \cd \ppick_c(\valvec_{\task})\\
    \sum_{\task \in [\ntasks]} \valvec_{c, \task} \cd \ppick_c(\valvec_{1}') > &  \sum_{\task \in [\ntasks]} \valvec_{c, \task} \cd \ppick_c(\valvec_{\task}) \tag{Rewriting $\Val_c$}\\
        \sum_{\task \in [\ntasks]} \valvec_{c, \task} \cd (\ppick_c(\valvec_{1}')- \ppick_c(\valvec_{\task})) > & 0
\end{align*}
This last inequality holds because if $\ppick_c(\valvec_{1}')- \ppick_c(\valvec_{\task})$ for any $\task \in [\ntasks]$ by design: 
\begin{align*}
    \ppick_c(\valvec_{1}')\geq & \ppick_c(\valvec_{\task})\\
    \frac{\exp(\Val/\beta)}{\exp(\Val/\beta) + E_1} \geq & \frac{\exp(\val_{c\task}/\beta)}{\exp(\val_{c\task}/\beta) + E_{\task}} \tag{$E_{\task} = \sum_{i \ne c, i \in [\nagents]} \exp(\valvec_{i\task}/\beta)$}\\
    \exp((\Val + \val_{c\task})/\beta) + \exp(\Val/\beta)\cd E_{\task} \geq & \exp((\Val + \val_{c\task})/\beta) +  E_1\cd \exp(\val_{c\task}/\beta)\\
\exp(\Val/\beta)\cd E_{\task} \geq &  E_1\cd \exp(\val_{c\task}/\beta)\\
\end{align*}
This holds by a pair of inequalities: we know that $\exp(\Val_c/\beta) \geq \exp(\valvec_{c,\task}/\beta)$ because $\Val_c \geq \valvec_{c,\task}$. Additionally, we know that $E_\task \geq E_1$ by the definition of task 1.  
\end{proof}

\begin{restatable}{lemma}{optwinbound}\label{lem:optwinbound}
Consider a set of tasks $\ntasks\geq 1$ with constant value $\{\valvec\}$ across tasks. Then, whenever: 
$$\Val \leq \ntasks \cd \beta \cd \log\p{\sum_{i \in \abs{\valvec}} \exp(\valvec_{i}/\beta)}$$
equalizing winrate across tasks does \emph{not} maximize total winrate. 
\end{restatable}
\begin{proof}
Consider the allocation that equalizes winrate: because each task is identical, this must also involve equalizing value across tasks. Because
$$\Val < \ntasks \cd \beta \cd \log\p{\sum_{i \in \abs{\valvec}} \exp(\valvec_{i}/\beta)} =  \ntasks \cd \beta \cd\p{\log(E)}$$
we know that this allocation must have $\val_i = \Val/\ntasks < \beta \cd \log(E)$. By Lemma \ref{lem:reallocatemoreb}, we know that this means that for any two tasks, total winrate can be improved by placing all value on a single task, and none on another. This means that equalizing winrate across tasks cannot be the allocation that maximizes total winrate. 
\end{proof}

\weightedwinratebetter*
\begin{proof}
By Lemma \ref{lem:addnew}, the change in consumer welfare when a new agent with value $\val'$ joins is given by: 
$$\agg{\valvec \oplus \val'} - \agg{\valvec} =   \ppick_{\nagents}(\valvec \oplus \val') \cd (\valvec' - \agg{\valvec})$$
where $\valvec$ is the vector of values for the task before a new agent joins, and $\ppick_{\nagents}(\cd)$ gives the probability of the new ($\nagents$th) agent's response being picked. 

Using the results  of Theorem \ref{thrm:optwinb}, agents best-responding to winrate would spread their value $\Val$ across tasks so as to equalize winrate, while by Lemma \ref{lem:optcombb}, for weighted winrate they would put all value on a single task, and value 0 elsewhere. 

Denote $\valvec_i \oplus 0, \valvec_i \oplus \Val, \valvec_i \oplus \val_i$ as the allocations where the new agent allocates 0, $\Val, \val_i$ respectively. 

The change in utility for winrate is: 
\begin{align*}
    \sum_{i \in [\ntasks]} \ppick_{\nagents}(\valvec_i \oplus \val_i) \cd \p{\val_i - \agg{\valvec_i}} = \sum_{i \in [\ntasks]} \ppick_{\nagents}^* \cd \p{\val_i - \agg{\valvec_i}}
\end{align*}
where in the second step we have used that $\ppick(\cd)$ is constant for all tasks with winrate best response. 
The change in utility for weighted winrate is: 
$$\ppick_{\nagents}(\valvec_1 \oplus \Val) \cd \p{\Val - \agg{\valvec_1}} + \sum_{i \in [\ntasks], i \ne 1}\ppick_{\nagents}(\valvec_i \oplus 0) \cd \p{0 - \agg{\valvec_i}}$$
where without loss of generality we have assumed that the task that weighted winrate places all weight on is task 1. Recall from Lemma \ref{lem:optcombb} that this task must be the task that maximizes $\ppick_{\nagents}(\valvec_1 \oplus \Val)$. 
Weighted winrate is better whenever: 
\begin{align*}
    \ppick_{\nagents}(\valvec_1 \oplus \Val) \cd \p{\Val - \agg{\valvec_1}} + \sum_{i \in [\ntasks], i \ne 1}\ppick_{\nagents}(\valvec_i \oplus 0) \cd \p{0 - \agg{\valvec_i}} \geq &  \sum_{i \in [\ntasks]} \ppick_{\nagents}^* \cd \p{\val_i - \agg{\valvec_i}}\\
    \ppick_{\nagents}(\valvec_1 \oplus \Val) \cd \p{\Val - \agg{\valvec_1}} - \sum_{i \in [\ntasks], i \ne 1}\ppick_{\nagents}(\valvec_i \oplus 0) \cd \agg{\valvec_i} \geq &  \sum_{i \in [\ntasks]} \ppick_{\nagents}^* \cd \p{\val_i - \agg{\valvec_i}}\\
    \ppick_{\nagents}(\valvec_1 \oplus \Val) \cd \p{\Val - \agg{\valvec_1}}  \geq &  \sum_{i \in [\ntasks]} \ppick_{\nagents}^* \cd \p{\val_i - \agg{\valvec_i}} \\
     & + \sum_{i \in [\ntasks], i \ne 1}\ppick_{\nagents}(\valvec_i \oplus 0) \cd \agg{\valvec_i}\\
    \ppick_{\nagents}(\valvec_1 \oplus \Val) \cd \p{\Val - \agg{\valvec_1}} - \sum_{i \in [\ntasks]} \ppick_{\nagents}^* \cd \val_i \geq &  -\sum_{i \in [\ntasks]} \ppick_{\nagents}^* \cd   \agg{\valvec_i} + \sum_{i \in [\ntasks], i \ne 1}\ppick_{\nagents}(\valvec_i \oplus 0) \cd \agg{\valvec_i}\\
    \ppick_{\nagents}(\valvec_1 \oplus \Val) \cd \p{\Val - \agg{\valvec_1}} - \ppick_{\nagents}^* \cd \Val \geq &  -\sum_{i \in [\ntasks]} \ppick_{\nagents}^* \cd   \agg{\valvec_i} + \sum_{i \in [\ntasks], i \ne 1}\ppick_{\nagents}(\valvec_i \oplus 0) \cd \agg{\valvec_i}\\
         \Val \cd \p{\ppick_{\nagents}(\valvec_1 \oplus \Val) - \ppick_{\nagents}^*} \geq & \ppick_{\nagents}(\valvec_1 \oplus \Val) \cd \agg{\valvec_1}  -\sum_{i \in [\ntasks]} \ppick_{\nagents}^* \cd   \agg{\valvec_i}\\
         &+ \sum_{i \in [\ntasks], i \ne 1}\ppick_{\nagents}(\valvec_i \oplus 0) \cd \agg{\valvec_i}\\
    \Val \cd \p{\ppick_{\nagents}(\valvec_1 \oplus \Val) - \ppick_{\nagents}^*} \geq &   \agg{\valvec_1} \cd (\ppick_{\nagents}(\valvec_1 \oplus \Val) -\ppick_{\nagents}^*)\\
    & + \sum_{i \in [\ntasks], i \ne 1}(\ppick_{\nagents}(\valvec_i \oplus 0)-\ppick_{\nagents}^*) \cd \agg{\valvec_i}
\end{align*}
Note that $\ppick_{\nagents}(\valvec_i \oplus 0)\leq \ppick_{\nagents}^*$, and thus we can drop the last term from the summation on the righthand side because it is always negative, in order to upper bound the term on the righthand side. 
\begin{align*}
    \Val \cd \p{\ppick_{\nagents}(\valvec_1 \oplus \Val) - \ppick_{\nagents}^*} \geq &  \cd \agg{\valvec_1} \cd (\ppick_{\nagents}(\valvec_1 \oplus \Val) -\ppick_{\nagents}^*) \\
\end{align*}
This holds whenever $\Val \geq \agg{\valvec_1}$, which is true by the assumption of Theorem \ref{thrm:optwinb} (when agents would spread out for winrate), which holds that: 
\begin{align*}
    \Val \geq &  2 \cd \ntasks \cd \beta \cd \max_{j \in [\ntasks]}\log\p{\sum_{i \in \abs{\valvec_j}} \exp(\valvec_{i,j}/\beta)}
\end{align*}
Note that by properties of the logexpsum, we have that: 
\begin{align*}
    \log\p{\sum_{i \in \abs{\valvec_j}} \exp(\valvec_{i,j}/\beta)} \geq \max_{k \in \abs{\valvec_j}}\valvec_{k,j}/\beta
\end{align*}
Additionally, by the properties of BTL, we know that $\agg{\valvec_i} \leq \max_{k \in \abs{\valvec_i}} \valvec_{k,i}$. Thus, we know that: 
\begin{align*}
    \Val \geq & 2 \cd \ntasks \cd \beta \cd \max_{j \in [\ntasks]}\log\p{\sum_{i \in \abs{\valvec_j}} \exp(\valvec_{i,j}/\beta)}\\
    > &  2 \cd \ntasks \cd \beta \cd \max_{j \in [\ntasks]} \cd \max_{k \in \abs{\valvec_j}} \valvec_{k,j}/\beta\\
    \geq &  2 \cd \ntasks \cd \max_{j \in [\ntasks]} \agg{\valvec_j}
\end{align*}
which implies $\Val \geq \agg{\valvec_1}$, as desired. 
\end{proof}

We note that from Lemmas \ref{lem:optvalueb} and \ref{lem:optcombb}, optimizing for consumer welfare and for weighted winrate would both incentivize producers to place all value on a single task. However, it is \emph{not} the case that maximizing for weighted winrate and consumer welfare would both incentivize placing weight on the \emph{same} task: it could be that they would pick different tasks $i \ne j$. For intuition, weighted winrate incentivizes producers to focus on a task where they maximize their probability of being selected, while consumer welfare would be maximized by picking a task that maximizes the probability of being picked, multiplied by the gap in value $\Val - \agg{\val_i}$. 
Lemma \ref{lem:weightedvaluegap} provides a bound on the gap in social welfare induced by this slight misalignment in objectives. 

\begin{restatable}{lemma}{weightedvaluegap}\label{lem:weightedvaluegap}
Given a set of existing tasks with values $\{\valvec_{\task}\}$, weighted winrate and consumer welfare may \emph{not} incentivize picking the same task to place all value on. However, the suboptimality of weighted winrate \emph{on this task} is bounded by $\agg{\valvec_{i^*}} - \agg{\valvec_{j^*}}$, where $i^*$ is the index of the task that maximizes the probability $\ppick_{\nagents+1}(\valvec_{i} \oplus \Val)$, and $j^*$ is the index of the task with smallest current value $\agg{\valvec_{j^*}}$.  
\end{restatable}
\begin{proof}
WLOG suppose that task 1 is the task that maximizes consumer welfare, and task 2 is the one that maximizes weighted winrate. Note from Lemma \ref{lem:optcombb} that we must have that task 2 maximizes $\ppick_{\nagents+1}(\valvec_{i} \oplus \Val)$. The gap in consumer welfare is given by: 
\begin{align*}
    & \ppick_{\nagents+1}(\valvec_1 \oplus \Val) \cd (\Val - \agg{\valvec_1}) - \ppick_{\nagents+1}(\valvec_2 \oplus \Val) \cd (\Val - \agg{\valvec_2}) \\
    <&  \ppick_{\nagents+1}(\valvec_2 \oplus \Val) \cd (\Val - \agg{\valvec_1}) - \ppick_{\nagents+1}(\valvec_2 \oplus \Val) \cd (\Val - \agg{\valvec_2})\\
     = & \ppick_{\nagents+1}(\valvec_2 \oplus \Val) \cd \p{\Val - \agg{\valvec_1} - \Val + \agg{\valvec_2}}\\
    = & \ppick_{\nagents+1}(\valvec_2 \oplus \Val) \cd \p{\agg{\valvec_2} - \agg{\valvec_1}}\\
    \leq & \agg{\valvec_2} - \agg{\valvec_1}
\end{align*}
\end{proof}

\scalinglaw*
\begin{proof}
This can be seen almost immediately from the proofs of Lemmas \ref{lem:optvalueb} or \ref{lem:optcombb}: the final allocation is one that allocates up to the maximum value on the task $\val_i^*$ sequentially according to which task maximizes that objective. This is optimal by the same reasoning with proofs of Lemmas \ref{lem:optvalueb} and \ref{lem:optcombb}: any other allocation could be strictly improved by re-allocating value to tasks that are higher on that objective. Finally, note that for consumer welfare objective, if the \enquote{remainder} of value is less than the current aggregate value on a task $\agg{\valvec_i}$, then the optimal choice is to abstain. 
\end{proof}

\section{Proofs for Section \ref{sec:incentretrain}}\label{app:incentretrain}

\begin{definition}\label{def:instantaneoususerutility}
For model replacement, instantaneous consumer welfare is the \emph{sum} of change in consumer welfare, given the instantaneous change in all producers' values over tasks. 
That is, if $\task_j^*$ is the task that producer $j$ chooses to improve in, then instantaneous consumer welfare is given by: 
\begin{align*}
\sum_{i \in [\nagents]} \frac{\partial}{\partial \valvec_{i, \task_j^*,}}\agg{\valvec_{\task_j^*}} & = \sum_{i \in [\nagents]} \frac{\partial}{\partial \valvec_{i, \task_j^*,}}\valvec_{i, \task_j^*} \cd \ppick_i(\valvec_{\task_j^*})
\end{align*}
\end{definition}
\begin{restatable}{lemma}{derivs}\label{lem:derivs}
Given a fixed task $i$, the instantaneous change in the producer's objective given a change in $\valvec_i$ is given by: 
\begin{itemize}
    \item For winrate: $\ppick_i(\valvec) \cd (1-\ppick_i(\valvec))/\beta$. 
    \item For weighted winrate: $\ppick_i(\valvec) + \valvec_i \cd \ppick_i(\valvec) \cd (1-\ppick_i(\valvec))/\beta$. 
    \item For consumer welfare, with two agents:   $\ppick_i(\valvec) + (\valvec_i-\valvec_j) \cd \ppick_i(\valvec) \cd (1-\ppick_i(\valvec))/\beta$. 
\end{itemize}
\end{restatable}
\begin{proof}
This can be seen simply by taking the derivative of each objective. 

For winrate, 
\begin{align*}
    \frac{\partial}{\partial \valvec_i} \ppick_i(\valvec)  = & \ppick_i(\valvec) \cd (1-\ppick_i(\valvec))/\beta 
\end{align*}
For weighted winrate,  
\begin{align*}
    \frac{\partial}{\partial \valvec_i} \valvec_i \cd \ppick_i(\valvec)  = & 
    \ppick_i(\valvec) + \valvec_i \cd \ppick_i(\valvec) \cd (1-\ppick_i(\valvec))/\beta 
\end{align*}
For consumer welfare, 
\begin{align*}
    \frac{\partial}{\partial \valvec_i} \agg{\valvec}  = & \frac{\partial}{\partial \valvec_i} \sum_{j \in \abs{\valvec}} \valvec_j \cd \ppick_j(\valvec)\\
    = & \frac{\partial}{\partial \valvec_i} \valvec_i \cd \ppick_i(\valvec) + \frac{\partial}{\partial \valvec_i}  \sum_{j \ne i}\valvec_j \cd \ppick_j(\valvec) \\
     = & \ppick_i(\valvec) + \valvec_i \cd \ppick_i(\valvec) \cd (1-\ppick_i(\valvec))/\beta - \valvec_j \cd \ppick_j(\valvec) \cd \ppick_i(\valvec)/\beta \tag{Properties of BTL}
\end{align*}
For the case with $\nagents=2$, the above term reduces to: 
\begin{align*}
    \ppick_i(\valvec) + \valvec_i \cd \ppick_i(\valvec) \cd (1-\ppick_i(\valvec))/\beta - \valvec_j \cd (1-\ppick_i(\valvec)) \cd \ppick_i(\valvec)/\beta = \ppick_i(\valvec) + (\valvec_i-\valvec_j) \cd \ppick_i(\valvec) \cd (1-\ppick_i(\valvec))/\beta
\end{align*}
as desired. 
\end{proof}

\winvaluepick*
\begin{proof}
This is proved through Lemmas \ref{lem:winpick} and \ref{lem:valuepick}. 
\end{proof}

\begin{restatable}{lemma}{winrate}\label{lem:winpick}
For $\nagents=2$ agents optimizing for winrate,  agents would always pick the same task to increase value in. 
\end{restatable}
\begin{proof}
If there are only two agents, then they split the winrate, so $\ppick_1(\valvec_{\task}) = 1-\ppick_2(\valvec_{\task})$. 
This immediately implies that $\frac{\partial}{\partial \val_1} \ppick_1(\valvec_{\task}) = \frac{\partial}{\partial \val_2}\ppick_2(\valvec_{\task})$, which means that both producers have the exact same preferences over all tasks to improve in. 
\end{proof}

\begin{restatable}{lemma}{valuepick}\label{lem:valuepick}
For $\nagents=2$ agents, if they are optimizing for consumer welfare, then they would \emph{never} prefer to pick the same task.   
\end{restatable}
\begin{proof}
WLOG, suppose that agent $A$ would increase value in task 1. This means that: 
\begin{align*}
    \frac{\partial}{\partial \valvec_{a,1}} \agg{\valvec_1} \geq &  \frac{\partial}{\partial \valvec_{a,2}} \agg{\valvec_2}\\
    \frac{\partial}{\partial \valvec_{a,1}} \valvec_{a,1} \cd \ppick_a(\valvec_1) + \valvec_{b,1} \cd (1-\ppick_a(\valvec_1) ) \geq  & \frac{\partial}{\partial \valvec_{a,2}} \valvec_{a,2} \cd \ppick_a(\valvec_2) + \valvec_{b,2} \cd (1-\ppick_a(\valvec_2) )\\
    \ppick_a(\valvec_1) + (\valvec_{a,1} - \valvec_{b,1}) \cd \ppick_a(\valvec_1) \cd (1-\ppick_a(\valvec_1))/\beta \geq & \ppick_a(\valvec_2) + (\valvec_{a,2} - \valvec_{b,2}) \cd \ppick_a(\valvec_2) \cd (1-\ppick_a(\valvec_2))/\beta
\end{align*}
Note by symmetry, the condition where agent $B$ would prefer to increase value in task 2 is given by: 
\begin{align*}
        \ppick_b(\valvec_1) + (\valvec_{b,1} - \valvec_{a,1}) \cd \ppick_b(\valvec_1) \cd (1-\ppick_b(\valvec_1))/\beta \leq & \ppick_b(\valvec_2) + (\valvec_{b,2} - \valvec_{a,2}) \cd \ppick_b(\valvec_2) \cd (1-\ppick_b(\valvec_2))/\beta\\
        1-\ppick_a(\valvec_1) + (\valvec_{b,1} - \valvec_{a,1}) \cd \ppick_a(\valvec_1) \cd (1-\ppick_a(\valvec_1))/\beta \leq & 1-\ppick_a(\valvec_2) + (\valvec_{b,2} - \valvec_{a,2}) \cd \ppick_a(\valvec_2) \cd (1-\ppick_a(\valvec_2))/\beta \tag{$\ppick_a(\cd) = 1-\ppick_b(\cd)$}\\
        -\ppick_a(\valvec_1) - (\valvec_{a,1} - \valvec_{b,1}) \cd \ppick_a(\valvec_1) \cd (1-\ppick_a(\valvec_1))/\beta \leq & -\ppick_a(\valvec_2) - (\valvec_{a,2} - \valvec_{b,2}) \cd \ppick_a(\valvec_2) \cd (1-\ppick_a(\valvec_2))/\beta \\
        \ppick_a(\valvec_1) + (\valvec_{a,1} - \valvec_{b,1}) \cd \ppick_a(\valvec_1) \cd (1-\ppick_a(\valvec_1))/\beta \geq &\ppick_a(\valvec_2) + (\valvec_{a,2} - \valvec_{b,2}) \cd \ppick_a(\valvec_2) \cd (1-\ppick_a(\valvec_2))/\beta
\end{align*}
which is exactly equal to the condition when player $A$ would prefer task 1 over 2. \\
This tells us that for any pair of tasks $i, j$, agents $A$ and $B$ would have exactly inverted preferences over which task they would prefer to improve in. Applied to all $\ntasks$ tasks, this says that agents $A$ and $B$ would have exactly inverted orders over which all tasks they would choose to improve in. 
\end{proof}

\valueincent*
\begin{proof}
Given any two tasks $1, 2$, the condition where player $A$ prefers to increase value in task $1$ occurs exactly when: 
\begin{align*}
        \frac{\partial}{\partial \valvec_{a,1}} \valvec_{a,1} \cd \ppick_a(\valvec_1) \geq &  \frac{\partial}{\partial \valvec_{a,2}} \valvec_{a,2} \cd \ppick_a(\valvec_2)\\
        \ppick_a(\valvec_1) + \valvec_{a,1} \cd \ppick_a(\valvec_1) \cd (1-\ppick_a(\valvec_1)/\beta \geq & \ppick_a(\valvec_2) + \valvec_{a,2} \cd \ppick_a(\valvec_2) \cd (1-\ppick_a(\valvec_2)/\beta\\
        \ppick_a(\valvec_1) - \ppick_a(\valvec_2) \geq &  \p{\valvec_{a,2} \cd \ppick_a(\valvec_2) \cd (1-\ppick_a(\valvec_2)) - \valvec_{a,1} \cd \ppick_a(\valvec_1) \cd (1-\ppick_a(\valvec_1)}/\beta \vcentcolon=A_{12}
\end{align*}

Note by identical reasoning the condition where player $B$ prefers to increase value in task $1$ is given by: 
\begin{align*}
        \ppick_b(\valvec_1) - \ppick_b(\valvec_2) \geq &  \p{\valvec_{b,2} \cd \ppick_b(\valvec_2) \cd (1-\ppick_b(\valvec_2)) - \valvec_{b,1} \cd \ppick_b(\valvec_1) \cd (1-\ppick_b(\valvec_1)}/\beta\\
        -\ppick_a(\valvec_1) + \ppick_a(\valvec_2) \geq &  \p{\valvec_{b,2} \cd \ppick_a(\valvec_2) \cd (1-\ppick_a(\valvec_2)) - \valvec_{b,1} \cd \ppick_a(\valvec_1) \cd (1-\ppick_a(\valvec_1)}/\beta \tag{$\ppick_a(\cd) = 1-\ppick_b(\cd)$}\\
        \ppick_a(\valvec_1) - \ppick_a(\valvec_2) \leq &  -\p{\valvec_{b,2} \cd \ppick_a(\valvec_2) \cd (1-\ppick_a(\valvec_2)) - \valvec_{b,1} \cd \ppick_a(\valvec_1) \cd (1-\ppick_a(\valvec_1)}/\beta \vcentcolon= B_{12}
\end{align*}
Thus, the condition where player $A$ and player $B$ both prefer task $1$ to $2$ is when: 
$$A_{12}<\abs{\ppick_a(\valvec_1) - \ppick_a(\valvec_2)} = \abs{\ppick_b(\valvec_1) - \ppick_b(\valvec_2)} < B_{12}$$
The condition where player $A$ and player $B$ both prefer task $2$ to $1$ is when: 
$$B_{12}<\abs{\ppick_a(\valvec_1) - \ppick_a(\valvec_2)} = \abs{\ppick_b(\valvec_1) - \ppick_b(\valvec_2)} < A_{12}$$
And the condition where they must have different preferences is when: 
$$\max(A_{12}, B_{12})<\abs{\ppick_a(\valvec_1) - \ppick_a(\valvec_2)} = \abs{\ppick_b(\valvec_1) - \ppick_b(\valvec_2)}$$
Finally, note that the scenario where players both pick the same task (e.g. task 1), occurs when: 
\begin{align*}
            \ppick_a(\valvec_1) + \valvec_{a,1} \cd \ppick_a(\valvec_1) \cd (1-\ppick_a(\valvec_1)/\beta \geq & \ppick_a(\valvec_2) + \valvec_{a,2} \cd \ppick_a(\valvec_2) \cd (1-\ppick_a(\valvec_2)/\beta
\end{align*}
and: 
\begin{align*}
            \ppick_b(\valvec_2) + \valvec_{b,2} \cd \ppick_b(\valvec_2) \cd (1-\ppick_b(\valvec_2)/\beta \leq & \ppick_b(\valvec_1) + \valvec_{b,1} \cd \ppick_b(\valvec_1) \cd (1-\ppick_b(\valvec_1)/\beta
\end{align*}
and implies: 
\begin{align*}
            (\valvec_{a,1}  + \valvec_{b,1})\cd \ppick_a(\valvec_1) \cd (1-\ppick_a(\valvec_1)/\beta \geq &  (\valvec_{a,2}+ \valvec_{b,2}) \cd \ppick_a(\valvec_2) \cd (1-\ppick_a(\valvec_2)/\beta\\
\end{align*}
\end{proof}

This table shows an example where weighted winrate does lead to increased specialization: winrate would incentivize players to improve on the same task and weighted winrate would incentivize producers to improve on different tasks. However, this is the \emph{wrong} type of specialization: consumer welfare would be higher if each producer swapped which task they improved on, and in fact, would be higher if producers picked the same task to improve on, as given by winrate. 

\begin{table}[h]
    \begin{subtable}[h]{0.45\textwidth}
        \centering
\begin{tabular}{|c|c|c|}
\hline
\textbf{} & 1         & 2         \\ \hline
A        & 74.3 & 16.45 \\ \hline
B        & 70.4 & 14.34 \\ \hline
\end{tabular}
       \caption{Value: player $A$ has higher value on both tasks.}
       \label{tab:value}
     \end{subtable}
    \hfill
    \begin{subtable}[h]{0.45\textwidth}
        \centering
\begin{tabular}{|c|c|c|}
\hline
 & 1 & 2       \\ \hline
A        & 0.029  & \textbf{0.0979} \\ \hline
B        & 0.029  & \textbf{0.0979} \\ \hline
\end{tabular}
        \caption{Winrate (derivative).}
        \label{tab:derivewinrate}
     \end{subtable}
     \vspace{2ex}

         \begin{subtable}[h]{0.45\textwidth}
        \centering
\begin{tabular}{|c|c|c|}
\hline
\textbf{} & 1 & 2        \\ \hline
A        & \textbf{2.48}   & 2.47 \\ \hline
B        & 1.44  & \textbf{1.48}\\ \hline
\end{tabular}
       \caption{Weighted winrate (derivative).}
       \label{tab:deriveww}
    \end{subtable}
    \hfill
    \begin{subtable}[h]{0.45\textwidth}
        \centering
\begin{tabular}{|c|c|c|}
\hline
& 1& 2         \\ \hline
A        & 1.057   & \textbf{1.095}  \\ \hline
B        & \textbf{-0.057}  & -0.095 \\ \hline
\end{tabular}
        \caption{Consumer welfare (derivative)}
        \label{tab:derivevalue}
     \end{subtable}
     \caption{Case with two producers $A, B$ picking between two tasks $1, 2$ (maximum value in each row \textbf{bolded} for convenience). Here, note that producers $A, B$ each pick different tasks $1, 2$ according to weighted winrate, but end up resulting in worse instantaneous consumer welfare than if they had picked according to winrate. Note that player $B$ is dominated in both tasks, and in particular, on both tasks it has \emph{negative} derivative for consumer welfare. That is, no matter which task player $B$ picked, its increase in value strictly harms consumer welfare. Note all quantities are rounded. }
     \label{tab:exinst}
\end{table}

\bothneg*
\begin{proof}
WLOG, we assume agent $A$ picks task 1 and agent $B$ picks task 2. We will begin by eliminating conditions where weighted winrate improves consumer welfare, and then prove that in all remaining situations, one player must be in the non-monotonic regime on both picked tasks. 

By results from other lemmas, we know that the following situations are ones where weighted winrate strictly improves consumer welfare over winrate: 
\begin{enumerate}
    \item (By Lemma \ref{lem:changelower}): Both players get higher winrate in picking their task: $\ppick_b(\valvec_2) > \ppick_b(\valvec_1)$ and $\ppick_a(\valvec_1) > \ppick_a(\valvec_2)$. 
    \item (By Lemma \ref{lem:bothdom}): At least one agent has winrate $<0.5$ on the task that it picks. 
\end{enumerate}
Combining both facts, this must mean that it has winrate $<0.5$ on both tasks, and the other agent has winrate $>0.5$ on both tasks. WLOG, assume that it is agent $A$ that does worse. Note that this implies $\valvec_{bi} >\valvec_{ai}$ for $i=1, 2$. 
When is the case where we get \emph{worse} total sum of derivatives by using weighted winrate? 
\begin{align*}
\frac{\partial}{\partial \valvec_{a,1}} \agg{\valvec_1} + \frac{\partial}{\partial \valvec_{b,2}} \agg{\valvec_2} < &  1\\
    \ppick_a(\valvec_1) + (\valvec_{a,1} -\valvec_{b,1}) \cd\ppick_a(\valvec_1) \cd (1-\ppick_a(\valvec_1))/\beta &  \\
    +\ppick_b(\valvec_2) + (\valvec_{b,2} -\valvec_{a,2}) \cd\ppick_b(\valvec_2) \cd (1-\ppick_b(\valvec_2))/\beta & < 1\\
        \ppick_a(\valvec_1) + (\valvec_{a,1} -\valvec_{b,1}) \cd\ppick_a(\valvec_1) \cd (1-\ppick_a(\valvec_1))/\beta &  \\
   + 1-\ppick_a(\valvec_2) + (\valvec_{b,2} -\valvec_{a,2}) \cd\ppick_b(\valvec_2) \cd (1-\ppick_b(\valvec_2))/\beta & < 1\\
    \ppick_a(\valvec_1) + (\valvec_{a,1} -\valvec_{b,1})/\beta \cd\ppick_a(\valvec_1) \cd (1-\ppick_a(\valvec_1))< &    \ppick_a(\valvec_2)  +(\valvec_{a,2} -\valvec_{b,2})/\beta \cd\ppick_b(\valvec_2) \cd (1-\ppick_b(\valvec_2))\\
    \ppick(\Delta_1) + \Delta_1/\beta \cd \ppick(\Delta_1) \cd (1-\ppick(\Delta_1)) <  & \ppick(\Delta_2) + \Delta_2/\beta \cd \ppick(\Delta_2) \cd (1-\ppick(\Delta_2))\\
    f(\Delta_1) < & f(\Delta_2)
\end{align*}
We will apply results from the lemmas above to tell us more about when the condition above holds. 
Consider the function $f(x)$ as defined above. Because the winrate of agent $a$ is smaller in task 2, and player $A$ has winrate $<0.5$ in each of them, we know that $\Delta_2 < \Delta_1 < 0$. 

By Lemma \ref{lem:helperproof}, we know that the function $f(x)$ has a single point of 0 derivative for $x<0$, and the regime where $f(x)$ is decreasing in $x$ (and where $p(x) < 0.5)$ is also one where $f(x)<0$.  

Applying this here, we know that $\Delta_2 < \Delta_1$, and yet $f(\Delta_2) > f(\Delta_1)$. Because there is only a single point of 0 derivative, we know that this means that $f'(\Delta_2) < 0$, and thus $f(\Delta_2) <0$, which means that $f(\Delta_1) <0$ as well, as desired.
\end{proof}

\begin{restatable}{lemma}{same}\label{lem:same}
Given the setting of Theorem \ref{thrm:bothneg}, whenever both agents pick the same task to improve on, the change in consumer welfare is constant (equal to 1). 
\end{restatable}
\begin{proof}
Denote the players as $A$ and $B$, and WLOG label the task they both pick as $1$. Then, the instantaneous consumer welfare is given by: 
\begin{align*}
    &  \frac{\partial}{\partial \valvec_{a,1}} \agg{\valvec_1} + \frac{\partial}{\partial \valvec_{b,1}} \agg{\valvec_1}\\
     = & \ppick_a(\valvec_1) + (\valvec_{1a} -\valvec_{1b})\cd\ppick_a(\valvec_1) \cd (1-\ppick_a(\valvec_1))/\beta  \\ 
     & +\ppick_b(\valvec_1)+(\valvec_{1b} -\valvec_{1a}) \cd\ppick_b(\valvec_1) \cd (1-\ppick_b(\valvec_1))/\beta \tag{Value of derivative}\\
      = & 1-\ppick_b(\valvec_1) + (\valvec_{1a} -\valvec_{1b}) \cd\ppick_a(\valvec_1) \cd (1-\ppick_a(\valvec_1))/\beta \\
      &+\ppick_b(\valvec_1) + (\valvec_{1b} -\valvec_{1a}) \cd\ppick_a(\valvec_1) \cd (1-\ppick_a(\valvec_1))/\beta \tag{$\ppick_a(\valvec_1) = 1-\ppick_b(\valvec_1)$}\\ 
     = &  1
\end{align*}
\end{proof}

\begin{restatable}{lemma}{bothdom}\label{lem:bothdom}
Given the setting of Theorem \ref{thrm:bothneg}, if both players pick tasks where they have winrate at least 0.5, then they're doing strictly better than if they both picked the same task. 
\end{restatable}
\begin{proof}
Denote the players as $A$ and $B$, and WLOG assume that $A$ prefers task 1, and $B$ prefers task 2. The instantaneous consumer welfare is given by: 

\begin{align*}
    &  \frac{\partial}{\partial \valvec_{a,1}} \agg{\valvec_1} + \frac{\partial}{\partial \valvec_{b,2}} \agg{\valvec_2}\\
     = & \ppick_a(\valvec_1) + (\valvec_{1a} -\valvec_{1b})\cd\ppick_a(\valvec_1) \cd (1-\ppick_a(\valvec_1))/\beta  \\ 
     & +\ppick_b(\valvec_2)+(\valvec_{2b} -\valvec_{2a}) \cd\ppick_b(\valvec_2) \cd (1-\ppick_b(\valvec_2))/\beta \tag{Value of derivative}\\
     = & \ppick_a(\valvec_1) + \ppick_b(\valvec_2) \\
      + & (\valvec_{1a} -\valvec_{1b})\cd\ppick_a(\valvec_1) \cd (1-\ppick_a(\valvec_1))/\beta + \ppick_b(\valvec_2)+(\valvec_{2b} -\valvec_{2a}) \cd\ppick_b(\valvec_2) \cd (1-\ppick_b(\valvec_2))/\beta \\
      >& 1 \tag{$\ppick_a(\valvec_1),   \ppick_b(\valvec_2) >0.5$}\\
     + &(\valvec_{1a} -\valvec_{1b})\cd\ppick_a(\valvec_1) \cd (1-\ppick_a(\valvec_1))/\beta + \ppick_b(\valvec_2)+(\valvec_{2b} -\valvec_{2a}) \cd\ppick_b(\valvec_2) \cd (1-\ppick_b(\valvec_2))/\beta \\
      >& 1 \tag{$(\valvec_{1a} -\valvec_{1b})>0, (\valvec_{2b} -\valvec_{2a})>0$}
\end{align*}
as desired. 
\end{proof}

\begin{restatable}{lemma}{changelower}\label{lem:changelower}
Consider the setting of Theorem \ref{thrm:bothneg}, and assume that at least one player has \emph{lower} winrate on the task that it picks (e.g., 
    $\ppick_b(\valvec_1) > \ppick_b(\valvec_2)$, if player $B$ picks task 2). Then, instantaneous consumer welfare for this setting is at least as large as if both players picked the same task. 
\end{restatable}
\begin{proof}
First, let's write out multiple conditions. We know that players $A$ and $B$ prefer tasks 1 and 2. We also assume that at least one player has \emph{lower} winrate on the task that it prefers: WLOG, we will assume that this is player $B$ or: 
\begin{align}
    \ppick_b(\valvec_1) > \ppick_b(\valvec_2) \label{eq:4}
\end{align}
If we have that player A prefers task 1, then we know that: 
 \begin{align}
    \ppick_a(\valvec_1) + \valvec_{a,1} \cd\ppick_a(\valvec_1) \cd (1-\ppick_a(\valvec_1))/\beta > &  \ppick_a(\valvec_2) + \valvec_{a,2} \cd \ppick_a(\valvec_2) \cd (1-\ppick_a(\valvec_2) ) \nonumber \\
    \ppick_a(\valvec_1) - \ppick_a(\valvec_2) >  & \valvec_{a,2} \cd \ppick_a(\valvec_2) \cd (1-\ppick_a(\valvec_2) ) - \valvec_{a,1} \cd\ppick_a(\valvec_1) \cd (1-\ppick_a(\valvec_1))/\beta \label{eq:1}
 \end{align}
If we know that player $B$ prefers task 2, we know that:
\begin{align}
    \ppick_b(\valvec_1) + \valvec_{b,1} \cd\ppick_a(\valvec_1) \cd (1-\ppick_a(\valvec_1))/\beta < &  \ppick_b(\valvec_2) + \valvec_{b,2} \cd \ppick_a(\valvec_2) \cd (1-\ppick_a(\valvec_2) )/\beta \nonumber\\
    \ppick_b(\valvec_1)-\ppick_b(\valvec_2)  < &    \valvec_{b,2} \cd \ppick_a(\valvec_2) \cd (1-\ppick_a(\valvec_2) )/\beta-\valvec_{b,1} \cd\ppick_a(\valvec_1) \cd (1-\ppick_a(\valvec_1))/\beta\label{eq:2}
\end{align}
Note that by Equation \ref{eq:4}, we know that the LHS of Equation \ref{eq:2} is positive, and therefore the RHS is also positive. 
Next, the condition we want to prove is that instantaneous consumer welfare for having player $A$ pick task 1 and player $B$ pick task 2 is at least as large as if they both picked the same task (which gives consumer welfare 1). That is, we want to show: 
\begin{align*}
    &\ppick_a(\valvec_1) + (\valvec_{a,1} -\valvec_{b,1}) \cd\ppick_a(\valvec_1) \cd (1-\ppick_a(\valvec_1))/\beta \\
    +&\ppick_b(\valvec_2) + (\valvec_{b,2} -\valvec_{a,2}) \cd\ppick_b(\valvec_2) \cd (1-\ppick_b(\valvec_2))/\beta > 1\\
    \Leftrightarrow & \quad \ppick_a(\valvec_1) + (\valvec_{a,1} -\valvec_{b,1}) \cd\ppick_a(\valvec_1) \cd (1-\ppick_a(\valvec_1))/\beta \\
    +&1-\ppick_a(\valvec_2) + (\valvec_{b,2} -\valvec_{a,2}) \cd\ppick_a(\valvec_2) \cd (1-\ppick_a(\valvec_2))/\beta > 1 \tag{$\ppick_b(\valvec_2) = 1-\ppick_a(\valvec_2)$}\\
         \Leftrightarrow  &\quad  \ppick_a(\valvec_1)-\ppick_a(\valvec_2)   \\
    +&(\valvec_{a,1} -\valvec_{b,1}) \cd\ppick_a(\valvec_1) \cd (1-\ppick_a(\valvec_1))/\beta+  (\valvec_{b,2} -\valvec_{a,2}) \cd\ppick_a(\valvec_2) \cd (1-\ppick_a(\valvec_2))/\beta > 0 \\
             \Leftarrow  &\quad  \valvec_{a,2} \cd \ppick_a(\valvec_2) \cd (1-\ppick_a(\valvec_2) ) - \valvec_{a,1} \cd\ppick_a(\valvec_1) \cd (1-\ppick_a(\valvec_1))/\beta \tag{By Eq. \ref{eq:1}}  \\
    +&(\valvec_{a,1} -\valvec_{b,1}) \cd\ppick_a(\valvec_1) \cd (1-\ppick_a(\valvec_1))/\beta+  (\valvec_{b,2} -\valvec_{a,2}) \cd\ppick_a(\valvec_2) \cd (1-\ppick_a(\valvec_2))/\beta > 0 \\
     \Leftrightarrow  &\quad  -\valvec_{b,1} \cd\ppick_a(\valvec_1) \cd (1-\ppick_a(\valvec_1))/\beta+  \valvec_{b,2} \cd \ppick_a(\valvec_2) \cd (1-\ppick_a(\valvec_2))/\beta > 0 
\end{align*}
Note that the last condition is satisfied by Equation \ref{eq:2}: knowing that the RHS of this equation is positive is equivalent to satisfying the condition above.  
\end{proof}

\begin{lemma}\label{lem:helperproof}
Consider the function $f(x, \beta) = \ppick(x/\beta) + x/\beta\cd \ppick(x/\beta) \cd (1-\ppick(x/\beta))$, where $\ppick(x/\beta)$ gives the BTL probability of picking an item, given a gap in values of $x$ and noise level $\beta$. Then, a) $f(x, \beta)$ has a single point of 0 derivative for $x<0$, and b) whenever $f'(x, \beta) <0$ and $\ppick(x/\beta)<0.5$, we know that $f(x, \beta) <0$. 
\end{lemma}
\begin{proof}
First, we calculate the derivative of $f(x)$ (for simplicity, replacing $\ppick = \ppick(x/\beta)$): 
\begin{align*}
    f'(x) = &  \ppick \cd (1-\ppick)/\beta +  \ppick \cd (1-\ppick)/\beta + \ppick/\beta \cd (\ppick \cd (1-\ppick)^2/\beta + \ppick \cd (0-\ppick \cd (1-\ppick))/\beta)\\
    = & 2 \cd \ppick \cd (1-\ppick)/\beta + x/\beta^2 \cd \ppick \cd ((1-\ppick)^2 - \ppick \cd (1-\ppick))\\ 
     = & 2 \cd \ppick \cd (1-\ppick)/\beta  + x/\beta^2 \cd \ppick\cd (1-\ppick) \cd (1-\ppick) - \ppick )\\ 
     = & \ppick \cd (1-\ppick)/\beta \cd (2+ x/\beta \cd (1-2\ppick))
\end{align*}
This is negative if $2 + x/\beta\cd (1-2\ppick) < 0$. First, we will note that there is only a single point of zero derivative when $x<0$. This is because the derivative is zero whenever: 
\begin{align*}
    2 + x/\beta\cd (1-2\ppick)  & =0 \\
    1-2 \cd \ppick = &  -2/(x/\beta)\\ 
    -2 \cd \ppick = & -1 -2/(x/\beta)\\ 
     \ppick(x/\beta) = & 0.5 +1/(x/\beta)\\ 
\end{align*}
where in the last line, we have added in the fact that $\ppick(\cd)$ is a function of $x/\beta$. First, we note that $\ppick(x/\beta)$ is always increasing in $x$. Second we note that $1/x$ is negative for $x<0$. As $x$ increases (becomes less negative), $1/\abs{x}$ increases, and thus $1/x$ \emph{decreases}. Thus, because the LHS is \emph{increasing} in $x$ and the RHS is \emph{decreasing} in $x$, they can only have a single point where they cross for $x<0$. 

Next, we want to show that if $f'(x)<0$ and $x<0$, then $f(x) < 0$. Note that $f(x)<0$ if: 
\begin{align*}
    \ppick + x/\beta \cd \ppick(1-\ppick) < &  0\\
    1 + x/\beta\cd (1-\ppick) <& 0
\end{align*}
Note that we have: 
$$0 > 2 + x /\beta\cd (1-2\ppick) > 1 + x/\beta \cd (1-2\ppick) > 1 + x/\beta \cd (1-\ppick)$$
where the last line is because $0 < 1-2 \cd \ppick < 1-\ppick$, because $\ppick <0.5$ by assumption. 
\end{proof}


\newpage
\ifarxiv
\else 
\section*{NeurIPS Paper Checklist}

The checklist is designed to encourage best practices for responsible machine learning research, addressing issues of reproducibility, transparency, research ethics, and societal impact. Do not remove the checklist: {\bf The papers not including the checklist will be desk rejected.} The checklist should follow the references and follow the (optional) supplemental material.  The checklist does NOT count towards the page
limit. 

Please read the checklist guidelines carefully for information on how to answer these questions. For each question in the checklist:
\begin{itemize}
    \item You should answer \answerYes{}, \answerNo{}, or \answerNA{}.
    \item \answerNA{} means either that the question is Not Applicable for that particular paper or the relevant information is Not Available.
    \item Please provide a short (1--2 sentence) justification right after your answer (even for \answerNA). 
\end{itemize}

{\bf The checklist answers are an integral part of your paper submission.} They are visible to the reviewers, area chairs, senior area chairs, and ethics reviewers. You will also be asked to include it (after eventual revisions) with the final version of your paper, and its final version will be published with the paper.

The reviewers of your paper will be asked to use the checklist as one of the factors in their evaluation. While \answerYes{} is generally preferable to \answerNo{}, it is perfectly acceptable to answer \answerNo{} provided a proper justification is given (e.g., error bars are not reported because it would be too computationally expensive'' or ``we were unable to find the license for the dataset we used''). In general, answering \answerNo{} or \answerNA{} is not grounds for rejection. While the questions are phrased in a binary way, we acknowledge that the true answer is often more nuanced, so please just use your best judgment and write a justification to elaborate. All supporting evidence can appear either in the main paper or the supplemental material, provided in appendix. If you answer \answerYes{} to a question, in the justification please point to the section(s) where related material for the question can be found.

IMPORTANT, please:
\begin{itemize}
    \item {\bf Delete this instruction block, but keep the section heading ``NeurIPS Paper Checklist"},
    \item  {\bf Keep the checklist subsection headings, questions/answers and guidelines below.}
    \item {\bf Do not modify the questions and only use the provided macros for your answers}.
\end{itemize}


\begin{enumerate}

\item {\bf Claims}
    \item[] Question: Do the main claims made in the abstract and introduction accurately reflect the paper's contributions and scope?
    \item[] Answer: \answerYes{} 
    \item[] Justification: Our main claim is that winrate can incentivize homogenization, and weighted winrate improves incentives: these are proven in the relevant sections. 

\item {\bf Limitations}
    \item[] Question: Does the paper discuss the limitations of the work performed by the authors?
    \item[] Answer: \answerYes{} 
    \item[] Justification: These are discussed in Section \ref{sec:model}, which describes our model and key assumptions. We also discuss generalizations in Section \ref{sec:conclude}. 

\item {\bf Theory assumptions and proofs}
    \item[] Question: For each theoretical result, does the paper provide the full set of assumptions and a complete (and correct) proof?
    \item[] Answer: \answerYes{} 
    \item[] Justification: These are discussed in Section \ref{sec:model}, which describes our model and key assumptions. All proofs are in the appendix. 

    \item {\bf Experimental result reproducibility}
    \item[] Question: Does the paper fully disclose all the information needed to reproduce the main experimental results of the paper to the extent that it affects the main claims and/or conclusions of the paper (regardless of whether the code and data are provided or not)?
    \item[] Answer: \answerYes{} 
    \item[] Justification: The appendix describes each experiment in detail. We either provide data directly or link to the data source. In the final paper we will provide all code to reproduce experiments. 

\item {\bf Open access to data and code}
    \item[] Question: Does the paper provide open access to the data and code, with sufficient instructions to faithfully reproduce the main experimental results, as described in supplemental material?
    \item[] Answer: \answerNo{} 
    \item[] Justification: Code will be released upon a final version. 

\item {\bf Experimental setting/details}
    \item[] Question: Does the paper specify all the training and test details (e.g., data splits, hyperparameters, how they were chosen, type of optimizer) necessary to understand the results?
    \item[] Answer: \answerYes{} 
    \item[] Justification: These are only relevant for Appendix \ref{app:errorpatternsrouter}, where we're reproducing results in an existing paper. 

\item {\bf Experiment statistical significance}
    \item[] Question: Does the paper report error bars suitably and correctly defined or other appropriate information about the statistical significance of the experiments?
    \item[] Answer: \answerNo{} 
    \item[] Justification: No results involve or require error bars. 

\item {\bf Experiments compute resources}
    \item[] Question: For each experiment, does the paper provide sufficient information on the computer resources (type of compute workers, memory, time of execution) needed to reproduce the experiments?
    \item[] Answer: \answerNo{} 
    \item[] Justification: All experiments can be done easily on a laptop. 
    
\item {\bf Code of ethics}
    \item[] Question: Does the research conducted in the paper conform, in every respect, with the NeurIPS Code of Ethics \url{https://neurips.cc/public/EthicsGuidelines}?
    \item[] Answer: \answerYes{} 
    \item[] Justification: There are no human experiments or novel data, and our core goal of our research is to improve consumer welfare. 

\item {\bf Broader impacts}
    \item[] Question: Does the paper discuss both potential positive societal impacts and negative societal impacts of the work performed?
    \item[] Answer: \answerYes{} 
    \item[] Justification: Discussed in the introduction, conclusion, and model section. 
    
\item {\bf Safeguards}
    \item[] Question: Does the paper describe safeguards that have been put in place for responsible release of data or models that have a high risk for misuse (e.g., pre-trained language models, image generators, or scraped datasets)?
    \item[] Answer: \answerNA{} 
    \item[] Justification: Not relevant for given experiments. 

\item {\bf Licenses for existing assets}
    \item[] Question: Are the creators or original owners of assets (e.g., code, data, models), used in the paper, properly credited and are the license and terms of use explicitly mentioned and properly respected?
    \item[] Answer: \answerYes{} 
    \item[] Justification: All data and results are cited when appropriate. 

\item {\bf New assets}
    \item[] Question: Are new assets introduced in the paper well documented and is the documentation provided alongside the assets?
    \item[] Answer: \answerNA{} 
    \item[] Justification: No new assets are created. 

\item {\bf Crowdsourcing and research with human subjects}
    \item[] Question: For crowdsourcing experiments and research with human subjects, does the paper include the full text of instructions given to participants and screenshots, if applicable, as well as details about compensation (if any)? 
    \item[] Answer: \answerNA{} 
    \item[] Justification: No human experiments. 

\item {\bf Institutional review board (IRB) approvals or equivalent for research with human subjects}
    \item[] Question: Does the paper describe potential risks incurred by study participants, whether such risks were disclosed to the subjects, and whether Institutional Review Board (IRB) approvals (or an equivalent approval/review based on the requirements of your country or institution) were obtained?
    \item[] Answer: \answerNA{} 
    \item[] Justification: No human experiments. 

\item {\bf Declaration of LLM usage}
    \item[] Question: Does the paper describe the usage of LLMs if it is an important, original, or non-standard component of the core methods in this research? Note that if the LLM is used only for writing, editing, or formatting purposes and does \emph{not} impact the core methodology, scientific rigor, or originality of the research, declaration is not required.
    \item[] Answer: \answerNA{}{} 
    \item[] Justification: LLMs only used for minor editing and literature review. 

\end{enumerate}
\fi

\end{document}